\documentclass{article}
\usepackage{arxiv}

\usepackage{tikz}
\usetikzlibrary{shapes,backgrounds,patterns,calc}
\usepackage{booktabs} 
\usepackage{hyperref}

\usepackage{pict2e,picture,graphicx}
\usepackage[overload]{empheq}
\usepackage{microtype}
\usepackage{graphicx}
\usepackage{subfigure}
\usepackage{booktabs}
\usepackage{amsmath}
\usepackage{cases}
\usepackage{mathtools}
\usepackage{amsthm}
\usepackage{thm-restate}
\usepackage{enumitem}
\usepackage{xcolor}
\usepackage{nicefrac, xfrac}
\usepackage[capitalize,noabbrev]{cleveref}
\usepackage{xparse}
\usepackage{hyperref}
\usepackage{bbm}
\usepackage{yfonts}
\usepackage{subfigure}
\usepackage{nicefrac} 
\usepackage{multirow}
\usepackage{bm}
\usepackage{bbm}
\usepackage{xspace}
\usepackage{amsfonts}
\usepackage{wrapfig}
\usepackage{algorithm}
\usepackage{algorithmic}

\tikzset{
	hyperlink node/.style={
		alias=sourcenode,
		append after command={
			let \p1 = (sourcenode.south west),
			\p2 = (sourcenode.south east),
			\p3 = (sourcenode.north east),
			\p4 = (sourcenode.north west),
			\p5 = (sourcenode.center)
			in
			node [inner sep=0pt, outer sep=0pt,at=(\p1)] {\pdfsavepos%
				\writeAux{}{%
					\string\expandafter\string\xdef\string\csname\space srcnd\thesrcnd.x1\string\endcsname{%
						\noexpand\hypercalcbp{\noexpand\number\pdflastxpos sp}}%
					\string\expandafter\string\xdef\string\csname\space srcnd\thesrcnd.y1\string\endcsname{%
						\noexpand\hypercalcbp{\noexpand\number\pdflastypos sp}}}%
			}
			node [inner sep=0pt, outer sep=0pt,at=(\p2)] {\pdfsavepos%
				\writeAux{}{%
					\string\expandafter\string\xdef\string\csname\space srcnd\thesrcnd.x2\string\endcsname{%
						\noexpand\hypercalcbp{\noexpand\number\pdflastxpos sp}}%
					\string\expandafter\string\xdef\string\csname\space srcnd\thesrcnd.y2\string\endcsname{%
						\noexpand\hypercalcbp{\noexpand\number\pdflastypos sp}}}%
			}
			node [inner sep=0pt, outer sep=0pt,at=(\p3)] {\pdfsavepos%
				\writeAux{}{%
					\string\expandafter\string\xdef\string\csname\space srcnd\thesrcnd.x3\string\endcsname{%
						\noexpand\hypercalcbp{\noexpand\number\pdflastxpos sp}}%
					\string\expandafter\string\xdef\string\csname\space srcnd\thesrcnd.y3\string\endcsname{%
						\noexpand\hypercalcbp{\noexpand\number\pdflastypos sp}}}%
			}
			node [inner sep=0pt, outer sep=0pt,at=(\p4)] {\pdfsavepos%
				\writeAux{}{%
					\string\expandafter\string\xdef\string\csname\space srcnd\thesrcnd.x4\string\endcsname{%
						\noexpand\hypercalcbp{\noexpand\number\pdflastxpos sp}}%
					\string\expandafter\string\xdef\string\csname\space srcnd\thesrcnd.y4\string\endcsname{%
						\noexpand\hypercalcbp{\noexpand\number\pdflastypos sp}}}%
			}
			node [inner sep=0pt, outer sep=0pt,at=(\p5)] {%
				\makebox[0pt][c]{{%
						\edef\QuadPoints{/QuadPoints [
							\csname srcnd\thesrcnd.x1\endcsname\space\csname srcnd\thesrcnd.y1\endcsname\space
							\csname srcnd\thesrcnd.x2\endcsname\space\csname srcnd\thesrcnd.y2\endcsname\space
							\csname srcnd\thesrcnd.x3\endcsname\space\csname srcnd\thesrcnd.y3\endcsname\space
							\csname srcnd\thesrcnd.x4\endcsname\space\csname srcnd\thesrcnd.y4\endcsname\space
							]}%
						\hyperlink{#1}{%
							\raisebox{-\height}{%
								\tikz \useasboundingbox (\p1) -- (\p2) -- (\p3) -- (\p4) -- cycle;
							}%
						}%
				}}%
				\stepcounter{srcnd}%
			}
		}
	}
}

\newcommand{\cC} {{\mathcal{C}}}
\newcommand{\cD} {{\mathcal{D}}}

\newcommand{\cP} {{\mathcal{P}}}

\newcommand{\cX} {{\mathcal{X}}}

\renewcommand{\vec}[1]{\bm{#1}}

\newcommand{\vc}{\vec{c}}

\newcommand{\vp}{\vec{p}}

\newcommand{\vgamma}{\vec{\gamma}}

\newcommand{\cvec}{\boldsymbol{c}}

\newcommand{\Fvec}{\boldsymbol{F}}
\newcommand{\rvec}{\boldsymbol{r}}
\newcommand{\pvec}{\boldsymbol{p}}
\newcommand{\gvec}{\boldsymbol{\gamma}}
\newcommand{\supp}{\textnormal{supp}}
\newcommand{\mcP}{\mathcal{P}}
\newcommand{\mcPr}{\mathcal{P}^\textsc{R}}
\newcommand{\mcM}{\mathcal{M}}
\newcommand{\mcMr}{\mathcal{M}^\textsc{R}}
\newcommand{\mcC}{\mathcal{C}}
\newcommand{\mcR}{\mathcal{R}}
\newcommand{\Ua}{U^{\mathfrak{a}}}
\newcommand{\Up}{U^{\mathfrak{p}}}
\newcommand{\term}[1]{\ensuremath{\textnormal{\texttt{#1}}}\xspace}
\newcommand{\OPT}{{\term{OPT}}}

\newcommand{\reg}{\mathcal{R}}

\newcommand{\cY} {\ensuremath{\mathcal{Y}}}

\usepackage{color}
\definecolor{mygreen}{rgb}{0.0, 0.5, 0.0}
\definecolor{myorange}{rgb}{0.55, 0.62, 1}

\newtheorem{observation}{Observation}

\makeatletter
\newcommand{\pushright}[1]{\ifmeasuring@#1\else\omit\hfill$\displaystyle#1$\fi\ignorespaces}
\newcommand{\pushleft}[1]{\ifmeasuring@#1\else\omit$\displaystyle#1$\hfill\fi\ignorespaces}
\newcommand{\specialcell}[1]{\ifmeasuring@#1\else\omit$\displaystyle#1$\ignorespaces\fi}
\makeatother

\newtheorem{theorem}{Theorem}[section]
\newtheorem*{theorem-non}{Theorem}
\newtheorem{definition}[theorem]{Definition}
\newtheorem{corollary}[theorem]{Corollary}

\usepackage{natbib}

\title{Regret-Minimizing Contracts:\\Agency Under Uncertainty}

\author{
	Martino Bernasconi\\
	Bocconi University\\
	\texttt{martino.bernasconi@unibocconi.it}
	\And
	Matteo Castiglioni\\
	Politecnico di Milano\\
	\texttt{matteo.castiglioni@polimi.it}
	 \And
	Alberto Marchesi\\
	Politecnico di Milano\\
	\texttt{alberto.marchesi@polimi.it}
}

\begin{document}
	
\maketitle


	
	

	\begin{abstract}
	We study the fundamental problem of designing contracts in \emph{principal-agent problems under uncertainty}.
	Previous works mostly addressed \emph{Bayesian} settings in which principal's uncertainty is modeled as a probability distribution over agent's types.
	In this paper, we study a setting in which the principal has no distributional information about agent's type.
	In particular, in our setting, the principal only knows some uncertainty set defining possible agent's action costs.
	Thus, the principal takes a \emph{robust (adversarial) approach} by trying to design contracts which minimize the \emph{(additive) regret}: the maximum difference between what the principal could have obtained had them known agent's costs and what they actually get under the selected contract.
	
	We provide a \emph{complete} characterization of the regret achievable by different classes of contracts, namely deterministic contracts, randomized contracts, and menus of deterministic/randomized contracts.
	Surprisingly, we show that \emph{deterministic contracts are sufficient to achieve worst-case-optimal regret}.
	%
	In particular, deterministic contracts attain regret at most $O(\sqrt{\delta})$ as $\delta \to 0$, where $\delta$ is a measure of the uncertainty level of the instance.
	This result crucially relies on a ``linearization'' operation inspired by \emph{linear} contracts. 
	Additionally, we prove that the regret is $\frac{1}{2}$-H\"older continuous as a function of $\delta$.
	This result is tight even when considering the regret attained by the most general class of menus of randomized contracts.
	Moreover, we analyze the relative power of different classes of contracts in minimizing the regret, showing that randomization and menus are equally useful, in different ways, in order to minimize the regret.
	
	We conclude the paper by studying the computational complexity of finding regret-minimizing contracts.
	While the problem is intractable even for deterministic contracts, we provide a general algorithmic template to efficiently compute approximately-regret-minimizing contracts.
	Then, we show how to instantiate our template in different principal-agent settings and for different classes of contracts, deriving several efficient approximation algorithms to find regret-minimizing contracts.
\end{abstract}

\section{Introduction}

%
\emph{Principal-agent problems}~\citep{grossman1992analysis} model scenarios in which a \emph{principal} interacts with an \emph{agent} that takes an action inducing some externalities on the former.
The vast majority of the works on principal-agent problems focus on \emph{hidden-action} models, where the principal does \emph{not} observe the action played by the agent, but they only observe an \emph{outcome} that is stochastically determined as an effect of such an action.
In these models, the principal gets a reward associated with the realized outcome, while the agent incurs in a cost for their action.
Thus, the goal of the principal is to induce the agent to take an action leading to desirable outcomes.
This is accomplished by committing upfront to a \emph{contract}, which defines a payment from the principal to the agent for every possible outcome, so as to incentivize the agent to take the desired action.

The problem of \emph{designing contracts} in principal-agent settings raises several interesting challenges at the intersection of computer science and economics.
These challenges have attracted a growing interest from the economics and computation community (see, \emph{e.g.},~\citep{dutting2019simple,guruganesh2021contracts,castiglioni2022bayesian}).  
This is also motivated by the fact that principal-agent problems are ubiquitous in real-world scenarios, such as, \emph{e.g.}, crowdsourcing platforms~\citep{ho2016adaptive}, blockchain-based smart contracts~\citep{cong2019blockchain}, and healthcare~\citep{bastani2016analysis}.

A crucial question concerning principal-agent problems that still needs to be fully answered is the following: \emph{How a principal should effectively design contracts when uncertain about agent's~features?}
Most of the previous works on principal-agent problems consider settings where the principal knows everything about the agent.
Recently, some works (see, \emph{e.g.},~\citep{alon2021contracts,guruganesh2021contracts,castiglioni2023designing}) tackled the question above by studying \emph{Bayesian} principal-agent problems.
In Bayesian settings, principal's uncertainty is modeled as a (known) probability distribution over agent's types, with types encoding relevant features of the agent (such as, \emph{e.g.}, action costs).
In this paper, we study
principal-agent settings in which the principal has \emph{no} distributional information about agent's costs, through the computational lens.
While some works (see, \emph{e.g.},~\citep{carroll2015robustness,walton2022general}) studied specific principal-agent settings with \emph{non-Bayesian} uncertainty, in this paper we address arbitrary uncertainty over costs by taking a more computational approach. 
%

We introduce and study a \emph{principal-agent problem under uncertainty} (called PAPU for short) in which the principal is uncertain about agent's action costs.
In particular, in a PAPU instance, the principal only knows an \emph{uncertainty set} $\mcC$ containing all the possible agent's costs.
Since no probability distribution over $\mcC$ is available, we take a \emph{robust (adversarial) approach}, following the same line pursued by~\citet{babichenko2022regret} in related Bayesian persuasion settings.
%
%
Thus, we address the problem faced by a principal who wants to design contracts $\Gamma$ that perform well against \emph{all} agent's costs in $\mcC$, \emph{i.e.}, contracts minimizing the \emph{(additive) regret} $\mcR(\Gamma)$. This is defined as follows:
\[
	\mcR(\Gamma) := \sup_{\vc\in  \mcC}  \Big\{  \OPT(\vc) - \Up(\Gamma, \vc) \Big\}.
\]
In words, the regret $\mcR(\Gamma)$ is the maximum difference---over all agent's costs $\vc$ in the uncertainty set $\mcC$---between $ \OPT(\vc) $, which is the expected utility that the principal could have achieved had them known agent's costs, and $\Up(\Gamma, \vc)$, which is what the principal actually gets under contract $\Gamma$.

\subsection{Our Contributions}

At a conceptual level, we split the paper into two parts.

\subsubsection{Part One: Worst-Case Regret Guarantees}

In the first part of the paper, we provide a \emph{complete} characterization of the minimum possible regret achievable by means of different classes of contracts, namely deterministic contracts, randomized contracts, menus of deterministic contracts, and menus of randomized contracts.
%
%
Our characterization shows how the regret depends on a quantity---called $\delta$ in the following---measuring the \emph{uncertainty level} of PAPU instances.
In particular, for any class of contracts $\cX$, we analyze $\inf_{\Gamma \in \cX} \mcR(\Gamma)$, \emph{i.e.}, the minimum regret attainable by contacts in $\cX$, by taking its \emph{worst-case} value over all the PAPU instances characterized by a given uncertainty level $\delta$.

As a preliminary step, in Section~\ref{sec:existence}, we investigate the existence of regret-minimizing contracts in each possible class of contracts.
We prove that, in the classes of contracts that do \emph{not} involve randomization, a regret-minimizing contract always exists.
Moreover, we show that, in classes involving randomization, a regret-minimizing contract may \emph{not} exist.
Such a dichotomy between deterministic and randomized contracts had already been discovered in the context of {Bayesian} principal-agent problems (see, \emph{e.g.},~\citep{castiglioni2023designing}).
Our results corroborate it, by showing that it also arises in our setting with distribution-free principal's uncertainty.

In Section~\ref{sec:worstCase}, we provide the core pillars of our characterization. 
%
%
Our first fundamental result is concerned with the regret $\mcR(\vp)$ attained by deterministic contracts $\vp \in \mcP$, and it can be informally stated as follows (see Theorem~\ref{thm:det_regret_upper} and Proposition~\ref{thm:regret_lower} for formal statements).
\begin{theorem-non}[Informal]
	Deterministic contracts achieve worst-case-optimal regret. More formally, $\inf_{\vp \in \mcP} \mcR(\vp)$ is at most $O(\sqrt{\delta} )$ as $\delta \to 0$, while there are PAPU instances where the minimum possible achievable regret is at least $\Omega(\sqrt{\delta})$ as $\delta \to 0$, even when considering menus of randomized contracts. 
\end{theorem-non}
The result above surprisingly shows that deterministic contracts are sufficient to achieve the minimum possible regret in the \emph{worst case}.
The crucial insight behind the result above resides in a particular ``linearization'' operation that makes any deterministic contract robust against \emph{all} the possible agent's costs.
In particular, it consists in adding a fraction of principal's utilities (\emph{i.e.}, rewards minus payments) to the contract.
%
%
This ensures that, if the original contract performs well against an agent with costs defined by some element of $\mcC$, then the ``linearized'' contract has good performances also when agent's costs are specified by \emph{any} other element of $\mcC$.
A similar ``linearization'' operation has been employed to deal with approximate-incentive compatibility of contracts in other settings (see, \emph{e.g.},~\citep{dutting2021complexity,castiglioni23multi}).
To the best of our knowledge, we are the first to employ an approach of this kind to tackle agent's costs uncertainty.
Notice that, despite \emph{linear} contracts~\citep{dutting2019simple} are \emph{not} sufficient to achieve worst-case-optimal regret in our setting, the result above shows that they still play a fundamental role.
This also corroborates the importance of linear contracts in designing robust payment schemes.

%
In Section~\ref{sec:worstCase}, we also analyze how the regret achievable by deterministic contracts varies as a function of the uncertainty level.
%
We consider a family of PAPU instances with (convex) uncertainty sets $\mcC_\delta(\vc_0)$ obtained by suitably ``scaling'' a base uncertainty set $\mcC$ with respect to some $\cvec_0 \in \mcC$, where $\delta \in [0,1]$ is a parameter that allows to control the \emph{varying} uncertainty level of the instances.
%
%
%
An example of set $\mcC_\delta(\vc_0)$ is the $L_\infty$-norm ball centered in $\vc_0$, with a (varying) radius $\delta$ representing the uncertainty level.
%
%
The ``linearization'' operation introduced above allows us to derive our second fundamental result of Section~\ref{sec:worstCase}, informally stated as follows (Theorem~\ref{thm:change_diameter} is its formal version).
\begin{theorem-non}[Informal]
	Given a family of PAPU instances with (convex) uncertainty sets $\mcC_\delta(\vc_0)$ obtained by ``scaling'' a base uncertainty set $\mcC$ with respect to $\vc_0 \in \mcC$ and according to a parameter $\delta \in [0,1]$ controlling the uncertainty level of the instances, the regret $\inf_{\vp \in \mcP} \mcR_{\delta, \vc_0}(\vp)$ attainable by means of deterministic contracts is $\frac{1}{2}$-H\"older continuous as a function of $\delta$.
	%
\end{theorem-non}
Intuitively, the result above shows that, if one modifies the uncertainty level of a given PAPU instance by suitably ``scaling'' its uncertainty set, then the regret attainable by deterministic contracts changes proportionally to the square root of the uncertainty level variation.  
%
%
This may be of interest in several real-world applications.
For instance, when the principal builds the uncertainty set by collecting data on agent's costs through experiments, it is natural to assume that the uncertainty set progressively shrinks as new data becomes available.
In such settings, the result above formalizes how the performance of deterministic contracts varies as the uncertainty set shrinks.
This may allow the principal to thoroughly evaluate the possible benefits of acquiring new data.

Finally, in Section~\ref{sec:robustness}, we complete our characterization by analyzing the relationship between different classes of contracts in terms of regret.
In particular, we show that, in order to minimize regret, \emph{both} randomization and menus may be useful.
Indeed, we show that there are PAPU instances in which the regret attained by menus of deterministic contracts is $\Omega(\sqrt{\delta})$ (as $\delta \to 0$) more than the regret achievable by means of randomized contracts, while there are other PAPU instances in which the opposite holds.
This shows that randomization and menus are incomparable in terms of regret-minimization power, and either one of them may be needed depending on the setting.

\subsubsection{Part Two: Computing Regret-Minimizing Contracts}
In the second part of the paper, we abandon the worst-case analysis perspective to focus on the computational problem of finding a regret-minimizing contract in a given PAPU instance.

As a preliminary step, we show that, in general instances, the problem of computing a regret-minimizing contract is \textsf{APX}-hard, even when considering the easiest case of deterministic contracts.
Thus, in the remaining of the paper, we study under which circumstances such a negative result can be circumvented.
This raises considerable computational challenges, since the uncertainty set may contain an infinite number of possible agent's costs, and, even when one can restrict to a finite subset of them, this may contain a ``prohibitively large'' number of elements.

In Section~\ref{sec:beyond}, we introduce a general template that allows to compute approximately-regret-minimizing contracts.
The template is generic, as it works for any given class of contracts $\cX$.
At a high level, it works by constructing a finite $\epsilon$-cover of the uncertainty set $\mcC$.
This is a subset of possible agent's costs such that, for every $\vc \in \mcC$, there is an element in the subset within $\epsilon$ ``distance'' from $\vc$.
Then, the template computes a contract in $\cX$ minimizing the regret against all agents with costs in the finite $\epsilon$-cover.
Finally, it applies a ``linearization'' procedure to the obtained contract.
This is similar in nature to the operation used in Section~\ref{sec:worstCase} to prove that deterministic contracts achieve worst-case-optimal regret, but it generally  applies to any class of contracts.

The main result related to our template can be informally stated as follows (see Theorem~\ref{thm:regret_template}).
\begin{theorem-non}[Informal]
	There exists a general template that computes an approximately-regret-minimizing contract in a given class $\cX$.
	In particular, the template can be instantiated to obtain a polynomial-time approximation algorithm whenever:
	\begin{enumerate}
		\item The uncertainty set admits a suitable ``small'' finite $\epsilon$-cover of the uncertainty set $\mcC$.
		\item There exists a polynomial-time algorithm to compute a contract in $\cX$ which minimizes the regret against all agents with costs in a given finite subset of $\mcC$.
	\end{enumerate}
\end{theorem-non}

Finally, in Section~\ref{sec:apply_template}, we showcase some PAPU instances in which it is possible to efficiently build a suitable ``small'' finite $\epsilon$-cover of $\mcC$.
Moreover, for every possible class of contracts $\cX$, we provide an algorithm to efficiently compute an approximately-regret-minimizing contract against all agents with costs in a given ``small'' finite subset of $\mcC$.
These last two pieces allow us to instantiate our template in different principal-agent settings and for different classes of contracts, deriving several efficient approximation algorithms to find regret-minimizing contracts.

\subsection{Related Works}

Next, we survey the previous works that are most related to this paper.

\paragraph{Works on Principal-Agent Problems Without Uncertainty}
The computational study of classical \emph{single-agent} principal-agent problems was initiated by~\citet{babaioff2014contract}, who study the complexity of contracts in terms of the number of different payments that they specify.
Later on, \citet{dutting2019simple} revitalized the research field by employing the computational lens to analyze the efficiency in terms of principal’s utility of \emph{linear} (\emph{i.e.}, commission-based) contracts.
Subsequent works studied the problem of computing optimal (for the principal) contracts in various settings.
\citet{dutting2021complexity} address problems in which the outcome space is combinatorial, while~\citet{dutting2022combinatorial}~and~\citet{deo2024supermodular} focus on settings where the agent can take any subset of a given set of unobservable actions.
Moreover, while~\citet{dutting2022combinatorial} focus on submodular principal's rewards,~\citet{dutting2024combinatorial} go beyond them by considering supermodular ones.
A few works addressed \emph{multi-agent} principal-agent problems.
In participial, \citet{babaioff2012combinatorial,emek2012computing} study settings in which agents have binary actions.
\citet{dutting2023multi} consider a binary-action setting with submodular principal's rewards, while \citet{deo2024supermodular} address the supermodular case.
Finally, \citet{castiglioni23multi} study particular multi-agent settings where externalities among the agents are somehow limited.

\paragraph{Works on Principal-Agent Problems With Uncertainty}
Recently, some works initiated the study of \emph{Bayesian} principal-agent problems in which the principal knows a probability distribution over a set of possible agent's types.
\citet{guruganesh2021contracts,GuruganeshPower23} study the relative power (in terms of principal's utility maximization) of different classes of contracts, including menus and linear contracts.
\citet{castiglioni2022bayesian} compare linear contracts with the best that can be achieved by contracts computable in polynomial time.
\citet{alon2021contracts,alon2023bayesian} study one-dimensional settings where agent's costs scale linearly with their type.
\citet{castiglioni2023designing} introduce the class of menus of randomized contracts in Bayesian settings, showing that they are superiority to other contracts and can be computed efficiently.
Finally, \citet{gan2022optimal} study a generalization of Bayesian principal-agent problems, proving that optimal randomized mechanisms can be computed efficiently.
Some works~\citep{carroll2015robustness,carroll2019robustness,dutting2019simple,walton2022general, yu2020robust} consider principal-agent problems with \emph{non-Bayesian} uncertainty, though restricted to very specific ``types'' of uncertainty.
%
In this paper, we address arbitrary uncertainty over costs, by taking a more computational approach. 
%
%
Finally, there are works which model uncertainty in principal-agent problems by casting them into \emph{online learning} frameworks; see~\citep{han2023learning,cohen2022learning,zhu2022online,bacchiocchi2023learning,ho2015adaptive}.

\paragraph{Related Results in Different Settings}
Some previous works studied problems related to those tackled in this paper, thought in settings different from principal-agent problems.
In particular, \citet{babichenko2022regret} study \emph{Bayesian persuasion} problems in which the sender is ignorant of receiver's utilities.
They adopt a {robust} approach, by seeking for signaling schemes for the sender that perform well for \emph{all} possible receiver’s utilities.
To do so, they introduce a notion of \emph{regret} that is the counterpart of ours in Bayesian persuasion settings.
Moreover,~\citet{10.1145/3580507.3597680} study the problem of computing {robust} leader's commitments in \emph{Stackelberg games}.
The authors adopt an approach that is substantially different from ours, since robustness is modeled by accounting for the possibility that the follower may \emph{not} perfectly best respond to leader's commitments.

\section{Preliminaries}\label{sec:preliminaries}

In this section, we introduce all the preliminary concepts needed in the rest of the paper.
First, in Section~\ref{sec:papu}, we introduce all the elements of the principal-agent settings studied in this paper.
Then, in Section~\ref{sec:regret_min}, we define the regret-minimization problem tackled by the principal in such settings.

\subsection{Principal-Agent Problems Under Uncertainty}\label{sec:papu}

An instance of \emph{principal-agent problem under uncertainty} (PAPU) is a tuple $(A,\Omega,\rvec,\mcC,\Fvec)$, where:
$A$ is a finite set of $n$ agent's actions,
$\Omega$ is a finite set of $m$ outcomes of agent's actions,
$\rvec \in [0,1]^m$ is a reward vector whose components $r_\omega$ encode principal's rewards for each outcome $\omega \in \Omega$,
$\mcC \subseteq [0,1]^n$ is an uncertainty set of possible cost vectors for the agent, and
$\Fvec \in [0,1]^{n \times m}$ is probability matrix---satisfying $\sum_{\omega \in \Omega} F_{a,\omega} = 1$ for all $a \in A$---with an entry $F_{a,\omega} $ for each action $a \in A$ and outcome $\omega \in \Omega$ encoding the probability that $\omega$ is realized when the agent plays $a$.\footnote{In the rest of this work, we assume that rewards and costs are in $[0,1]$. All the results in this paper can be easily generalized to the case of an arbitrary range of positive numbers, by applying a suitable normalization.}
Differently from classical principal-agent problems~\citep{dutting2019simple}, in the PAPU agent's costs are \emph{not} fully specified, but they could be determined by any cost vector $\cvec \in \mcC$ in the uncertainty set, with each component $c_a$ of $\cvec$ representing the cost suffered by the agent when playing action $a \in A$.
%
%
%
%
In the following, we will also make use of the notion of \emph{uncertainty level} of an instance of the PAPU, which is defined as the $L_\infty$-diameter of its uncertainty set $\cC$; formally, $d(\mcC)\coloneqq\sup_{\vc,\vc'\in\cC}\|\vc-\vc'\|_\infty$.
%
%

%
In the PAPU, the principal can incentivize the agent to take desirable actions by \emph{committing to a menu of randomized contacts}~\citep{castiglioni2023designing}.
A \emph{contract} is a payment scheme that defines a monetary transfer from the principal to the agent for each possible realized outcome.
Formally, a contract is defined as a vector $\pvec \in \mcP := \mathbb{R}_+^m$ with each component $p_\omega $ encoding a payment from the principal to the agent when the realized outcome is $\omega \in \Omega$.
A \emph{randomized} contract $\gvec \in \Delta_{\supp(\gvec)}$ is a discrete probability distribution over (deterministic) contracts, with $\supp(\gvec) \subseteq \mcP$ being its support and $\gamma_{\pvec}$ denoting the probability assigned to $\pvec \in \supp(\gvec)$.
Finally, a \emph{menu} of randomized contracts is defined as a finite set $\Gamma = \left\{  \gvec^1, \ldots, \gvec^K \right\}$ such that $\gvec^i \in \Delta_{\supp(\gvec^i)}$ for every $i \in [K]$.\footnote{In this paper, we denote by $[x]$ the set of the first $x$ natural numbers, namely $[x] := \{ 1,\ldots, x \}$.}

The interaction between the principal and the agent goes as follows:
\begin{enumerate}
	\item[(i)] The principal publicly commits to a menu of randomized contacts $\Gamma = \left\{ \gvec^1, \ldots, \gvec^K \right\}$.
	\item[(ii)] The agent selects a randomized contract $\gvec \in \Gamma$ from the menu.
	\item[(iii)] The principal samples a contract $\pvec \sim \gvec$ and publicly commits to $\pvec$.
	\item[(iv)] The agent plays an action $a \in A$.
	\item[(v)] An outcome $\omega \in \Omega$ is sampled according to the probabilities $F_{a,\omega}$.
\end{enumerate}
When the principal commits to a randomized contract $\gvec \in \Delta_{\supp(\gvec)}$ at step~(i), then the interaction directly goes to step~(iii), while, if they commit to a contract $\pvec \in \mcP$, then the next step is~(iv).

In this paper, we also consider a particular subclass of menus, which are called \emph{menus of deterministic contracts} and are sets $\Pi = \left\{  \pvec^1, \ldots, \pvec^K \right\}$ such that $\pvec^i \in \mcP$ for every $i \in [K]$.\footnote{Notice that a menu of deterministic contracts $\Pi = \left\{  \pvec^1, \ldots ,\pvec^K \right\}$ can be seen as a special menu of randomized contracts $\Gamma = \left\{ \gvec^1, \ldots, \gvec^K \right\}$ in which each $\gvec^i$ has only $\pvec^i$ in its support $\supp(\gvec^i)$. Moreover, a deterministic, respectively randomized, contract is a special menu of deterministic, respectively randomized, contracts with $K = 1$.}
When the principal commits to a menu of deterministic contracts at step (i) of the interaction, then at step (ii) the agent selects a contract $\pvec \in \Pi$ from the menu and the interaction immediately goes to step (iv).
%
In the following, we denote by $\mcPr$ the set of all the randomized contracts, while we denote by $\mcM$ and $\mcMr$ the sets of all the menus of deterministic and randomized contracts, respectively.

The relationship among sets $\mcP$, $\mcPr$, $\mcM$, and $\mcMr$ is depicted in Figure~\ref{fig:difContracts}(\emph{Left}) in Section~\ref{sec:robustness}.

%

\subsection{Regret-Minimizing Contracts}\label{sec:regret_min}

In the PAPU, agent's costs are \emph{unknown} to the principal, who only knows that they could be defined by any cost vector in the uncertainty set $\mcC$.
By taking a \emph{robust (adversarial) approach}---similar to the one adopted by~\citet{babichenko2022regret} in related Bayesian persuasion settings---the performance of a menu of randomized contracts $\Gamma \in \mcMr$ is evaluated by means of the notion of \emph{(additive) regret}.
This is defined as the maximum difference over the uncertainty set $\mcC$ between what the principal could have obtained had them known actual agent's costs and what they actually get by committing to $\Gamma$.
In the following, we formally define such a notion of regret.

After observing a contract $\pvec \in \mcP$, at step (iv) of the interaction an agent whose costs are defined by $\cvec \in \mcC$ plays an action that is:
\begin{enumerate}
	\item \emph{incentive compatible} (IC), \emph{i.e.}, it maximizes agent's expected utility over all the actions in $A$;
	\item \emph{individually rational} (IR), \emph{i.e.}, it has non-negative expected utility (if there is no IR action, the agent abstains from playing, preserving the \emph{status quo}).
\end{enumerate}
As customary in the literature (see, \emph{e.g.},~\citep{dutting2019simple}), we assume w.l.o.g. that every instance of the PAPU has an action $a \in A$ such that $c_a = 0$ for all $\cvec \in \mcC$.
This ensures that any action that is IC is also IR, thus allowing to focus w.l.o.g.~on incentive compatibility only.
For ease of notation, we let $\Ua(\pvec,\cvec, a):=\sum_{\omega\in\Omega}F_{a,\omega}\,  p_\omega  - c_a$ be the expected utility that an agent with cost vector $\cvec \in \mcC$ gets by playing an action $a \in A$ under contract $\pvec\in\mcP$.
Then, given a contract $\pvec \in \mcP$ and a cost vector $\cvec \in \mcC$, we formally denote by $A^\star (\pvec,\cvec) := \arg\max_{a \in A}   \Ua(\pvec,\cvec,a)$ the set of all IC agent's actions, which are commonly referred to as \emph{best responses}.
%
%
Moreover, we assume that the agent breaks ties in favor of the principal when indifferent among multiple best responses, as it is standard in the literature (see, \emph{e.g.},~\citep{dutting2019simple,castiglioni2023designing}).
Formally, we denote by $a^\star(\pvec, \cvec) \in \arg\max_{a \in A^\star(\pvec,\cvec)} \left\{ \sum_{\omega \in \Omega} F_{a,\omega} \left( r_\omega - p_\omega  \right) \right\}$ the best response that is actually played by an agent with cost vector $\cvec \in \mcC$ under contract $\pvec \in \mcP$.
In the following, for ease of notation, we denote by $\Up(\pvec,\cvec)$ principal's expected utility under a contract $\pvec \in \mcP$ when facing an agent with cost vector $\cvec \in \mcC$; formally, it holds $\Up(\pvec,\cvec) := \sum_{\omega \in \Omega} F_{a^\star(\pvec, \cvec) , \omega} \left( r_\omega - p_\omega  \right) $.

Given a menu of randomized contracts $\Gamma \in \mcMr$, at step (ii) of the interaction an agent whose costs are defined by $\cvec \in \mcC$ selects a randomized contract $\gvec \in \Gamma$ that maximizes their expected utility, which is equal to $\Ua(\gvec, \vc)\coloneqq\sum_{\pvec \in \supp(\gvec)} \gamma_{\pvec} \, \Ua(\pvec,\cvec, a^\star(\pvec,\cvec))$.
We denote by $S^\star(\Gamma,\cvec) \subseteq \Gamma$ the set of all such randomized contracts, while, by assuming ties are broken in favor of the principal, we let $\gvec^\star(\Gamma, \cvec) \in \arg\max_{\gvec \in S^\star(\Gamma,\cvec)} \big\{ \sum_{\pvec \in \supp(\gvec)} \gamma_{\pvec} \, \Up(\pvec,\cvec) \big\}$ be the one actually selected by the agent.
In the following, for ease of notation and by slightly abusing it, we let $\Up(\gvec, \cvec) := \sum_{\pvec \in \supp(\gvec)} \gamma_{\pvec} \, \Up(\pvec,\cvec)$ be principal's expected utility by committing to a randomized contract $\gvec \in \mcPr$ against an agent with cost vector $\cvec \in \mcC$, while we let $\Up(\Gamma, \cvec) := \Up(\gvec^\star(\Gamma, \cvec), \cvec)$ be their expected utility by committing to a menu of randomized contracts $\Gamma \in \mcMr$.
For menus of deterministic contracts $\Pi \in \mcM$, we adopt a similar notation, by writing $\Up(\Pi, \cvec) $. 
%
%
Moreover, we let $\OPT(\cvec) := \max_{\vp \in \cP} \Up(\vp,\cvec)$ be the \emph{optimal} value of principal's expected utility against an agent with cost vector $\cvec \in \mcC$.\footnote{Notice that, given a cost vector $\cvec \in\mcC$, the value $\OPT(\cvec)$ can be defined w.l.o.g. by maximizing over the (deterministic) contracts in $\mcP$, rather than over the more general set of menus of randomized contracts $\mcMr$. This is because the principal can always maximize their expected utility by committing to a single (deterministic) contract, once agent's costs are fixed.}

We are now ready to introduce the formal definition of regret-minimizing contract:
%
%
\begin{definition}[Regret-minimizing contracts]\label{def:regret}
	Given a PAPU instance and a menu of randomized contracts $\Gamma \in \mcMr$, its \emph{(additive) regret} is defined as $\mcR(\Gamma) := \sup_{\cvec \in \mcC} \left\{ \OPT(\cvec) - \Up(\Gamma,\cvec) \right\}$.
	Then, a \emph{regret-minimizing} menu of randomized contracts is one achieving regret $\inf_{\Gamma \in \mcMr} \mcR(\Gamma)$.
	Analogously, we define regret-minimizing menus of deterministic contracts, randomized contracts, and deterministic contracts, with their corresponding regrets being $\mcR(\Pi)$, $\mcR(\gvec)$, and $\mcR(\pvec)$, respectively.
\end{definition}
Let us observe that the notion of regret introduced in Definition~\ref{def:regret} also provides an upper bound on the performance of a contract when some distributional information on how agent's costs are determined is available, as it is the case in \emph{Bayesian} settings~\cite{castiglioni2023designing,guruganesh2021contracts,alon2021contracts}.
Indeed, by letting $\cD$ be a probability distribution over $\mcC$, we have:
\[
	\mcR(\Gamma)= \sup_{\cvec \in \mcC} \left\{ \OPT(\cvec) - \Up(\Gamma,\cvec) \right\} \geq \mathbb{E}_{\vc \sim \cD} \left[  \OPT(\cvec) - \Up(\Gamma,\cvec) \right] \geq \OPT(\cD) -  \mathbb{E}_{\vc \sim \cD} \left[  \ \Up(\Gamma,\cvec) \right] ,
\]
where $\OPT(\cD)$ is principal's expected utility in an optimal contract against an agent with costs determined by the distribution $\cD$.
Notice that this implies that all the positive results (in terms of regret minimization) provided in this paper carry over to Bayesian principal-agent problems.

\section{On the Existence of Regret-Minimizing Contracts}\label{sec:existence}

In this section, we start our analysis by investigating the existence of regret-minimizing contracts in general PAPU instances.
Our crucial result is that regret-minimizing deterministic contracts always exist (see Section~\ref{sec:exist}), while there are PAPU instances in which, when focusing on randomized contracts, the minimum possible regret value is unattainable (see Section~\ref{sec:exist}).

\subsection{When Do Regret-Minimizing Contracts Exist?}\label{sec:exist}

We start proving that regret-minimizing menus of deterministic contracts always exist.
In particular, we prove that, for any fixed $K \in \mathbb{N}_{>0}$, there always exists a regret-minimizing menu of deterministic contracts of size $K$, when regret minimization is performed over menus of size $K$.
%
%
%
\begin{restatable}{theorem}{theoremExistenceOne}\label{thm:exists_menu_fixed_size}
	Given a PAPU instance and a positive integer $K \in \mathbb{N}_{>0}$, there always exists a menu of deterministic contracts $\Pi^\star = \left\{  \pvec^1, \ldots, \pvec^K \right\} \in \mcM$ of size $K$ that minimizes the regret $\mcR(\Pi)$ over all the menus of deterministic contracts $\Pi\in \mcM$ of size $K$, namely $\mcR(\Pi^\star)=\inf_{\Pi\in \mcM: |\Pi|=K} \mcR(\Pi)$.
	%
\end{restatable}
Theorem~\ref{thm:exists_menu_fixed_size} instantiated with $K=1$ immediately gives the following corollary:
\begin{corollary}\label{cor:existence_det}
	Any PAPU instance admits a regret-minimizing deterministic contract.
\end{corollary}

One may n\"aively think that it is always possible to add elements to a menu of deterministic contracts so as to increase principal's expected utility.
This would imply that \Cref{thm:exists_menu_fixed_size} is \emph{not} general enough, since it only works for menus of a fixed size $K$.
However, we show that this is \emph{not} the case.
Indeed, there exists a maximum menu size beyond which utility cannot be further increased.
%
%
More precisely, when looking for regret-minimizing menus of deterministic contracts, one can restrict the attention w.l.o.g. to menus with size equal to the number of agent's actions $n$.
\begin{restatable}{theorem}{propositionExistenceOne}\label{thm:finite_support_menus}
	Given any PAPU instance, it holds that $\inf_{\Pi\in \mcM: |\Pi|=n} \mcR(\Pi) = \inf_{\Pi\in \mcM} \mcR(\Pi)$.
\end{restatable}
Theorem~\ref{thm:finite_support_menus}, together with Theorem~\ref{thm:exists_menu_fixed_size}, immediately allows us to prove the following:
\begin{corollary}\label{cor:det_exist}
	Any PAPU instance admits a regret-minimizing menu of deterministic contracts.
\end{corollary}
%
%
%
%
Corollary~\ref{cor:det_exist} shows that, if randomization is \emph{not} involved, then the regret-minimization problem is well defined, as it is the case for the utility maximization problem that is studied in related {Bayesian} principal-agent settings~\citep{castiglioni2023designing}.
This stands in sharp contrast to settings in which randomization is at play, as we show in the following.

\subsection{When Regret-Minimizing Contracts Do \emph{Not} Exist?}\label{sec:not_exist}

Next, we show that randomization makes the minimum possible regret value unattainable.
Notice that this is also the case for the utility maximization problem studied in {Bayesian} principal-agent problems, as shown by~\citet{castiglioni2023designing}.

Our results rely on a particular PAPU instance, which we introduce in the following definition: 
\begin{definition}[``Non-existence'' instance of the PAPU]\label{def:nonexist_inst}
	The \emph{``non-existence'' PAPU instance} is defined by a tuple $(A,\Omega,\rvec,\mcC,\Fvec)$ such that:
	\begin{itemize}
		\item The agent has only two actions, namely $A := \{ a_1, a_2 \}$, and each action deterministically leads to an outcome, namely $\Omega := \{ \omega_1, \omega_2 \}$ with $F_{a_1, \omega_1}=1$, $F_{a_1, \omega_2}=0$, $F_{a_2, \omega_1}=0$, and $F_{a_2, \omega_2}=1$.
		\item The principal gets positive reward in $\omega_{2}$ only, namely $r_{\omega_{1}} = 0$ and $r_{\omega_2}=1$.
		\item The principal has no knowledge about the cost of action $a_2$, namely $\mcC=\left\{ \cvec \in [0,1]^2 \mid c_{a_1}=0 \right\}$.
	\end{itemize}
	%
	%
\end{definition}

A careful analysis of the instance in Definition~\ref{def:nonexist_inst} allows us to prove the following result:
\begin{restatable}{proposition}{propositionExistenceTwo}\label{prop:nonexist_rand}
	There exists a PAPU instance in which $\inf_{\gvec \in \mcPr}\reg(\gvec)$ does not admit a minimum.
\end{restatable}
%
%
%
%
We conjecture that the result in Proposition~\ref{prop:nonexist_rand} also holds for the class of menus of randomized contracts $\mcMr$, and this can be proved by using again the instance in Definition~\ref{def:nonexist_inst}.
We leave resolving this conjecture as an open problem to be addressed in future works.

Moreover, the instance in Definition~\ref{def:nonexist_inst} also allows us to show that the $\sup$ in the definition of regret (see Definition~\ref{def:regret}) may \emph{not} admit a maximum when randomization is involved.
\begin{restatable}{proposition}{propositionExistenceFour}\label{prop:sup_no_max}
	There exists a PAPU instance and a randomized contract $\gvec \in \mcPr$ such that the problem $\sup_{\cvec \in \mcC} \left\{ \OPT(\cvec) - \Up(\gvec,\cvec)  \right\}$ does not admit a maximum.
\end{restatable}
%

\section{Worst-Case Regret Guarantees of Deterministic Contracts} \label{sec:worstCase}

In this section, we provide our two main results for the first part of the paper, \emph{by characterizing the minimum possible regret achievable by means of deterministic contracts in the worst case}.

In Section~\ref{sec:sub_worst_case_opt}, we show our first main result: \emph{deterministic contracts are sufficient to achieve worst-case-optimal regret}.
In particular, given a PAPU instance, we provide a \emph{constructive} (and efficient) proof for the existence of a deterministic contract with regret at most $O(\sqrt{d(\mcC)})$ as $d(\mcC) \to 0$, where we recall that $d(\mcC)$ is the $L_\infty$-diameter of the set $\mcC$, representing the uncertainty level.
Moreover, we exhibit a family of instances in which attaining regret at least $\Omega(\sqrt{d(\mcC)})$ as $d(\mcC) \to 0$ is unavoidable, \emph{even when the principal is allowed to commit to menus of randomized contracts}.
This shows that deterministic contracts allow to achieve the minimum possible regret (up to constant factors) in the worst case over all the PAPU instances with the same uncertainty level.
%
%


In Section~\ref{sec:sub_dep_uncertain}, we perform a more fine-grained analysis, by studying how the regret attained by deterministic contracts varies as a function of the uncertainty level.
In particular, we prove that, in any PAPU instance with a suitably-defined varying (convex) uncertainty set parametrized by some $\delta \in [0,1]$---encoding the uncertainty level of the instance---, the minimum possible regret attainable by means of deterministic contracts is $\frac{1}{2}$-H\"older continuous as a function of $\delta$.

\subsection{Deterministic Contracts Achieve Worst-Case-Optimal Regret}\label{sec:sub_worst_case_opt}

Our strategy for demonstrating that deterministic contracts achieve worst-case-optimal regret centers around a simple and intuitive operation.
This consists in slightly modifying a contract by adding a small ``linear'' component to the payments, representing a fraction of principal's utilities.   
Formally, a contract $\pvec \in \mcP$ is modified as $\pvec+\alpha(\rvec-\pvec)$ for a given fraction $\alpha \in [0,1]$.
Notice that the components of the vector $\rvec - \pvec$ encode principal's utilities under contract $\pvec$, in all the possible outcomes.
While approaches similar to ours have been adopted to deal with approximately-incentive compatible contracts (see, \emph{e.g.},~\citep{dutting2021complexity,castiglioni2023designing}), to the best of our knowledge, we are the first to employ an approach of this kind to tackle costs uncertainty.
%

First, we prove the following fundamental lemma:
\begin{restatable}{lemma}{lemmaWorstUnoRaw}\label{lem:delta_opt_raw}
	Given a PAPU instance and $\alpha \in [0,1]$, for every deterministic contract $\vp \in \mcP$ and pair of cost vectors $\cvec, \tilde\cvec\in\mcC$, the contract $\tilde\vp := \vp+ \alpha (\rvec-\vp)$  satisfies
	$
	\Up(\tilde\pvec,\tilde\cvec) \ge \Up(\pvec,\cvec)- \left( \frac{2 \| \vc - \tilde \vc \|_\infty}{\alpha} + \alpha \right).
	$
\end{restatable}
Intuitively, Lemma~\ref{lem:delta_opt_raw} states that, if the principal gets expected utility $\Up(\vp,\cvec)$ by committing to $\vp$ against an agent with cost vector $\cvec$, then, when facing an agent with cost vector $\tilde \cvec$, committing to the contract $\tilde \vp$ obtained by adding an $\alpha$ fraction of principal's utilities to $\vp$ results in expected utility at least $\Up(\vp,\cvec)$, minus an additive loss depending on $\alpha$ and the $L_\infty$-distance between $\vc$ and~$\tilde \vc$.
Notice that the result in Lemma~\ref{lem:delta_opt_raw} holds no matter how $\cvec, \tilde\cvec\in\mcC$ and $\vp \in \mcP$ are chosen.

By taking $\vp $ in Lemma~\ref{lem:delta_opt_raw} to be equal to a contract maximizing principal's expected utility against an agent with cost vector $\cvec$, and by setting $\alpha = \sqrt{2 d(\mcC)}$, we immediately prove that the modified contract $\tilde \vp$ provides the principal with expected utility at most $2 \sqrt{2 d(\mcC)}$ less than $\Up(\pvec,\cvec) = \OPT(\cvec)$ against an agent with cost vector $\tilde \cvec$, for any $\tilde \vc \in \mcC$.
This is made formal by the following lemma:
\begin{restatable}{lemma}{lemmaWorstUno}\label{lem:delta_opt}
	Given a PAPU instance, for every cost vector $\cvec\in\mcC$ and contract $\vp^\star \in \mcP$ that is optimal for the principal against an agent with cost vector $\vc$, namely $\Up(\pvec^\star,\cvec) = \OPT(\cvec)$, it holds that
	$
	\Up(\tilde\pvec,\tilde\cvec) \ge \OPT(\cvec)-2\sqrt{2 d(\mcC)}
	$
	for every $\tilde \vc \in \mcC$, where we let $\tilde\vp := \vp^\star+\sqrt{2d(\mcC)}(\rvec-\vp^\star)$.
	%
	%
	%
	%
\end{restatable}
Notice that, without our ``linearization'' operation, principal's expected utility under any contract is highly susceptible to small modifications in the cost vector, as the best response of the agent changes abruptly as a function of their costs.
Thus, Lemmas~\ref{lem:delta_opt_raw}~and~\ref{lem:delta_opt} strongly suggest for the consistent adoption of incorporating a small linear component into contracts.

The following lemma establishes a relation between optimal principal's expected utilities against two agents with different cost vectors in the uncertainty set $\mcC$.
\begin{restatable}{lemma}{lemmaWorstDue}\label{lem:holder_opt}
	Given a PAPU instance, for every $\vc,\tilde\vc\in\cC$, it holds that
	$
	|\OPT(\vc)-\OPT(\tilde\vc)|\le 2\sqrt{2 d(\mcC)}.
	$
\end{restatable}
%

Then, Lemmas~\ref{lem:delta_opt}~and~\ref{lem:holder_opt} allow to prove the following result:

\begin{restatable}{theorem}{theoremWorst}\label{thm:det_regret_upper}
	Given a PAPU instance, it holds that
	$
	\inf_{\vp\in \mcP}\reg(\vp) \le 4\sqrt{2 d(\mcC)}.
	$
\end{restatable}
%

Next, we show that the result provided in Theorem~\ref{thm:det_regret_upper} is tight (up to constant factors), by providing a family of PAPU instances (see Definition~\ref{def:hard}) in which the principal incurs in a regret at least of $\Omega(\sqrt{d(\mcC)})$ as ${d(\mcC)} \to  0$, even when employing menus of randomized contracts.
\begin{definition}[‘‘Hard'' instances of the PAPU]\label{def:hard}
	The \emph{``hard'' PAPU instances} are parametrized by a parameter $\delta \in [0, \nicefrac{1}{2}]$ and are defined by tuples $(A,\Omega,\rvec,\mcC,\Fvec)$ such that:
	\begin{itemize}
		\item The agent has two actions and there are two outcomes, \emph{i.e.}, $A := \{ a_1, a_2 \}$ and $\Omega := \{  \omega_1, \omega_2 \}$.
		\item Action $a_1$ deterministically leads to outcome $\omega_{1}$, namely $F_{a_1, \omega_1} = 1$ and $F_{a_1, \omega_2} = 0$, while action $a_2$ is such that $F_{a_2, \omega_1} = 1 -\alpha$ and $F_{a_2, \omega_2} = \alpha$, with $\alpha = \sqrt{2 \delta}$.
		\item The principal gets positive reward in $\omega_1$ only, namely $r_{\omega_1}=1$ and $r_{\omega_2}=0$.
		\item The uncertainty set is $\mcC := \left\{ \cvec \in [0,1]^2 \mid c_{a_1} \in\{0,\delta\} \wedge c_{a_2}=0  \right\}$, and, thus, $d(\mcC) = \delta$.
	\end{itemize}
	%
	%
\end{definition}
%
%
\begin{restatable}{proposition}{propostionWorst}\label{thm:regret_lower}
	There exists a family of PAPU instances $\mathcal{I}_\delta := (A,\Omega,\rvec,\mcC,\Fvec)$ parametrized by a parameter $\delta \in [0, \nicefrac{1}{2}]$ in which $d(\mcC) = \delta$ and $\inf_{\Gamma \in \mcMr} \mcR(\Gamma) \geq \sqrt{2\delta}/4$.
	%
	%
\end{restatable}
%
%
Notice that a direct consequence of Proposition~\ref{thm:regret_lower} is that the regret attainable by menus of randomized contracts is \emph{not} Lipschitz continuous as a function of the uncertainty level.
In the following Section~\ref{sec:sub_dep_uncertain}, we show that, in particular PAPU instances with a varying uncertainty set, the regret attained by deterministic contracts satisfies a weaker continuity property when seen as as function of the uncertainty level, namely $\frac{1}{2}$-H\"older continuity.

\subsection{The Regret of Deterministic Contracts as a Function of the Uncertainty Level}\label{sec:sub_dep_uncertain}

Next, we provide a more fine grained analysis of the regret achievable by deterministic contracts, by showing how this varies as a function of the uncertainty level $d(\mcC)$.
%
In the rest of this subsection, we assume that all the uncertainty sets $\mcC \subseteq [0,1]^n$ are convex.

%
%
%
%

Our analysis relies on a suitably-defined ``scaling'' of the uncertainty sets of PAPU instances.
In particular, given a convex uncertainty set $\mcC \subseteq [0,1]^n$, we denote by $\mcC_\delta(\cvec_0) \subseteq [0,1]^n$ the \emph{$\delta$-scaling} of the set $\mcC$ centered in $\cvec_0$.
Formally, for every $\delta \in [0,1]$ and $\cvec_0 \in \mcC$, it holds:
%
\[
	\mcC_\delta(\cvec_0) := \left\{ \delta\cvec+\cvec_0(1-\delta) \mid \cvec\in \mcC\right\}.
\]
Intuitively, the $\delta$-scaling $\mcC_\delta(\cvec_0)$ of an uncertainty set $\mcC$ is obtained by moving the points in $\mcC$ towards $\cvec_0$, by taking their convex combinations of weight $\delta$ with $\cvec_0$.
Notice that $\cvec_0$ is a ``center point'' for the $\delta$-scaling, since it is an invariant for such a transformation, as $\delta\cvec_0+\cvec_0(1-\delta) =\cvec_0$.

The following lemma characterizes the properties of the $\delta$-scaling transformation:
\begin{restatable}{lemma}{lemmaWorstTre}\label{lem:scaling}
	Given a convex uncertainty set $\mcC \subseteq [0,1]^n$ with $L_\infty$-diameter $d(\mcC)$, we have:
	\begin{enumerate}
		\item for every $\delta, \delta' \in[0,1]$ such that $\delta \leq \delta'$ and $\cvec_0\in \mcC$, it holds $ \mcC_\delta(\cvec_0)\subseteq  \mcC_{\delta'}(\vc_0)$;
		\item for every $\delta\in[0,1]$ and $\cvec_0\in \mcC$, the $L_\infty$-diameter of $\mcC_\delta(\cvec_0)$ is $\sup_{\cvec,\cvec'\in \mcC_\delta(\cvec_0)}\|\cvec-\cvec'\|_\infty=\delta d(\mcC)$.
	\end{enumerate}
\end{restatable}

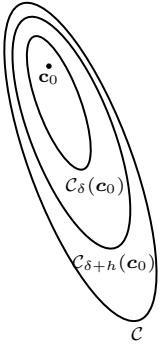
\begin{wrapfigure}[16]{l}{0.25\textwidth}
	\begin{center}
		\vspace{-0.5cm}
		\tikzset{every picture/.style={line width=0.75pt}} 

\begin{tikzpicture}[x=0.75pt,y=0.75pt,yscale=-0.65,xscale=0.65]
	
	\draw   (148.24,36.27) .. controls (155.77,10.91) and (181.41,25) .. (205.5,67.76) .. controls (229.59,110.51) and (243.01,165.73) .. (235.48,191.1) .. controls (227.95,216.46) and (202.31,202.37) .. (178.22,159.62) .. controls (154.13,116.86) and (140.71,61.64) .. (148.24,36.27) -- cycle ;
	\draw   (158.67,44.17) .. controls (163.11,31.36) and (177.2,41.76) .. (190.15,67.38) .. controls (203.09,93) and (209.98,124.14) .. (205.53,136.94) .. controls (201.09,149.75) and (187,139.35) .. (174.06,113.73) .. controls (161.12,88.11) and (154.23,56.97) .. (158.67,44.17) -- cycle ;
	\draw    (174.41,62.04) ;
	\draw [shift={(174.41,62.04)}, rotate = 0] [color={rgb, 255:red, 0; green, 0; blue, 0 }  ][fill={rgb, 255:red, 0; green, 0; blue, 0 }  ][line width=0.75]      (0, 0) circle [x radius= 1.34, y radius= 1.34]   ;
	\draw   (143.04,25.88) .. controls (153.63,-4.62) and (187.21,20.14) .. (218.04,81.18) .. controls (248.88,142.23) and (265.3,216.44) .. (254.71,246.94) .. controls (244.12,277.44) and (210.54,252.69) .. (179.71,191.64) .. controls (148.87,130.6) and (132.46,56.38) .. (143.04,25.88) -- cycle ;
	
	\draw (174.41,65.44) node [anchor=north] [inner sep=0.75pt]    {\scriptsize$\mathbf{c}_{0}$};
	\draw (185,145.59) node [anchor=north west][inner sep=0.75pt]    {\scriptsize $\mathcal{C}_{\delta }(\vc_0)$};
	\draw (243,262.4) node [anchor=north] [inner sep=0.75pt]    {\scriptsize $\mathcal{C}$};
	\draw (225,205.4) node [anchor=north] [inner sep=0.75pt]    {\scriptsize $\mathcal{C}_{\delta +h}(\vc_0)$};

\end{tikzpicture}
	\end{center}
	\caption{Uncertainty set $\mcC$ and two of its ``scalings''.}
	\label{fig:Cdelta}
\end{wrapfigure}
Intuitively, point~(1) in Lemma~\ref{lem:scaling} shows that the $\delta$-scaling transformation preserves the ``shape'' of the uncertainty set $\mcC$, while point~(2) states that the $L_\infty$-diameter of the $\delta$-scaling $\mcC_\delta(\cvec_0)$ shrinks proportionally to $\delta$ with respect to the $L_\infty$-diameter $d(\mcC)$ of the original set $\mcC$.
\Cref{fig:Cdelta} illustrates some examples of $\delta$-scaling transformations.

%

A representative example of $\delta$-scaling is the case of $L_\infty$-norm balls.
Given an uncertainty set $\mcC := \left\{ \cvec \in [0,1]^n \mid || \cvec - \cvec_0 ||_\infty \leq D  \right\}$ defined as an $L_\infty$-norm ball contained in $[0,1]^n $ with radius $D \in [0,1]$ and center $\cvec_0 \in [0,1]^n$, the $\delta$-scaling $\mcC_\delta(\cvec_0)$ of $\mcC$ centered in $\cvec_0$ is the $L_\infty$-norm ball with the same center and radius $\delta D$.
%
%
As we also discuss in Section~\ref{sec:eps_cover}, an uncertainty set with this shape may arise in all settings in which agent's costs are estimated from data, in order to build some confidence region in which the true agent's costs are contained with sufficiently high probability.
In this case, the parameter $\delta$ of the $\delta$-scaling transformation represents how the confidence region shrinks after observing some new data that allows to refine the estimation of agent's costs.

Next, we introduce our main result relating the regret achievable by deterministic contracts with the uncertainty level.
Given a PAPU instance $(A,\Omega,\rvec,\mcC,\Fvec)$, for every $\delta \in [0,1]$ and $\vc_0 \in \mcC$, we let $\mcR_{\delta,\cvec_0}(\pvec)$ be the regret of a given deterministic contract $\pvec \in \mcP$ in a new PAPU instance obtained by applying a $\delta$-scaling centered in $\cvec_0$ to the uncertainty set $\mcC$.
Formally, we have that $\mcR_{\delta,\cvec_0}(\pvec) := \sup_{\cvec \in \mcC_\delta(\cvec_0)} \left\{ \OPT(\cvec) - \Up(\pvec,\cvec)  \right\}$.
Notice that, by definition of $\delta$-scaling, it holds that $\mcR_{1,\cvec_0}(\pvec) = \mcR(\pvec)$ for every $\cvec_0 \in \mcC$, while, by applying Lemma~\ref{lem:scaling}, the $L_\infty$-diameter of the uncertainty set of the new PAPU instance is $\delta d(\mcC)$.
Then, we can state the following result:
%
%
\begin{restatable}{theorem}{theoremWorstDue}\label{thm:change_diameter}
	Given an instance $(A,\Omega,\rvec,\mcC,\Fvec)$ of the PAPU, for every scaling parameter $\delta \in [0,1]$, offset $h \in [0,1]$ such that $\delta + h \leq 1$, and center $\cvec_0 \in \mcC$, the following holds:
	\[
		0\le\inf_{\pvec \in \mcP} \mcR_{\delta+h, \cvec_0}(\pvec) - \inf_{\pvec \in \mcP} \mcR_{\delta, \cvec_0} (\pvec) \leq 4\sqrt{2h}.
	\]
\end{restatable}
Intuitively, Theorem~\ref{thm:change_diameter} establishes that the regret of deterministic contracts is $\frac{1}{2}$-H\"older continuous as a function of the parameter $h$ controlling how the uncertainty level varies.
Specifically, the theorem states that the regret attainable by deterministic contracts increases of at most $4 \sqrt{2h}$ by switching from the instance with uncertainty set defined by the $\delta$-scaling of $\mcC$, where $d(\mcC) = \delta$, to the instance in which the uncertainty set is the $(\delta+h)$-scaling of $\mcC$, where $d(\mcC) = \delta + h$. 
%

Notice that the bound in Theorem~\ref{thm:change_diameter} is tight due to \cref{thm:regret_lower}.
Indeed, since $\mcR_{0,\cvec_0} (\pvec) = 0$ for every $\cvec_0 \in \mcC$ by definition, if we let $\delta = 0$ in Theorem~\ref{thm:change_diameter}, then we get that $\mcR_{h, \cvec_0}(\pvec) \leq 4 \sqrt{2h}$ for every $\cvec_0 \in \mcC$ and $h \in [0,1]$.
This is tight (up to constant factors) by \cref{thm:regret_lower} with $\delta = h$.
%
%

%
%
%
%

\section{The Regret Across Different Classes of Contracts} \label{sec:robustness}

In this section, we describe the relationship between different classes of contracts in terms of regret, by adopting a {worst-case approach} as in Section~\ref{sec:worstCase}.
In the previous Section~\ref{sec:worstCase}, we proved that deterministic contracts achieve a regret of the order of $O(\sqrt{d(\mcC)})$ as $d(\mcC) \to 0$, which is optimal when considering the worst case over all instances with the same uncertainty level.
Next, we study \emph{how the worst-case ``regret gap'' between different classes of contracts depends on the uncertainty level $d(\mcC)$}, deriving the set of results graphically depicted in Figure~\ref{fig:difContracts}(\emph{Right}).

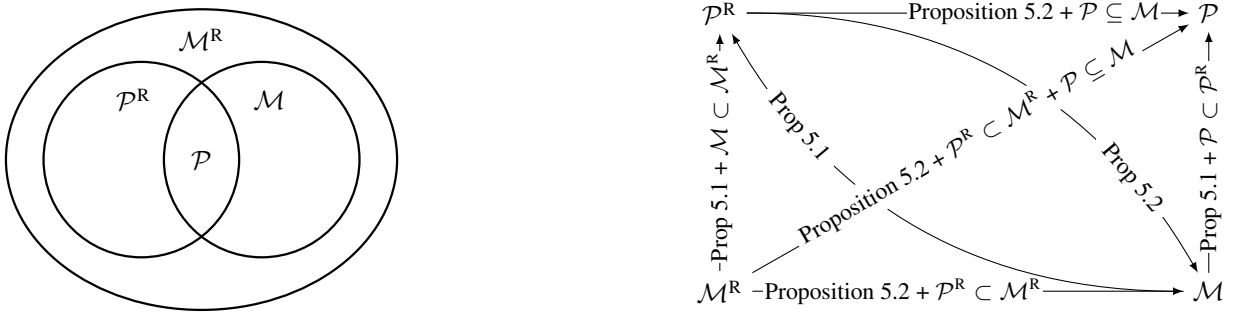
\begin{figure}[!htp]
	\centering
	\subfigure{
		\centering
		\begin{tikzpicture}
	
	\draw[line width=0.3mm, black] (1.6,1.6) ellipse (2.6cm and 2cm);
	\node at (1.6, 3.2) {$\mcMr$};
	
	\draw[line width=0.3mm, black] (0.8,1.6) circle (1.3cm);
	\node at (0.7, 2.4) {$\mcPr$};
	
	\draw[line width=0.3mm, black] (2.4,1.6) circle (1.3cm);
	\node at (2.5, 2.4) {$\mcM$};
	
	\node at (1.6, 1.6) {$\cP$};
\end{tikzpicture}}
	\hfill
	\subfigure{
		\begin{tikzpicture}
	\def\dhorizontal{6.5cm}
	\def\dvertical{3.7cm}
	\centering
	\node (R) at (3,6) {$\mcPr$};
	\node (D) at ($(R)+(\dhorizontal,0)$) {$\cP$};
	\node (MR) at ($(R)+(0,-\dvertical)$) {$\mcMr$};
	\node (MD) at ($(D)+(0,-\dvertical)$) {$\mcM$};
	
	
	\path[->,>=latex] (R) edge              node [inner sep=0.1ex, pos=0.65, fill=white, align=center]  {\small\Cref{lem:RbetterthenMD} + $\mcP\subseteq \mcM$} (D)
	(R)    edge [bend left] node  [pos=0.8, align=center, inner sep=0.1ex, fill=white]        {\rotatebox[origin=c]{-50}{\small Prop~\ref{lem:RbetterthenMD}}} (MD)
	(MR)    edge node [align=center, pos=0.35, inner sep=0.1ex, fill=white]          {\small\Cref{lem:RbetterthenMD} + $\mcPr\subset\mcMr$} (MD)  
	(MD)	    edge [bend left] node [align=center, pos=0.8, inner sep=0.1ex, fill=white]  {\rotatebox[origin=cb]{-50}{\small Prop~\ref{lem:MDbetterthenR}}} (R)
	(MR) edge              node [above, align=center, anchor=base, pos=0.5,inner sep=0.5ex, fill=white, rotate=30]  {\rotatebox[origin=c]{0}{\small\Cref{lem:RbetterthenMD} + $\mcPr\subset\mcMr$ + $\mcP\subseteq \mcM$}} (D)
	(MR) edge              node [align=center, pos=0.48, inner sep=0.1ex, fill=white]  {\rotatebox[origin=c]{90}{\small Prop~\ref{lem:MDbetterthenR} + $\mcM \subset \mcMr$}} (R)
	(MD) edge              node [align=center, pos=0.48, inner sep=0.1ex, fill=white]  {\rotatebox[origin=c]{90}{\small Prop~\ref{lem:MDbetterthenR} + $\mcP\subset\mcPr$}} (D);
\end{tikzpicture}}
		\caption{(\emph{Left}) Diagram depicting the ``inclusion'' relationship between different classes of contracts. (\emph{Right}) Diagram depicting the relationship between different classes of contracts in terms of ``regret gap''. An arrow from class $\cX$ to class $\cY$ means that there exists a family of PAPU instances parametrized by $\delta$ in which the uncertainty level is $d(\mcC) = \delta$, the regret attained by contracts of class $\cX$ is at most $R_\delta$, while all the contracts of class $\cY$ achieve regret at least $R_\delta+\Omega(\sqrt\delta)$ as $\delta \to 0$.}
		\label{fig:difContracts}
\end{figure}

For ease of presentation, in Figure~\ref{fig:difContracts}(\emph{Left}) we showcase the ``inclusion'' relationship between different classes of contracts.
In particular, the most general class is $\mcMr$, which is the one of menus of randomized contracts.
Moreover, both the class of menus of deterministic contracts $\mcM$ and the one of randomized contracts $\mcPr$ are special subclasses of $\mcMr$.
Finally, the class of deterministic contracts $\mcP$ exactly coincides with the ``intersection'' of $\mcPr$ and $\mcM$.

The results in this section revolve around the following two propositions. 
These are crucial in order to establish the worst-case ``regret gap'' between randomized contracts $\mcPr$ and menus of deterministic contracts $\mcM$, which are the only two classes that are \emph{not} comparable in terms of ``inclusion'' relationship (see Figure~\ref{fig:difContracts}(\emph{Right})).
The first proposition introduces a family of PAPU instances parametrized by $\delta \in [0,1]$ in which the uncertainty level is $\delta$ and the regret attainable by menus of deterministic contracts is $\Omega(\sqrt{\delta})$ better than the one of randomized contracts as $\delta \to 0$.
\begin{restatable}{proposition}{propositionCompareOne}\label{lem:MDbetterthenR}
	There exists a family of PAPU instances $\mathcal{I}_\delta := (A,\Omega,\rvec,\mcC,\Fvec)$ parametrized by $\delta \in [0,1]$ in which $d(\mcC) = \delta$, the regret $\inf_{\Pi \in \mcM} \mcR(\Pi)$ attained by menus of deterministic contracts is at most $R_\delta$, and randomized contracts attain regret $\inf_{\gvec \in \mcPr} \mcR(\gvec)$ at least $ R_\delta + \Omega(\sqrt{\delta})$ as $\delta \to 0$. 
	%
\end{restatable}
The second proposition shows a ``reversed'' result, in which the roles of randomized contracts and menus of deterministic contracts are exchanged.
To do so, it reuses the family of PAPU instances introduced in Definition~\ref{def:hard} to show that deterministic contracts are worst-case optimal.
\begin{restatable}{proposition}{propositionCompareTwo}\label{lem:RbetterthenMD}
	There exists a family of PAPU instances $\mathcal{I}_\delta := (A,\Omega,\rvec,\mcC,\Fvec)$ parametrized by $\delta \in [0,\nicefrac{1}{2}]$ in which $d(\mcC) = \delta$, the regret $\inf_{\gvec \in \mcPr} \mcR(\gvec)$ attained by randomized contracts is at most $R_\delta$, and menus of deterministic contracts attain regret $\inf_{\Pi \in \mcM} \mcR(\Pi)$ at least $ R_\delta + \Omega(\sqrt{\delta})$ as $\delta \to 0$. 
	%
\end{restatable}
Propositions~\ref{lem:MDbetterthenR}~and~\ref{lem:RbetterthenMD} essentially show that the classes $\mcPr$ and $\mcM$ are \emph{not} comparable in terms of worst-case ``regret gap'', as expected given their ``inclusion'' relationship.
Moreover, they also show that \emph{randomization and menus are incomparable in terms of regret-minimization power}, as either one of them may be needed depending on the specific instance at hand.

The other results depicted in Figure~\ref{fig:difContracts}(\emph{Right}) can be derived by combining Propositions~\ref{lem:MDbetterthenR}~and~\ref{lem:RbetterthenMD} with the ``inclusion'' relationships in Figure~\ref{fig:difContracts}(\emph{Left}).
In particular, the arrow from $\mcM$ to $\mcP$ is obtained by Proposition~\ref{lem:MDbetterthenR} since deterministic contracts are a special case of randomized contracts, namely $\mcP \subset \mcPr$.
Similarly, Proposition~\ref{lem:RbetterthenMD} gives the arrow from $\mcPr$ to $\mcP$, since $\mcP \subset \mcM$.
Moreover, given that menus of deterministic contracts are a subclass of  menus of randomized contracts, namely $\mcM \subset \mcMr$, Proposition~\ref{lem:MDbetterthenR} immediately gives the arrow from $\mcMr$ to $\mcPr$.
An analogous reasoning gives the arrow from $\mcMr$ to $\mcM$, by Proposition~\ref{lem:RbetterthenMD} and $\mcPr \subset \mcMr$.
Finally, the arrow from $\mcMr$ to $\mcP$ can be obtained by Proposition~\ref{lem:RbetterthenMD}, together with the fact that $\mcPr \subset \mcMr$ and $\mcP \subset \mcM$.


Finally, notice that all the results in \Cref{fig:difContracts}(\emph{Right}) are tight, as \Cref{thm:regret_lower} shows that, in the worst case, the regret of menus of randomized contracts is at least $\Omega(\sqrt{\delta})$ as $\delta \to 0$. 
%

\section{Beyond the Worst Case: A Template to Minimize the Regret} \label{sec:beyond}

In Sections~~\ref{sec:worstCase}~and~\ref{sec:robustness}, we deeply investigated the regret achievable by different classes of contracts, by taking a worst-case analysis perspective.
From now on, we change point of view and we study the computational problem of finding a regret-minimizing contract given a PAPU instance.

First, we prove the following negative result, showing that \emph{a regret-minimizing contract cannot be approximated within any constant factor}, even in the case of \emph{deterministic contracts}.
%
%
\begin{restatable}{theorem}{theoremBeyondOne}
	The problem of computing a regret-minimizing deterministic contract is \textnormal{\textsf{APX}}-hard.
	%
\end{restatable}
In the remaining of the paper, we study in which special cases such a negative result can be circumvented.
In the rest of this section, we introduce a general template to compute approximately-regret-minimizing contracts.
Then, in Section~\ref{sec:apply_template}, we show how the template can be efficiently instantiated for particular instances and for specific classes of contracts.
This allows us to derive polynomial-time approximation algorithms to find regret-minimizing contracts in several settings. 

%

\subsection{A General Template to Compute Approximately-Regret-Minimizing Contracts}\label{sec:template}

In the following, we introduce our general template to compute approximately-regret-minimizing contracts.
The pseudocode of the template is provided in Algorithm~\ref{alg:template}.
At a high level, our template works by computing a regret-minimizing contract which only accounts for cost vectors in a finite subset of the uncertainty set $\mcC$.
Such a set is a suitable $\epsilon$-cover of $\mcC$, made by cost vectors which are selected in such a way that every element of $\mcC$ has at least one element of the $\epsilon$-cover within $L_\infty$-distance $\epsilon$ from it (see Definition~\ref{def:eps_cover}).
Then, a suitable ``linearization'' procedure is applied to the computed contract, so as to obtain the desired output.
Such a procedure follows the same idea as the operation employed in Section~\ref{sec:worstCase} to show that deterministic contracts achieve worst-case-optimal regret, but it more generally applies to any possible class of contracts.
%

Intuitively, working with a finite $\epsilon$-cover of the uncertainty set $\cC$ allows to find an approximately-regret-minimizing contract by solving an optimization problem defined over a finite number of possible cost vectors, instead of dealing with general uncertainty sets with a possibly infinite number of elements.
As we show in the following Section~\ref{sec:apply_template}, in the cases in which there exists a finite $\epsilon$-cover with a sufficiently ``small'' number of elements, the optimization problem can be solved efficiently.
Moreover, the final ``linearization'' procedure performed by the algorithm allows to ensure that the computed contract performs well \emph{not} only against agents with cost vectors in the $\epsilon$-cover of $\mcC$, but also against other agents whose cost vectors are \emph{not} in the $\epsilon$-cover.

\begin{algorithm}[!htp]\caption{\textsc{Compute Approximately-Regret-Minimizing Contracts}}\label{alg:template}
	\begin{algorithmic}[1]
	\REQUIRE{PAPU instance $(\Omega,A,\rvec,\mcC, \Fvec)$, Class of contracts $\cX$, Parameter $\epsilon \in [0, 1 ]$ of the $\epsilon$-cover of $\mcC$}
	\STATE $\cC_\epsilon \gets$ \texttt{Build-Cover}$(\cC,\epsilon)$\;
	\STATE ${\Gamma} \gets$ \texttt{Compute-RM-Contract}$(\cX,\cC_\epsilon)$\;
	\STATE $\widetilde{\Gamma} \gets$ \texttt{Linearize}$ ( {\Gamma},\sqrt{2\epsilon} )$\;
	\RETURN{contract $\widetilde{\Gamma} \in \cX$ of class $\cX$ such that $\mcR(\widetilde{\Gamma}) \le \inf_{\Gamma' \in \cX}  \sup_{\vc \in \cC} \Big\{ \OPT(\vc)-\Up(\Gamma',\vc) \Big\} +2\sqrt{2\epsilon}$}
		\end{algorithmic}
\end{algorithm}

Algorithm~\ref{alg:template} takes as input a PAPU instance $(\Omega,A,\rvec,\mcC, \Fvec)$, a class of contracts $\cX$ (\emph{i.e.}, one among the classes $\mcP$, $\mcPr$, $\mcM$, and $\mcMr$), and a parameter $\epsilon \in [0, 1 ]$ characterizing the $\epsilon$-cover of $\mcC$.
%
For ease of presentation, in Algorithm~\ref{alg:template} and in the rest of this section, we use $\Gamma$ to denote a contract belonging to a generic class of contracts $\cX$, by using the same notation adopted for menus of randomized contracts.
Notice that contracts in the classes $\mcP$, $\mcPr$, and  $\mcM$ can always be seen as special menus of randomized contracts, since the class $\mcMr$ ``includes'' all the other classes.
Algorithm~\ref{alg:template} sequentially performs the following three steps:
\begin{itemize}
	\item[(i)] It builds a finite $\epsilon$-cover $\mcC_\epsilon$ of the uncertainty set $\mcC$ (see Definition~\ref{def:eps_cover}) by calling a procedure \texttt{Build-Cover}$(\cC,\epsilon)$, which takes as input $\mcC$ and the parameter $\epsilon$ characterizing the $\epsilon$-cover.
	\item[(ii)] It computes a contract $\Gamma \in \cX$ of class $\cX$ which minimizes the regret over all contracts in class $\cX$ and against all agents with cost vectors in the $\epsilon$-cover $\mcC_\epsilon$ built at step~(i). This is done by a procedure \texttt{Compute-RM-Contract}$(\cX,\widetilde{\cC})$, formally introduced in Definition~\ref{def:apxAlgo}, for $\widetilde{\mcC} = \mcC_\epsilon$.
	\item[(iii)] It applies a ``linearization'' procedure to the contract $\Gamma$ computed at step~(ii), by calling a procedure \texttt{Linearize}$ ( {\Gamma},\beta )$ for $\beta = \sqrt{2\epsilon} $. Such a procedure adds a fraction $\beta$ of principal's utilities $\rvec - \vp$ to every deterministic contract $\vp \in \mcP$ involved in $\Gamma$. Formally, given any menu of randomized contracts $\Gamma = \{ \gvec^1, \ldots \gvec^K \} \in \mcMr$ and $\beta \in [0,1]$, the procedure \texttt{Linearize}$ ( {\Gamma},\beta )$ modifies $\Gamma$ by replacing $\vp$ with $\vp + \beta (\rvec - \vp)$ for every $i \in [K]$ and $\vp \in \supp (\gvec^i)$.
\end{itemize}
Notice that, in the rest of this section, we assume that the procedures \texttt{Build-Cover}$(\cC,\epsilon)$ and \texttt{Compute-RM-Contract}$(\cX,\cC_\epsilon)$ are general.
Their specific implementations are provided in the following Section~\ref{sec:apply_template}, where we instantiate our template to specific settings.

\subsection{Computational and Regret Guarantees Provided by the Template}

Before deriving the technical results related to our template, we introduce some needed formal definitions. 
We start with the definition of $\epsilon$-cover of an uncertainty set $\mcC$.
%
%
\begin{definition}[$\epsilon$-Cover]\label{def:eps_cover}
	Given an uncertainty set $\mcC \in [0,1]^n$ and any $\epsilon \in [0,1]$, an \emph{$\epsilon$-cover} of $\mcC$ is a set $\mcC_\epsilon \subseteq \mcC$ such that, for every cost vector $\vc \in \mcC$, there exists a $\tilde \vc \in \mcC_\epsilon$ with
		$ \|\vc- \tilde \vc  \|_{\infty}\le \epsilon$. 
	%
\end{definition}
It is easy to see that, given any $\mcC \in [0,1]^n$ and $\epsilon \in (0,1]$, it is always possible to build an $\epsilon$-cover of $\mcC$ that contains a finite number of elements.
However, the size of such a cover depends on the specific structure of the uncertainty set.
As we show in the following Section~\ref{sec:apply_template}, in some cases it is possible to efficiently build a suitable finite $\epsilon$-cover having ``small'' size. This is exploited to design efficient approximation algorithms by instantiating our general template.
%

Next, we formally define the procedure $\texttt{Compute-RM-Contract}(\cX,\widetilde{\cC})$.
Intuitively, it returns a contract $\Gamma \in \cX$ of class $\cX$ which attains the minimum possible regret achievable by contracts in class $\cX$, when only considering the cost vectors in a finite set~$\widetilde{\cC}$.
In the rest of this section, we always assume that a PAPU instance $(\Omega, A, \rvec, \mcC, \Fvec)$ is given.
%
%
\begin{definition}[\texttt{Compute-RM-Contract}]\label{def:apxAlgo}
	Given a class of contracts $\cX$ and a finite subset $\widetilde{\cC} \subseteq \mcC$ of cost vectors, the $\texttt{Compute-RM-Contract}(\cX,\widetilde{\cC})$ procedure returns a contract $ \Gamma \in \cX$ such that:
	\[\sup_{\vc \in \widetilde{\cC}} \Big\{ \OPT(\vc)-\Up(\Gamma,\vc) \Big\} - \inf_{\Gamma' \in \cX}  \sup_{\vc \in \widetilde{\cC}} \Big\{ \OPT(\vc)-\Up(\Gamma',\vc) \Big\} \le 0.\]
	%
\end{definition}
%
%
%

%
Finally, in order to derive our results, we need to introduce a particular notion of \emph{approximate-incentive compatibility}, which generally applies to any class of contracts, including the most general class of randomized contracts.\footnote{Notice that the notion of $\eta$-incentive compatibility introduced in Definition~\ref{def:apxIC} should \emph{not} be confused with the notion of IC action introduced in Section~\ref{sec:preliminaries}, as the former is related to menus of contracts rather than actions.}
Formally, we provide the following definition:
\begin{definition}[$\eta$-incentive compatibility]\label{def:apxIC}
	Given any $\eta \in [0,1]$, a menu of randomized contracts $\Gamma = \{ \gvec^1, \ldots, \gvec^K \} \in \mcMr$, and two functions $t:\cC\rightarrow [K]$ and $g:\cC \times \mathcal{P}\rightarrow A$, we say that the triplet $(\Gamma, t, g)$ is \emph{$\eta$-incentive compatible} if the following holds for every $\vc \in \cC$ and $i \in [K]$:
	%
	%
	\[ \sum_{\vp \in \supp(\gvec^{t(\vc)})} \gamma^{t(\vc)}(\vp) \left[\sum_{\omega \in \Omega} F_{g(\vc,\vp),\omega} p_\omega - c_{g(\vc,\vp)} \right]\ge \Ua (\gvec^i, \vc) - \eta.
	\]
\end{definition}
Intuitively, the function $t:\cC\rightarrow [K]$ in Definition~\ref{def:apxIC} encodes how agents with different cost vectors select randomized contracts from the proposed menu, while the function $g:\cC \times \mathcal{P}\rightarrow A$ specifies how they select which action to play given the realized deterministic contracts.
Thus, the relation in Definition~\ref{def:apxIC} can be interpreted as the fact that, if an agent with cost vector $\vc \in \mcC$ selects contracts and actions as specified by the functions $t$ and $g$, respectively, then the agent loses at most $\eta$ with respect to selecting an expected-utility-maximizing contract $\gvec^\star(\Gamma,\vc)$ and playing a best response $a^\star(\vp,\vc)$ to every possible deterministic contract $\vp \in \mcP$.
Notice that the functions $t$ and $g$ are needed in order to specify how an agent behaves when we relax the assumption that they select contracts and actions so as to perfectly maximize expected utility, allowing for an additive loss $\eta$.
Let us also observe that Definition~\ref{def:apxIC} can be also applied when $\Gamma$ is an element of a class of contracts that do \emph{not} involve menus, such as $\mcP$ and $\mcPr$.
In such cases, the function $t:\cC\rightarrow [K]$ is always such that $t(\vc) = 1$ for every $\vc \in\mcC$, since it holds $K=1$.

Now, we are ready to derive our results related to our template.
First, we show that, given an $\epsilon$-cover of $\mcC$ and a contract $\Gamma \in \cX$ of class $\cX$, there always exist two functions $t$ and $g$ as in Definition~\ref{def:apxIC} such that: (1) the triplet $(\Gamma,t,g)$ is $(2 \epsilon)$-incentive compatible, and (2) the regret attained by $\Gamma$ against agents with cost vectors in $\mcC$ behaving as prescribed by $t$ and $g$ is the same as the regret that $\Gamma$ achieves against expected-utility-maximizer agents with cost vectors in $\mcC_\epsilon$.
%
%
%
\begin{restatable}{lemma}{lemmaTemplateOne}\label{lm:existsEpsIC}
	Given any $\epsilon \in [0,1]$, an $\epsilon$-cover $\mcC_\epsilon$ of the uncertainty set $\mcC$, and a contract $\Gamma \in \cX$ of class $\cX$, there exist two functions $t$ and $g$ as in Definition~\ref{def:apxIC} such that $(\Gamma, t, g)$ is $(2\epsilon)$-incentive compatible and the following holds:
	\[ 
		\sup_{\vc \in \cC} \left\{  \OPT(\vc)- \sum_{\vp \in \supp(\gvec^{t(\vc)})} \gamma^{t(\vc)}_{\vp} \left[ \sum_{\omega \in \Omega} F_{g(\vc,\vp),\omega} \, r_\omega  - \sum_{\omega \in \Omega} F_{g(\vc,\vp),\omega} \, p_\omega \right] \right\} =  \sup_{\vc \in \cC_\epsilon} \Big\{ \OPT(\vc)-\Up(\Gamma,\vc) \Big\} . 
	\]
	%
	%
\end{restatable}

Then, we show that, given a contract $\Gamma \in \cX$ of class $\cX$ such that $(\Gamma,t,g)$ is $\eta$-incentive compatible for two functions $t$ and $g$ as in Definition~\ref{def:apxIC}, the contract $\widetilde{\Gamma} \in \cX$ obtained from $\texttt{Linearize}(\Gamma,\beta)$ provides the principal with an expected utility that is at most $\beta + \eta$ less than what they obtain under $\Gamma$ when the agents behave as prescribed by $t$ and $g$.
Intuitively, such a result shows that the ``linearization'' operation allows to align agent's utility with principal's rewards, so as to bound the expected utility loss of the principal.
This is made formal by the following lemma, which generalizes a similar result by~\citet{castiglioni2023designing}, by dealing with menus of \emph{randomized} contracts.
%
%
%
%
\begin{restatable}{lemma}{lemmaTemplateTwo}\label{lm:linearized}
	Given any $\eta \in [0,1]$, let $\Gamma \in \cX$ be a a contract of class $\cX$ such that $(\Gamma,t,g)$ is $\eta$-incentive compatible for two functions $t$ and $g$ as in Definition~\ref{def:apxIC}.
	Then, given any $\beta \in [0,1]$, the procedure $\texttt{Linearize}(\Gamma,\beta)$ produces a contract $\widetilde{\Gamma} \in \cX$ of class $\cX$ such that, for every $\vc \in \cC$, the following holds:
	\[ 
		\Up(\widetilde \Gamma,\vc)\ge \sum_{\vp \in \supp(\gvec^{t(\vc)})} \gamma^{t(\vc)}_{\vp}   \left[\sum_{\omega \in \Omega} F_{g(\vc,p),\omega} \, r_\omega - \sum_{\omega \in \Omega} F_{g(\vc,p),\omega} \,p_\omega \right]- \beta -\eta.
	\]
	%
\end{restatable}

Finally, we are ready to prove the following theorem, which is our main result in this section.
The theorem establishes the guarantees provided by our template, showing that the computed contract $\widetilde{\Gamma} \in \cX$ of class $\cX$ achieves a regret that is at most $2 \sqrt{\epsilon}$ greater than the minimum possible regret attainable by contracts of class $\cX$.
Formally:
\begin{restatable}{theorem}{theoremTemplateOne}\label{thm:regret_template}
	Given a PAPU instance $(\Omega,A,\rvec,\mcC, \Fvec)$, a class of contracts $\cX$, and a parameter $\epsilon \in [0,  1 ]$ of the $\epsilon$-cover, Algorithm~\ref{alg:template} computes a contract $\widetilde \Gamma\in \cX$ of class~$\cX$ that achieves regret:
	\[\mcR(\widetilde{\Gamma}) = \sup_{\vc \in \cC} \Big\{  \OPT(\vc)-\Up(\widetilde \Gamma,\vc) \Big\}  \le \inf_{\Gamma' \in \cX}  \sup_{\vc \in \cC} \Big\{ \OPT(\vc)-\Up(\Gamma',\vc) \Big\} +2\sqrt{2\epsilon}.\]
	Moreover, Algorithm~\ref{alg:template} requires time polynomial in the size of the input provided that the procedures \texttt{Build-Cover}$(\cC,\epsilon)$ and \texttt{Compute-RM-Contract}$(\cX,{\cC_\epsilon})$ admit a polynomial-time implementation.
\end{restatable}
Theorem~\ref{thm:regret_template} shows that our template in Algorithm~\ref{alg:template} is guaranteed to compute an approximately-regret-minimizing contract, as desired.
In the following Section~\ref{sec:apply_template}, we show how the unspecified components of Algorithm~\ref{alg:template} can be implemented efficiently in some specific settings, proving that the template can be instantiated to design efficient approximation algorithms in those settings.

\section{How to Instantiate the Template}\label{sec:apply_template}

In this section, we show how our template (\cref{alg:template} in Section~\ref{sec:beyond}) can be instantiated in different settings, characterized by particular problem structures and specific classes of contracts.

First, in Section~\ref{sec:eps_cover}, we introduce some subclasses of PAPU instances in which it is possible to build an $\epsilon$-cover of the uncertainty set (see Definition~\ref{def:eps_cover}) having ``small'' size.
Then, in Section~\ref{sec:known_br}, we provide a crucial result needed to instantiate our template.
Finally, Sections~\ref{sec:res_menus_rand}~and~\ref{sec:res_others} exploit the results derived in the preceding subsections in order to instantiate the template in several settings, so as to derive our approximation algorithms to find regret-minimizing contracts.

At a high level, the results in Section~\ref{sec:eps_cover} show how to implement the $\texttt{Build-Cover}(\mcC,\epsilon)$ procedure in \cref{alg:template} for some particular subclasses of PAPU instances, while the results in Section~\ref{sec:known_br} show how to implement the procedure $\texttt{Compute-RM-Contract}(\cX,\mcC_\epsilon)$.
Sections~\ref{sec:res_menus_rand}~and~\ref{sec:res_others} show how to combine everything, by considering every possible class of contracts $\cX$.

\subsection{When Does a ``Small'' $\epsilon$-Cover Exist?}\label{sec:eps_cover}

In the following, we identify three particular subclasses of PAPU instances of interest, in which it is possible to build a finite $\epsilon$-cover of the uncertainty set $\mcC$ with suitably ``small'' size.
The following three definitions formally introduce such three subclasses of PAPU instances.
\begin{definition}[Finite PAPU instance]
	We say that a PAPU instance $(\Omega, A, \rvec, \mcC, \Fvec)$ is \emph{finite} whenever the uncertainty set $\mcC$ contains a finite number of cost vectors.
\end{definition} 
\begin{definition}[Single-dimensional PAPU instance]
	We say that a PAPU instance $(\Omega, A, \rvec, \mcC, \Fvec)$ is \emph{single-dimensional} whenever there exist a ``central'' cost vector $\vc^\circ \in \mcC$ and a range parameter $\rho \in \mathbb{R}_{>0}$ such that the uncertainty set $\mcC$ can be expressed as $\mcC = \left\{ \lambda \vc^\circ \mid \lambda \in [1-\rho, 1+\rho]  \right\}$.
\end{definition}
Intuitively, in a single-dimensional PAPU instance, all the cost vectors in the uncertainty set $\mcC$ belong to a $1$-dimensional subspace.
In particular, every cost vector $\vc \in \mcC$ can be expressed as a scaling of a particular ``central'' cost vector $\vc^\circ$, with the magnitude of the scaling $\lambda$ being in the range $[1-\rho, 1+\rho]$.
Notice that single-dimensional PAPU instances encompass principal-agent settings in which agent's costs are expressed on a cost-per-unit-of-effort basis~\citep{alon2021contracts}.
\begin{definition}[$L_p$-uncertain PAPU instance]
	For $p \geq 1$, we say that a PAPU instance $(\Omega, A, \rvec, \mcC, \Fvec)$ is \emph{$L_p$-uncertain} whenever the uncertainty set $\mcC$ can be expressed as an $L_p$-norm ball centered in a cost vector $\vc_0 \in [0,1]^n$ with radius $D \in [0,1]$, namely $\mcC = \left\{ \vc \in [0,1]^n \mid \| \vc - \vc_0 \|_p \leq D  \right\}$.
\end{definition}
Notice that, as also discussed in Section~\ref{sec:sub_dep_uncertain}, an uncertainty set having the shape of an $L_p$-norm ball may arise in all settings in which agent's costs are estimated from data, in order to build some confidence region in which true agent's costs are contained with sufficiently high probability.
%
%

Now, we show that the three subclasses of PAPU instances introduced above admit finite $\epsilon$-covers with ``small'' size.
%
First, notice that, in any finite PAPU instance, the uncertainty set $\mathcal{C}$ trivially serves as an $\epsilon$-cover of itself for any $\epsilon \in [0,1]$. Thus:
%
\begin{observation}\label{obs:finite}
	In any finite PAPU instance, the uncertainty set $\cC$ is a $0$-cover of $\mcC$. 
\end{observation}

Next, we show that any single-dimensional PAPU instance admits a finite $\epsilon$-cover of size $O(\nicefrac 1 \epsilon)$, which is obtained by constructing a discretization of the range $[1-\rho, 1+\rho]$ of the scaling parameter $\lambda$ into $\lceil \nicefrac {2 \rho} \epsilon \rceil + 1$ points.
This is made formal by the following lemma:
\begin{restatable}{lemma}{lemmaCoverOne}\label{lem_cover_one}
	In any single-dimensional PAPU instance, given any $\epsilon \in (0,1]$, there exists a finite $\epsilon$-cover $\mcC_\epsilon$ of $\mcC$  with cardinality $O(\nicefrac 1 \epsilon)$.
	%
\end{restatable}

Finally, we show that any $L_p$-uncertain PAPU instance admits a finite $\epsilon$-cover of size $O((\nicefrac 1 \epsilon)^n)$, where we recall that $n$ is the number of agent's actions.
Intuitively, this is obtained by constructing a discretization of the hypercube $[0,1]^n$, by using $\lceil \nicefrac {1} \epsilon \rceil + 1$ points along each dimension.
\begin{restatable}{lemma}{lemmaCoverTwo}\label{lem_cover_two}
	In any $L_p$-uncertain PAPU instance, given any $\epsilon \in (0,1]$, there exists a finite $\epsilon$-cover $\mcC_\epsilon$ of $\mcC$ with cardinality $O((\nicefrac 1 \epsilon)^n)$.
\end{restatable}

Observation~\ref{obs:finite}, Lemma~\ref{lem_cover_one}, and Lemma~\ref{lem_cover_two} also give ways of implementing the $\texttt{Build-Cover}(\mcC,\epsilon)$ procedure in \cref{alg:template} for finite, single-dimensional, and $L_p$-uncertain PAPU instances, respectively.
For finite and single-dimensional instances, this requires time that scales polynomially in the size of the PAPU instance and $\nicefrac 1 \epsilon$, where $\epsilon$ is the parameter of the $\epsilon$-cover.
For $L_p$-uncertain instances, the running time is polynomial only when the number of agent's actions $n$ is fixed.

\subsection{Knowing Agents' Behavior is Enough}\label{sec:known_br}

In the following, we derive a crucial result (Theorem~\ref{thm:computeGeneral}) that allows us to instantiate our template.
Such a result jointly encompasses all the possible classes of contracts, namely deterministic contracts, randomized contracts, menus of the deterministic contracts, and menus of randomized contracts.
In words, the result can be summarized as follows: \emph{A contract which minimizes the regret against a finite number of agent's cost vectors can be computed efficiently once agents' behavior is fixed.}

In order to formally state the result, we need to introduce some additional notation.
First, given a class of contracts $\cX$, $K \in \mathbb{N}_{>0}$, and $W \in \mathbb{N}_{>0}$, we let $\cX_{K,W}$ be the subclass of $\cX$ made by all contracts in which menus have size at most $K$ and randomized contracts have at most $W$ different elements in their supports.
Notice that, whenever menus, respectively randomized contracts, are \emph{not} involved in the class $\cX$, then it must be the case that $K = 1$, respectively $W = 1$.
In the following, for ease of presentation, we denote contracts in $\cX_{K,W}$ as $\Gamma \in \cX_{K,W}$, by assuming that they are always expressed as menus of randomized contracts, similarly to what done in Section~\ref{sec:template}.

Moreover, in order to encode agents' behavior, we introduce two functions which are similar in spirit to those in the definition of $\eta$-incentive compatibility (Definition~\ref{def:apxIC}).
In particular, given a \emph{finite} uncertainty set $\widetilde{\mcC} \subseteq [0,1]^n$, and two parameters $K \in \mathbb{N}_{>0}$ and $W \in \mathbb{N}_{>0}$, we define:
\begin{itemize}
	\item[(1)] $t: \widetilde{\mcC} \to [K]$, which is a function that specifies how agents select randomized contracts from menus. In particular, an agent with cost vector $\vc \in \widetilde{ \mcC}$ selects the element with index $t(\vc)$.
	\item[(2)] $g : \widetilde{\mcC} \times [W] \to A$, which is a function that specifies how agents play actions. In particular, when the deterministic contract proposed by the principal is the $j$-th one in the support of the $t(\vc)$-th element of the menu, an agent with cost vector $\vc \in \widetilde{\mcC}$ plays action $g(\vc,j)$.
\end{itemize}

Finally, the following Problem~\eqref{pr:general} encodes the problem of computing a contract in  $\cX_{K,W}$ which minimizes the regret against agents with cost vectors in a \emph{finite} uncertainty set $\widetilde{\mcC} \subseteq [0,1]^n$, when agents' behavior is fixed as specified by two given functions $t: \widetilde{\mcC} \to [K]$ and $g : \widetilde{\mcC} \times [W] \to A$.
\begin{subequations}\label{pr:general}
	\begin{align}
		\inf_{\gamma, p}& \quad \max_{\vc \in \widetilde{\mcC}} \,\, \OPT(\vc)- \sum_{j \in [W]} \gamma^{t(\vc)}_j \sum_{\omega \in \omega} F_{g(\vc,j),\omega} \left( r_{\omega}-p^{t(\vc)}_{j,\omega} \right)  \quad\quad \textnormal{s.t.} \hspace{5cm}\,\label{obj} \\
		& \specialcell{\sum_{j \in [W]} \gamma^{t(\vc)}_j\sum_{\omega \in \Omega}  F_{g(\vc,j), \omega} \left( p^{t(\vc)}_{j,\omega} - c_{g(\vc,j)} \right)\ge \sum_{j \in [W]} \gamma^{i}_j  \max_{a \in A} \,  \sum_{\omega \in \Omega}  F_{a,\omega} \left( p^{i}_{j,\omega} - c_{a} \right) \hfill \forall \vc \in \widetilde{\cC}, \forall i \in [K]} \label{cons}\\
		& \specialcell{\sum_{j \in [W]}\gamma^i_j=1 \hfill \forall i \in [K] } \\
		& \specialcell{\gamma^i_j \geq 0 \hfill \forall i \in [K], \forall j \in [W] } \\
		& \specialcell{p^i_{j,\omega} \geq 0 \hfill \forall i \in [K], \forall j \in [W], \forall \omega \in \Omega. }
	\end{align}
\end{subequations}
In Problem~\eqref{pr:general}, variables $\gamma_j^i$ encode the probabilities of randomized contracts, with $\gamma_j^i$ being the probability that the $i$-th randomized contract in the menu assigns to the $j$-th deterministic contract in its support.
Moreover, variables $p^i_{j,\omega}$ encode payments of deterministic contracts, with $p^i_{j,\omega}$ being the payment that the $j$-th deterministic contract in the support of the $i$-th randomized contract assigns to outcome $\omega$.
Moreover, Objective~\eqref{obj} encodes the worst-case value, over all the cost vectors in $\vc \in \widetilde{ \mcC}$, of the difference between $\OPT(\vc)$ and principal's expected utility when agents' behavior is fixed as specified by $t$ and $g$.
Finally, Constraints~\eqref{cons} ensure that agents actually select the randomized contracts specified by $t$ and play the actions defined by $g$ as best response.

Now, we are ready to prove our crucial result:
%
%
\begin{restatable}{theorem}{theoremEfficiently}\label{thm:computeGeneral}
	Given a PAPU instance $(\Omega, A, \rvec, \mcC, \Fvec)$, a class of contracts $\cX$, a finite subset $\widetilde{\cC} \subseteq \mcC$, two parameters $K \in \mathbb{N}_{>0}$ and $W \in \mathbb{N}_{>0}$, and two functions $t: \widetilde{\mcC} \to [K]$ and $g : \widetilde{\mcC} \times [W] \to A$, there is an algorithm solving Problem~\eqref{pr:general} in time polynomial in the instance size, the size of $\widetilde{\cC}$, $K$, and $W$.
	%
\end{restatable}
%
Theorem~\ref{thm:computeGeneral} provides an efficient way of implementing the $\texttt{Compute-RM-Contract}(\cX,\mcC_\epsilon)$ procedure in Algorithm~\ref{alg:template}.
In the following Sections~\ref{sec:res_menus_rand}~and~\ref{sec:res_others}, we show how to actually use it to instantiate our template.
Notice that, as a byproduct of Theorem~\ref{thm:computeGeneral}, we also get that Problem~\eqref{pr:general} admits an optimal solution (\emph{i.e.}, it admits a minimum) whenever $\widetilde{\cC}$ has finite size.
As we show in the following subsection, this allows us to derive an existence result for regret-minimizing menus of randomized contracts, \emph{when the size of the uncertainty set $\mcC$ is finite}.
This is in contrast with what usually happens with menus of randomized contracts~\citep{gan2022optimal,castiglioni2023designing}.

\subsection{Results on Menus of Randomized Contracts}\label{sec:res_menus_rand}

Theorem~\ref{thm:computeGeneral} provides the last missing piece we need to fully instantiate our template (\cref{alg:template} in Section~\ref{sec:beyond}).
In the following, we focus on the case in which $\cX$ is the class of menus of randomized contracts.
Then, in Section~\ref{sec:res_others}, we address the other classes of contracts.

First, we exploit Theorem~\ref{thm:computeGeneral} to derive the following result.
Intuitively, it states that, given a \emph{finite} subset $\widetilde{\mcC} \subseteq \mcC$ of cost vectors, a menu of randomized contracts which minimizes the regret against agents with cost vectors in $\widetilde{\mcC} $ can be computed efficiently.
Formally:
%
%
\begin{restatable}{theorem}{thmRandomizedMenus}\label{thm:randomized}
	Given a PAPU instance $(\Omega, A, \rvec, \mcC, \Fvec)$ and a finite subset $\widetilde{\cC} \subseteq \mcC$ of cost vectors, there is an algorithm that computes a menu of randomized contracts $\widetilde{\Gamma} \in \mcMr$ such that $\mcR(\widetilde{\Gamma}) = \inf_{\Gamma \in \mcMr} \mcR(\Gamma)$ in time polynomial in the size of the instance and the size of $\widetilde{\cC}$.
	%
\end{restatable}
Theorem~\ref{thm:randomized} follows from revelation-principle-style arguments similar to the ones employed in~\citep{gan2022optimal}. This shows that we can focus on direct menus of size $K=|\widetilde\cC|$ (where each element corresponds one-to-one to agent's cost vectors), in which each randomized contract is supported on at most $W=|A|$ deterministic ones.
%
%
%
As a byproduct of Theorem~\ref{thm:randomized}, we also get that, when the uncertainty set $\mcC$ is \emph{finite}, a regret-minimizing menu of randomized contracts exists.
%

Theorem~\ref{thm:randomized}, together with Observation~\ref{obs:finite} and Lemmas~\ref{lem_cover_one}~and~\ref{lem_cover_two}, immediately gives the following results, which provide efficient approximation algorithms for computing regret-minimizing menus of randomized contracts, for the PAPU instances introduced in Section~\ref{sec:eps_cover}.
In particular, there exists an exact algorithm for finite instances, an \textnormal{\textsf{FPTAS}} for single-dimensional instances, and, for $L_p$-uncertain instances, a \textnormal{\textsf{PTAS}} that becomes an \textnormal{\textsf{FPTAS}} when the number of actions is constant. 

\begin{corollary}
	Given a finite PAPU instance, there exists an algorithm that computes in polynomial time a menu of randomized contracts $\widetilde{\Gamma} \in \mcMr$ such that $\mcR(\widetilde{\Gamma}) = \inf_{\Gamma \in \mcMr} \mcR(\Gamma) $.
\end{corollary}

\begin{corollary}
	Given a single-dimensional PAPU instance and any $\epsilon \in (0, 1]$, there exists an algorithm that computes a menu of randomized contracts $\widetilde{\Gamma} \in \mcMr$ with $\mcR(\widetilde{\Gamma}) \leq \inf_{\Gamma \in \mcMr} \mcR(\Gamma) + 2 \sqrt{2 \epsilon}$ in time polynomial in the size of the PAPU instance and $\nicefrac 1 \epsilon$.
	%
\end{corollary}

\begin{corollary}
	Given an $L_p$-uncertain PAPU instance and any $\epsilon \in (0,  1 ]$, there exists an algorithm that computes a menu of randomized contracts $\widetilde{\Gamma} \in \mcMr$ with $\mcR(\widetilde{\Gamma}) \leq \inf_{\Gamma \in \mcMr} \mcR(\Gamma) + 2 \sqrt{2 \epsilon}$ in time polynomial in $(\nicefrac 1 \epsilon)^n$.
	Moreover, when the number of agent's actions $n$ is fixed, the algorithm requires time polynomial in the size of the PAPU instance and $\nicefrac 1 \epsilon$.
	%
\end{corollary}

\subsection{Results on Other Classes of Contracts}\label{sec:res_others}

Finally, we show how to exploit \cref{thm:computeGeneral} to instantiate our template (\cref{alg:template} in Section~\ref{sec:beyond}) for the classes of deterministic contracts, randomized contracts, and menus of deterministic contracts.

First, we provide the following result, which is the counterpart of Theorem~\ref{thm:randomized} for all the classes of contracts other than menus of randomized contracts.
%
%
\begin{restatable}{theorem}{thmOtherCntracts}\label{thm:others}
	Let $\cX$ be any class of contracts among $\mcP$, $\mcPr$, and $\mcM$.
	Given a PAPU instance $(\Omega, A, \rvec, \mcC, \Fvec)$ and a finite subset $\widetilde{\cC} \subseteq \mcC$ of cost vectors, there exists an algorithm that computes a contract $\widetilde{\Gamma} \in \cX$ of class $\cX$ such that $\mcR(\widetilde{\Gamma}) = \inf_{\Gamma \in \cX} \mcR(\Gamma)$ in time polynomial in the size of the PAPU instance, when the size of the set $\widetilde{\cC}$ is constant.
\end{restatable}

Differently from \Cref{thm:randomized}, in which functions $t$ and $g$ can be efficiently determined by relying on direct menus, Theorem~\ref{thm:others} has to adopt a ``brute force approach''.
%
Indeed, it needs to enumerate all the possible functions $t$ and $g$, and, thus, it needs a set $\widetilde \cC$ of constant size. 

Theorem~\ref{thm:others}, together with Observation~\ref{obs:finite} and Lemmas~\ref{lem_cover_one}~and~\ref{lem_cover_two}, yields the following results.
In particular, there exists an exact algorithm when the size of $\cC$ is constant, a \textnormal{\textsf{PTAS}} for single-dimensional instances, and a \textnormal{\textsf{PTAS}} for $L_p$-uncertain instances when $n$ is constant.

\begin{corollary}
	Given a finite PAPU instance, there exists an algorithm that computes a contract $\widetilde{\Gamma} \in \cX$ of class $\cX$ such that $\mcR(\widetilde{\Gamma}) = \inf_{\Gamma \in \cX} \mcR(\Gamma) $ in time polynomial in the size of the PAPU instance, when the size of the set $\cC$ is constant.
	%
\end{corollary}

\begin{corollary}
	Given a single-dimensional PAPU instance and any $\epsilon \in (0,  1]$, there exists an algorithm that computes a contract $\widetilde{\Gamma} \in \cX$ of class $\cX$ such that $\mcR(\widetilde{\Gamma}) \leq \inf_{\Gamma \in \cX} \mcR(\Gamma) + 2 \sqrt{2 \epsilon}$ in time polynomial in the size of the PAPU instance, when the parameter $\epsilon$ is constant.
	%
\end{corollary}

\begin{corollary}
	Given an $L_p$-uncertain PAPU instance and any $\epsilon \in (0,  1 ]$, there exists an algorithm that computes a contract $\widetilde{\Gamma} \in \cX$ such that $\mcR(\widetilde{\Gamma}) = \inf_{\Gamma \in \cX} \mcR(\Gamma) + 2 \sqrt{2 \epsilon}$ in time polynomial in the size of the PAPU instance, when both $n$ and $\epsilon$ are constant.
	%
\end{corollary}

\section*{Acknowledgements}
This paper is supported by the Italian MIUR PRIN 2022 Project “Targeted Learning Dynamics: Computing Efficient and Fair Equilibria through No-Regret Algorithms”, by the FAIR (Future Artificial Intelligence Research) project, funded by the NextGenerationEU program within the PNRR-PE-AI scheme (M4C2, Investment 1.3, Line on Artificial Intelligence), and by the EU Horizon project ELIAS (European Lighthouse of AI for Sustainability, No. 101120237).

\bibliographystyle{ACM-Reference-Format}
\bibliography{biblio}

\newpage
\appendix
\section*{Appendix}

The Appendix is organized as follows:
\begin{itemize}
	\item Appendix~\ref{sec:app_existence} provides the proofs of all the results in Section~\ref{sec:existence}.
	\item Appendix~\ref{sec:app_worst} provides the proofs of all the results in Section~\ref{sec:worstCase}.
	\item Appendix~\ref{sec:app_robust} provides the proofs of all the results in Section~\ref{sec:robustness}.
	\item Appendix~\ref{app:beyond} provides the proofs of all the results in Section~\ref{sec:beyond}.
	\item Appendix~\ref{sec:app_template} provides the proofs of all the results in Section~\ref{sec:apply_template}.
\end{itemize}

\section{Proofs Omitted From Section~\ref{sec:existence}}\label{sec:app_existence}

\theoremExistenceOne*

\begin{proof}
	For ease of presentation, let us consider the function $g : \mcP \times \mcC \to \mathbb{R}$, defined as follows:
	\[
	g(\vp, \vc):=\OPT(\vc)-\Up(\vp,\vc) \quad \forall \vp \in \mcP, \forall \vc \in \mcC.
	\]
	For every $\vc \in \mcC$, the function $\vp\mapsto g(\vp,\vc)$ obtained by fixing $\vc$ in $g$ is \emph{lower semi-continuous} (LSC). Indeed, the first term $\OPT(\vc)$ is independent of $\vp$, and, for every $\vc \in \mcC$, it is easy to see that the function $\vp\mapsto \Up(\vp,\vc)$ is upper semi-continuous, since ties are broken in favor of the principal.
	
	%
	Since the supremum of any collection of LSC functions is LSC, the function $\vp\mapsto\sup_{\vc\in\cC}g(\vp,\vc)$ is LSC.
	Then, by definition of regret, we immediately get that the function $\vp\mapsto\reg(\vp)$ is LSC.

	Notice that, if the set of deterministic contracts $\mcM$ were compact, the result obtained above would immediately imply the existence of a menu of deterministic contracts $\Pi^\star \in \mcM$ of size $K$ attaining regret $\mcR(\Pi^\star) = \inf_{\Pi \in \mcM: |\Pi| = K} \mcR(\Pi)$.
	However, the set $\Pi \in \mcM: |\Pi| = K$ is \emph{not} compact.
	Intuitively, this is because, in general, the deterministic contracts involved in a menu may place arbitrarily large payments on some outcomes, making the set $\Pi \in \mcM: |\Pi| = K$ unbounded.
	%
	
	Thus, in order to prove the theorem, we need to employ a different (less immediate) argument.
	In particular, we exploit the fact that a menu of deterministic contracts $\Pi^\star \in \mcM$ of size $K$ such that $\mcR(\Pi^\star)=\inf_{\Pi\in \mcM: |\Pi|=K} \mcR(\Pi)$ always belongs to a bounded subset of $\Pi \in \mcM: |\Pi| = K$.
	%

	W.l.o.g., let us assume that the PAPU instance is such that every outcome $\omega \in \Omega$ can be realized with strictly positive probability by at least one agent's action (otherwise the outcome can be simply removed from the instance).
	Then, it is easy to see that there exists a payment threshold $P > 0$ such that, if the principal commits to a deterministic contract $\vp \in \mcP$ placing a payment $p_\omega$ greater than $P$ on some outcome $\omega \in \Omega$, then principal's expected utility $\Up(\vp,\vc)$ is negative for every $\vc \in \mcC$.

	Thus, a menu $\Pi^\star \in \mcM$ of size $K$ such that $\mcR(\Pi^\star)=\inf_{\Pi\in \mcM: |\Pi|=K} \mcR(\Pi)$ must only use payments bounded by the threshold $P$ defined above.
	%
	%
	This immediately implies that we can restrict the regret-minimization problem over a compact subset of $\Pi \in \mcM: |\Pi| = K$, concluding the proof.
	%
	%
	%
	%
	%
	%
\end{proof}

\propositionExistenceOne*

\begin{proof}
	We show that, given any menu of deterministic contracts $\Pi\in \mcM$, we can build a menu $\Pi' \in \mcM$ of size $n$ with at most the same regret. This is sufficient to prove the proposition.
	
	Take any menu of deterministic contracts $\Pi = \left\{ \vp^1, \ldots, \vp^K \right\}\in \mcM$ of size $K \in \mathbb{N}_{>0}$.
	For ease of notation, for every $\vc \in \mcC$, let $\vp^{\vc} \in \Pi$ be the deterministic contract selected by an agent with cost vector $\vc$, \emph{i.e.}, one maximizing their expected utility (with ties broken in favor of the principal).

	Fix an action $a\in A$, and let $\cC_a := \{ \vc\in \cC \mid a^\star(\vp^{\vc},\cvec) = a \}$, \emph{i.e.}, the set of cost vectors such that the agent best responds by playing action $a$ under the chosen contract.
	%
	Then, since each agent with cost vector $\vc \in \mcC_a$ selects contract $\vp^{\vc}$ and plays action $a$ as best response, it is easy to see that there must exist an expected payment value $P_a > 0$ such that
	%
	$  \sum_ {\omega \in \Omega} F_{a,\omega} \, p^{\cvec}_\omega = P_a $ for all $ \cvec \in \cC_a$.
	%
	Hence, all the agents with cost vectors in $\cC_a$ receive the same expected payment and are indifferent among all the contracts $\vp^{\vc}$ with $\vc \in \mcC_a$.
	It is easy to check that, by keeping only one of such contracts (and removing all the others from $\Pi$), we obtain a new menu of deterministic contracts $\Pi' \in \mcM$ with the same principal's expected utility, namely $\Up(\Pi,\vc) = \Up(\Pi',\vc)$ for all $\vc \in \mcC$.
	%

	By applying the reasoning above for all the $n$ agent's actions $a \in A$, we obtain a menu $\Pi' \in \mcM$ of size $n$ such that $\Up(\Pi,\vc) = \Up(\pi',\vc)$ for every $\vc \in \mcC$.
	%
	By definition of regret, we immediately get to the result, concluding the proof.
\end{proof}


\propositionExistenceTwo*

\begin{proof}
	Let us consider the PAPU instance introduced in Definition~\ref{def:nonexist_inst}.
	We split the proof into two parts. First, we construct a sequence of randomized contracts $\left\{ \gvec^k \right\}_{k \in \mathbb{N}_{>0}}$, with $\gvec^k \in \mcPr$, having supports such that $|\supp(\gvec^k) | = k$ and regrets $\mcR(\gvec^k)$ such that $\lim_{k\to+\infty}\mcR(\gvec^k)=\nicefrac 1 e$. Then, we prove that, for every randomized contract $\gvec \in \mcPr$ with finite support, the regret $\mcR(\gvec)$ is strictly more than $\nicefrac 1 e$.
	This shows that there is no $\gvec \in \mcPr$ whose regret achieves the value $\inf_{\gvec \in \mcPr} \mcR(\gvec)$.
	
	\subsection*{Part I: \textnormal{\emph{Sequence of randomized contracts $\left\{ \gvec^k \right\}_{k \in \mathbb{N}_{>0}}$ with $\lim_{k\to+\infty}\mcR(\gvec^k)=\nicefrac 1 e$.}}} 
	
	%
	We construct a sequence of randomized contracts $\left\{ \gvec^k \right\}_{k \in \mathbb{N}_{>0}}$, with $\gvec^k \in \mcPr$, having supports such that $|\supp(\gvec^k) | = k$ and regrets that decrease as $k \to +\infty$. In particular, we show that the limit as $k \to +\infty$ of the regret is $\lim_{k\to+\infty}\mcR(\gvec^k)=\nicefrac 1 e$, by proving that $\mcR(\gvec^k)\le \nicefrac 1 e+O(\nicefrac 1 k)$ for all $k \in \mathbb{N}_{>0}$.
	
	First, let us define the elements $\gvec^k$ of the sequence.
	For all $k\in\mathbb{N}_{>0}$, we consider the randomized contract $\gvec^k \in \mcPr$ with support $\supp(\gvec^k) := \left\{ \vp^k_1,\ldots, \vp_k^k \right\}$, where:
	\[
	\vp_i^k=(0, \alpha_i^k):=\left(0, \left(\frac i{k}-\frac{1}{2k}\right)\left(1-\frac1e\right)\right).
	\]
	Moreover, each contract $\vp_i^k$ has probability:
	\[
	\gamma^k_{\vp_i^k} :=\int_{\alpha_i^k-\Delta_k}^{\alpha_i^k+\Delta_k}\frac{1}{1-x}dx=\ln \left(\frac{1-\alpha_i^k+\Delta_k}{1-\alpha_i^k-\Delta_k}\right),
	\]
	where, for ease of notation, we let $\Delta_k=\frac{1}{2k}\left(1-\frac1e\right)$.
	
	Notice that $\gvec^k$ is a well-defined probability distribution. Indeed, we have that:
	\begin{align*}
		\sum\limits_{i=1}^k\gamma^k_{\vp_i^k}&=\sum\limits_{i=1}^k\ln \left(\frac{1-\alpha_i^k+\Delta_k}{1-\alpha_i^k-\Delta_k}\right)\\
		&=\ln\left(\prod\limits_{i=1}^k\frac{1-\alpha_i^k+\Delta_k}{1-\alpha_i^k-\Delta_k}\right)\\
		&=\ln\left(\prod\limits_{i=1}^{k-1}\frac{1-\alpha_{i}^k+\Delta_k}{1-\alpha_{i+1}^k+\Delta_k} \cdot \frac{1-\alpha_k^k+\Delta_k}{1-\alpha_k^k-\Delta_k}\right)\\
		&=\ln\left(\frac{1-\alpha_1^k+\Delta_k}{1-\alpha_k^k-\Delta_k}\right)\\
		&=\ln(e)\\&=1,
	\end{align*}
	where in the third equality we used the fact that $\alpha_{i+1}^k=\alpha_{i}^k+2\Delta_k$ by construction, while the fourth equality is obtained by unfolding the product.
	
	Now, let us observe that, for the considered instance, any randomized contract $\gvec^k\in\mcPr$ satisfies:
	\[
	\OPT(\vc)-\Up(\gvec^k,\vc)=1-c_{a_2}-\sum\limits_{i \in [k]:\alpha_i^k\ge c_{a_2}}\gamma^k_{\vp_i^k}(1-\alpha_i^k) \quad \forall \vc \in \mcC,
	\]
	where $c_{a_2}$ is the unknown cost of action $a_2$. 
	In particular, the quantity above linearly decreases in $c_{a_2}$ until $c_{a_2}$ reaches value $\alpha_i^k$.
	At that point, the quantity above jumps down by $J_i^k\coloneqq\gamma^k_{\vp_i^k}(1-\alpha_i^k)$.
	Moreover, for $c_{a_2}=0$, the quantity above has value $E^k_0:=\sum_{i=1}^k \alpha_i^k \, \gamma^k_{\vp_i^k}$.
	Thus, we can write:
	\begin{equation}\label{eq:lem_decomposition}
		\OPT(\vc)-U(\gvec_k,\vc)=E_0^k+\sum\limits_{i \in [k]:\alpha_i^k< c_{a_2}}J_i^k-c_{a_2}.
	\end{equation}
	
	In order to conclude the first part of the proof, we need to prove the following two auxiliary results, which hold for all $k \in \mathbb{N}_{>0}$:
	\begin{itemize}
		\item $E^k_0\le\frac1e$; and
		\item $J_i^k \le 2\Delta_k+O\left( \frac 1 {k^2} \right)$ for all $i \in [k]$.
	\end{itemize}
	%
	
	
	%
	%
	%
	Then, we upper bound the value of $E_0^k$:
	\begin{align*}
		E_0^k&\coloneqq\sum\limits_{i=1}^k \alpha_i^k \, \gamma^k_{\vp_i^k} \\
		&=\sum\limits_{i=1}^k \alpha_i^k\int\limits_{\alpha_i^k-\Delta_k}^{\alpha_i^k+\Delta_k}\frac{1}{1-x}dx\\
		&\le\sum\limits_{i=1}^k \int\limits_{\alpha_i^k-\Delta_k}^{\alpha_i^k+\Delta_k}\frac{x}{1-x}dx\\
		&\le \int\limits_{0}^{1-\nicefrac 1 e}\frac{x}{1-x}dx\\
		&=\frac 1 e.
	\end{align*}
	
	%
	Next, we provide an upper bound on the value of $J_i^k$ for all $i \in [k]$:
	\begin{align*}
		J_i^k&\coloneqq\gamma^k_{\vp_i^k}(1-\alpha_i^k)\\
		&=\frac{e(1+2k)-1-2i(e-1)}{2ek}\ln\left(1+\frac{e-1}{ek-i(e-1)}\right)\\
		&\le \frac{e(1+2k)-1-2i(e-1)}{2ek}\frac{e-1}{ek-i(e-1)}\frac{1}{\sqrt{\frac{e-1}{ek-i(e-1)}+1}}\\
		&\le \Delta_k \frac{2k+e-1}{\sqrt{k(k-1+e)}}\\
		&\le \Delta_k\left(2+\frac{1}{k^2}\right),
	\end{align*}
	where, in the first inequality, we used that $\ln(1+x)\le \frac{x}{\sqrt{1+x}}$, which hold for all $x>0$.

	Now, we are ready to conclude the first part of the proof.

	By combining $E_0^k\le \nicefrac 1 e$ and $J_i^k\le \Delta_k\left(2+\nicefrac{1}{k^2}\right)$ for all $i \in [k]$, we show that the function defined as $\vc\mapsto \OPT(\vc)-\Up(\gvec^k, \vc)$ has value at most $\nicefrac 1 e+O(\Delta_k)$.
	We prove this by induction. Let us take any $\cvec \in \mcC$. Then, there exists an $m \in [k-1]$ such that $c_{a_2}\in(\alpha_m^k, \alpha_{m+1}^k]$ and, for every other $\cvec^{m} \in \mcC$ such that $c^{m}_{a_2}\in(\alpha_m^k, \alpha_{m+1}^k]$, the following holds:
	\[
	 \OPT(\vc^m)-\Up(\gvec^k, \vc^m)\le \OPT(\vc^k_m)-\Up(\gvec^k, \vc^k_{m+1})\coloneqq F_m,
	\]
	where we denote by $\vc^k_m \in \mcC$ the cost vector in which the cost of action $a_2$ is equal to $\alpha^k_{m}$.
	We work by induction on $F_m$. 
	First, let us observe that $c_{a_2}^{m}-c_{a_2}^{m-1}=2\Delta_k$. Thus:
	\begin{align*}
		F_m &= F_{m-1}+J_m^k-2\Delta_k\\
		&\le F_{m-1}+\frac{\Delta_k}{k^2}.
	\end{align*}
	Moreover, as base case of the induction, we use that:
	\begin{align*}
		F_1&=E_0^k+J_1^k-2\Delta_k\\
		&\le \frac1e+\frac{\Delta_k}{k^2}.
	\end{align*}
	In conclusion, we have that:
	\begin{align*}
		\sup\limits_{\vc \in \mcC} \Big\{\OPT(\vc)-\Up(\gvec^k, \vc) \Big\}&\le \sup_{m\in[k]}F_m\\
		&\le\frac1e+\frac{m}{k^2}\Delta_k\\
		&<\frac1e+\Delta_k.
	\end{align*}
	
	This proves that $\lim_{k\to\infty} \mcR(\gvec^k)\le \frac 1 e$.
		\begin{figure}[t]
		\begin{center}
			\tikzset{every picture/.style={line width=0.75pt}} 

\begin{tikzpicture}[x=0.75pt,y=0.75pt,yscale=-1,xscale=1]
	
	\draw    (190,290) -- (230,290) ;
	\draw [shift={(230,290)}, rotate = 180] [color={rgb, 255:red, 0; green, 0; blue, 0 }  ][line width=0.75]    (0,2.24) -- (0,-2.24)   ;
	\draw    (230,290) -- (290,290) ;
	\draw [shift={(290,290)}, rotate = 180] [color={rgb, 255:red, 0; green, 0; blue, 0 }  ][line width=0.75]    (0,2.24) -- (0,-2.24)   ;
	\draw    (290,290) -- (330,290) ;
	\draw [shift={(330,290)}, rotate = 180] [color={rgb, 255:red, 0; green, 0; blue, 0 }  ][line width=0.75]    (0,2.24) -- (0,-2.24)   ;
	\draw    (330,290) -- (367,290) ;
	\draw [shift={(370,290)}, rotate = 180] [fill={rgb, 255:red, 0; green, 0; blue, 0 }  ][line width=0.08]  [draw opacity=0] (8.93,-4.29) -- (0,0) -- (8.93,4.29) -- cycle    ;
	\draw    (190,290) -- (190,250) ;
	\draw [shift={(190,250)}, rotate = 90] [color={rgb, 255:red, 0; green, 0; blue, 0 }  ][line width=0.75]    (0,2.24) -- (0,-2.24)   ;
	\draw    (190,250) -- (190,210) ;
	\draw    (190,210) -- (190,190) ;
	\draw [shift={(190,190)}, rotate = 90] [color={rgb, 255:red, 0; green, 0; blue, 0 }  ][line width=0.75]    (0,2.24) -- (0,-2.24)   ;
	\draw    (190,190) -- (190,153) ;
	\draw [shift={(190,150)}, rotate = 90] [fill={rgb, 255:red, 0; green, 0; blue, 0 }  ][line width=0.08]  [draw opacity=0] (8.93,-4.29) -- (0,0) -- (8.93,4.29) -- cycle    ;
	\draw    (190,250) -- (230,290) ;
	\draw [shift={(230,290)}, rotate = 45] [color={rgb, 255:red, 0; green, 0; blue, 0 }  ][fill={rgb, 255:red, 0; green, 0; blue, 0 }  ][line width=0.75]      (0, 0) circle [x radius= 2.34, y radius= 2.34]   ;
	\draw    (230.95,190.95) -- (330,290) ;
	\draw [shift={(230,190)}, rotate = 45] [color={rgb, 255:red, 0; green, 0; blue, 0 }  ][line width=0.75]      (0, 0) circle [x radius= 2.34, y radius= 2.34]   ;
	\draw  [dash pattern={on 0.84pt off 2.51pt}]  (190,190) -- (230,190) ;
	\draw  [dash pattern={on 0.84pt off 2.51pt}]  (230,290) -- (230,190) ;
	
	\draw (189,251) node [anchor=east] [inner sep=0.75pt]  [font=\normalsize]  {$p$};
	\draw (232,293.4) node [anchor=north west][inner sep=0.75pt]    {$p$};
	\draw (188,193.4) node [anchor=north east] [inner sep=0.75pt]    {$1-p$};
	\draw (290,293.4) node [anchor=north] [inner sep=0.75pt]    {$1-\frac{1}{e}$};
	\draw (332,293.4) node [anchor=north west][inner sep=0.75pt]    {$1$};
	\draw (372,293.4) node [anchor=north west][inner sep=0.75pt]    {$c_{a_2}$};
	\draw (190,146.6) node [anchor=south] [inner sep=0.75pt]    {$\OPT(\vc)-\Up(p,\vc)$};

\end{tikzpicture}
		\end{center}
		\caption{Behavior of $\OPT(p)-\Up(p,\vc)$ as a function of $c_{a_2}$ in the instance described in \Cref{def:nonexist_inst}.}
		\label{fig:OPtp}
	\end{figure}
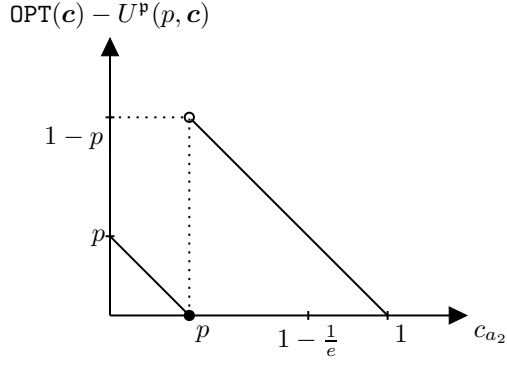
	\subsection*{Part II: \textnormal{\emph{Lower bound on the regret of randomized contracts with finite support.}}} 
	
	
	We show that all randomized contracts $\gvec \in \mcPr$ achieve regret $\mcR(\gvec)$ strictly greater than $\nicefrac 1 e$.
	
	First, notice that we can focus w.l.o.g. on randomized contracts that only include deterministic contracts placing a positive payment only on outcome $\omega_2$, since, as it is easy to see, using strictly positive payments on outcome $\omega_1$ is always sub-optimal.
	Thus, in the rest of the proof, we denote any deterministic contract $\vp \in \mcP$ as $p \in \mathbb{R}_+$, with the convention that $p_{\omega_1} = 0$ and $p_{\omega_2} = p$.
	
	
	As a first step, we show that there exists a probability distribution $\cD$ over $\mcC$ such that, for every $\gvec \in \mcPr$, it holds $\mathbb{E}_{\vc\sim \cD} \left[ \OPT(\vc)-\Up(\gvec,\vc) \right] \geq \frac 1 e$.
	Notice that the problem of minimizing the expectation above over all randomized contract can be cast as an instance of \emph{Bayesian} principal-agent problem, and, thus, the minimum value is achieved by means of deterministic contract $p \in \mathbb{R}_+$~\citep{castiglioni2023designing}.
	Hence, we can equivalently show that $\mathbb{E}_{\vc\sim \cD}[\OPT(p)-\Up(p,\vc)]\geq \frac  1 e$.
	%
	
	We design $\cD$ by following an approach similar to that employed by~\citet{babichenko2022regret} in related Bayesian persuasion settings.
	In particular, the cost vector $(0,0)$ is assigned probability~$\frac{1}{e}$, while the density function on the other cost vectors is such that, for every $\vc = (0,c_{a_2}) \in \mcC$ with $\vc \neq (0,0)$, it holds $\cD(\vc)=\frac{1}{e (1-c_{a_2})^2}$ whenever $c_{a_2} \in (0,1-\nicefrac 1 e]$, and $\cD(\vc)=0$ otherwise.
	Then, given a deterministic contract $p \in [0,1-\nicefrac 1 e]$, simple calculations show that $\mathbb{E}_{\vc\sim \cD}[\OPT(p)-\Up(p,\vc)]$ is:
	\begin{align*}
	\frac{p}{e}  + \int_{0}^{p} (p-c)\frac{1}{e(1-c)^2}  dc + \int_{p}^{1-1/e} (1-c)\frac{1}{e(1-c)^2}  dc & = \frac pe-\frac{p+\ln(1-p)}{e}+\frac{1+\ln(1-p)}{e}\\
	&=\frac 1 e.
\end{align*}
Moreover, for every deterministic contract $p>1-\nicefrac 1 e$, similar calculations allow us to show that $\mathbb{E}_{\vc\sim \cD}[\OPT(p)-\Up(p,\vc)]$ can be upper bounded as follows:
	\begin{align*}
	\frac{p}{e}  + \int_{0}^{1-1/e} (p-c)\frac{1}{e(1-c)^2}  dc & = \frac pe+\frac{2-e(1-p)-p}{e} = \frac 2e -1+p \ge\frac 1 e.
\end{align*}

	\Cref{fig:OPtp} illustrates the behavior of $\vc\mapsto\OPT(\vc)-\Up(\vp,\vc)$ for deterministic contracts $\vp$.
	
	As a result, we have that, for every $\vgamma \in \mcPr$, it holds $\mathbb{E}_{\vc\sim\cD}[\OPT(\vc)-\Up(\gvec,\vc)]\ge \frac 1 e$.
	This implies that either $\sup_{\vc \in \supp(\cD)} \left\{ \OPT(\vc)-\Up(\gvec,\vc) \right\}> \frac 1 e$ or $\OPT(\vc)-\Up(\gvec,\vc)= \frac 1 e$ for every $\vc \in \supp(\cD)$, where $\supp(\cD)$ is the support of the probability distribution $\cD$.
	Thus, in order to conclude the proof, it is sufficient to show that the second case never holds.
	To do that, we show that the function $\vc \mapsto \OPT(\vc)-\Up(\gvec,\vc)$ is \emph{not} constant over the interval $[0,1-\nicefrac 1 e]=\supp(\cD)$.
	This holds thanks to the fact that $\vc \mapsto \OPT(\vc)$ is a linear decreasing function, while the function $\vc \mapsto \Up(\gvec,\vc)$ is piece-wise constant.
	Indeed, such a function does \emph{not} change between any two cost vectors $\vc, \vc' \in \mcC$ such that $\supp(\gvec)$ does \emph{not} include any contract with $p_{\omega_2}\in [c_{a_2},c_{a_2}']$.
	This concludes the proof.
	%
	%
	%
\end{proof}
%
%
%
%

\propositionExistenceFour*

\begin{proof}
	Let us consider the PAPU instance introduced in Definition~\ref{def:nonexist_inst}.
	
	Let $\gvec \in \mcPr$ be a simple randomized contract that puts probability mass $\gamma_{\vp^1}=\gamma_{\vp^2}=\frac 1 2$ on two contracts $\vp^1=(0,\nicefrac 1 4)$ and $\vp^2=(0, \nicefrac 3 4)$.
	In the considered instance, for every $\vc \in \mcC$, it holds:
	\begin{align*}
		\OPT(\vc)-\Up(\gvec, \vc)&=1-c_{a_2}-\sum\limits_{i\in\{1,2\}:p^i_{\omega_2}\ge c_{a_2}}\gamma_{\vp^i} \left(1-p^i_{\omega_2} \right)\\
		&=1-c_{a_2}-\frac{1}{2}\sum\limits_{i\in\{1,2\}:p^i_{\omega_2}\ge c_{a_2}} \left( 1-p^i_{\omega_2} \right)\\
		&=\frac 1 2-c_{a_2}+\begin{cases}
			0&\text{if }c_{a_2}\in[0, \nicefrac 1 4] \\
			\frac 3 8&\text{if }c_{a_2}\in(\nicefrac 1 4,\nicefrac 3 4]\\
			\frac 1 2&\text{if }c_{a_2}\in(\nicefrac 3 4,1].
		\end{cases}
	\end{align*}
	
	Next, we show that the function obtained above does \emph{not} admit a maximum.
	Let us consider the cost vectors $\vc_\epsilon=(0,\nicefrac 1 4+\epsilon)$ for $\epsilon > 0$. We have that:
	\[
	\OPT(\vc_\epsilon)-\Up(\gvec, \vc_\epsilon) = \frac 5 8-\epsilon,
	\]
	for all $\epsilon>0$. This proves that
	\[\sup\limits_{\vc\in\cC} \Big\{ \OPT(\vc)-\Up(\gvec, \vc) \Big\} \ge \frac 5 8.\]
	
	Moreover, it holds $\OPT(\vc)-\Up(\gvec,\vc)< \frac 5 8$ for all $\vc=(0,c_{a_2})\in\cC$, concluding the proof.
\end{proof}
\section{Proofs Omitted from Section~\ref{sec:worstCase}}\label{sec:app_worst}


\lemmaWorstUnoRaw*

\begin{proof}
	%
	In the following, for ease of notation, let $\Delta c \coloneqq \| \cvec - \tilde \cvec \|_\infty$, which is the maximum absolute difference between the cost values $c_a$ and $\tilde c_a$ of an action $a \in A$.
	
	First, let $a^\star := a^\star(\pvec,\cvec)$ be the best response of an agent with cost vector $\vc$ under contract $\vp$.
	Then, we prove that $a^\star$ provides an agent with cost vector $\tilde \vc$ an expected utility that is at most $2 \Delta c$ worse than that of a best response.
	To see that, we notice that the following holds for every $a'\in A$:
	\begin{align}
		\sum_{\omega\in\Omega}F_{a^\star,\omega} p_\omega-\tilde c_{a^\star}&\ge\sum_{\omega\in\Omega}F_{a^\star,\omega}p_\omega- c_{a^\star} - \Delta c\nonumber\\
		&\ge \sum_{\omega\in\Omega}F_{a',\omega}p_\omega- c_{a'} - \Delta c\nonumber\\
		&\ge \sum_{\omega\in\Omega}F_{a',\omega}p_\omega-\tilde c_{a'} - 2 \Delta c, \label{eq:deltabr}
	\end{align}
	which proves that $\Ua (\pvec, \tilde{\cvec}, a^\star) \geq \Ua (\pvec, \tilde{\cvec}, a') - 2 \Delta c$ for all $a' \in A$.
	
	Now, let us consider the modified contract $\tilde \vp := \vp + \alpha (\rvec - \vp)$, which is obtained by adding a fraction $\alpha$ of principal's utilities to the payments defined by $\vp$.
	Formally, for every $\omega \in \Omega$:
	\begin{align*}
		\tilde p_\omega&:=p_\omega+\alpha (r_\omega-p_\omega)=(1-\alpha) p_\omega+\alpha r_\omega.
	\end{align*}	
	%
	Next, we show that, independently of the action chosen by an agent with cost vector $\tilde{\cvec}$, principal's expected utility is at least $\Up (\vp,\cvec) - \left( \frac{2 \Delta c}{\alpha} + \alpha \right)$.
 	First, notice that, for any $a \in A$ played by the agent, principal's expected utility under contract $\tilde\vp$ is a fraction $(1-\alpha)$ of the one under $\vp$. Formally:
	\begin{equation}\label{eq:scalingutils}
		\Up (\tilde \pvec, \tilde \cvec) = \sum_{\omega\in\Omega}F_{a,\omega} \left( r_\omega-\tilde p_\omega \right)=(1-\alpha)	\sum_{\omega\in\Omega}F_{a,\omega} \left( r_\omega-p_\omega \right).
	\end{equation}
	%
	Let $a := a^\star(\tilde \vp, \tilde \cvec)$ be the action played by an agent with cost vector $\tilde \vc$ under contract $\tilde \vp$.
	Then,
	\begin{align*}
		\sum_{\omega\in\Omega}F_{a,\omega}\tilde p_\omega-\tilde c_{a} & \ge \sum_{\omega\in\Omega}F_{a^\star,\omega}\tilde p_\omega-\tilde c_{a^\star}\tag{Action $a$ is a best response}\\
		&\ge(1-\alpha) \sum_{\omega\in\Omega}F_{a^\star,\omega}p_\omega+\alpha \sum_{\omega\in\Omega}F_{a^\star,\omega}r_\omega-\tilde c_{a^\star}\tag{Definition of $\tilde \vp$}\\
		%
		%
		&=\sum_{\omega\in\Omega}F_{a^\star,\omega}p_\omega-\tilde c_{a^\star}+\alpha \sum_{\omega\in\Omega}F_{a^\star,\omega} \left( r_\omega-p_\omega \right)\\
		&=\sum_{\omega\in\Omega}F_{a^\star,\omega}p_\omega-\tilde c_{a^\star}+\alpha \Up(\vp, \vc)\\
		&\ge \sum_{\omega\in\Omega}F_{a,\omega}p_\omega-\tilde c_{a}+\alpha \Up(\vp, \vc)-2 \Delta c\tag{\Cref{eq:deltabr}}.
	\end{align*}
	By using the definition of $\tilde \vp$ and the last inequality above, we obtain:
	\[
	(1-\alpha)\sum_{\omega\in\Omega}F_{a,\omega}p_\omega+\alpha\sum_{\omega\in\Omega}F_{a,\omega}r_\omega-\tilde c_{a}\ge \sum_{\omega\in\Omega}F_{a,\omega}p_\omega-\tilde c_{a}+\alpha \Up(\vp, \vc)-2 \Delta c,
	\]
	which simplifies to:
	\[
	\sum_{\omega\in\Omega}F_{a,\omega} \left( r_\omega-p_\omega \right) \ge \Up(\vp, \vc)-\frac{2 \Delta c}{\alpha}.
	\]
	Then, by using \Cref{eq:scalingutils}, we obtain the following:
	\[
	\Up(\tilde\vp,\tilde \vc)=(1-\alpha)\sum_{\omega\in\Omega}F_{a,\omega} \left( r_\omega-p_\omega \right)\ge (1-\alpha)\left(\Up(\vp, \vc)-\frac{2 \Delta c}{\alpha}\right)\ge \Up(\vp, \vc)-\left(\frac{2 \Delta c}{\alpha}+\alpha\right).
	\]
	Finally, by using the definition of $\Delta c$, we get to the result.
	%
\end{proof}

\lemmaWorstUno*

\begin{proof}
	The result immediately follows by applying Lemma~\ref{lem:delta_opt_raw} for $\vp = \vp^\star$ and $\alpha = \sqrt{2 d (\mcC)}$, by observing that $\Up (\vp^\star, \vc) = \OPT(\vc)$ by definition.
	Notice that the value of $\alpha$ is suitably selected in order to minimize the additive loss $\frac{2 \| \vc -\tilde \vc\|_\infty}{\alpha}+\alpha$ in the statement of Lemma~\ref{lem:delta_opt_raw}.
\end{proof}

\lemmaWorstDue*

\begin{proof}
	Let $p^\star \in \mcP$ be an optimal contract for the principal against an agent with cost vector $\vc$, \emph{i.e.}, one satisfying $ \Up(\vp^\star,\vc)=\OPT(\cvec)$.
	By applying \Cref{lem:delta_opt}, we have that the modified contract $\tilde \vp := \vp^\star + \sqrt{2 d(\mcC)} (\rvec - \vp^\star)$ satisfies the following:
	\[
	\Up(\tilde \vp,\tilde\vc)\ge \OPT(\vc)-2\sqrt{2 d(\mcC)}.
	\]
	Moreover, $\OPT(\tilde\vc)\ge \Up(\tilde \vp,\tilde\vc)$. Thus:
	\[
	\OPT(\tilde\vc)\ge\OPT(\vc)-2 \sqrt{2 d(\mcC)}.
	\]
	Finally, by an analogous (symmetric) argument. it clearly holds that:
	\[
	\OPT( \vc)\ge\OPT(\tilde \vc)-2 \sqrt{2 d(\mcC)},
	\]
	which proves the result.
\end{proof}

\theoremWorst*

\begin{proof}
	Let $\vc \in \mcC$ be any cost vector in the uncertainty set $\cC$ and $p^\star \in \mcP$ be an optimal contract for the principal against an agent with cost vector $ \vc$, \emph{i.e.}, one such that $ \Up(\vp^\star, \vc)=\OPT( \vc)$.
	Moreover, let $\tilde \vp := \vp^\star + \sqrt{2 d(\mcC)} (\rvec - \vp^\star)$ be the contract defined in \cref{lem:delta_opt}.
	%
	Then,
	\begin{align*}
		 \mcR(\tilde \vp) = \sup_{\tilde \cvec \in \mcC} \Big\{ \OPT(\tilde\cvec) - \Up(\tilde \vp,\tilde\cvec) \Big\} & \le  \sup_{ \tilde\cvec \in \mcC} \left\{ \OPT( \tilde\cvec) -  \Up(\vp^\star,  \vc)+2 \sqrt{2 d(\mcC)}\right\}  \\
		 &= \sup_{ \tilde\cvec \in \mcC} \left\{ \OPT(\tilde\cvec) - \OPT( \cvec) +2\sqrt{2 d(\mcC)}\right\}  \\
		 &\le  4\sqrt{2 d(\mcC)},  
	\end{align*}
	where the first inequality follows from \cref{lem:delta_opt} and the last one from \cref{lem:holder_opt}. 
	Thus, it clearly holds that $\inf_{\vp \in \mcP} \mcR(\vp)\le  4 \sqrt{2 d(\mcC)}$, proving the result.
\end{proof}

\propostionWorst*

\begin{proof}
	Consider one of the PAPU instances in Definition~\ref{def:hard}.
	%
	Notice that, since there are only two possible cost vectors, namely $(0,0)$ and $(\delta,0)$, by a revelation-principle-style argument we can focus w.l.o.g.~on menus specifying only two randomized contracts.
	Moreover, it is easy to see that the principal is always better off by committing to menus that include contracts placing a positive payment on outcome $\omega_1$ only.
	%
	Let $\Gamma = \left\{ \gvec^1, \gvec^2 \right\}$ be one of such menus, where $\gvec^1$ is the randomized contract selected by an agent with cost vector $(0,0)$ and $\gvec^2$ is the one chosen by an agent with cost vector $(\delta,0)$.
	Then, the expected utility that an agent with cost vector $(0,0)$ obtains by selecting $\gvec^1$ is equal to $\mathbb{E}_{\vp  \sim \gvec^1}[ p_{\omega_1} ]$, since playing action $a_1$ is a best response for every $\vp \in \mcP$ (as $p_{\omega_2} = 0$).
	%
	Moreover, the expected utility that the agent would get by selecting $\gvec^2$ and playing $a_1$ is $\mathbb{E}_{\vp\sim \gvec^2}[ p_{\omega_1} ]$.
	%
	Since an agent with cost vector $(0,0)$ is incentivized to select the randomized contract $\gvec^1$, rather than $\gvec^2$, it must be the case that $ \mathbb{E}_{\vp  \sim \gvec^1}[ p_{\omega_1} ]\ge \mathbb{E}_{\vp\sim \gvec^2}[ p_{\omega_1} ]$.
	Moreover, principal's expected reward is $1$ independently of the randomized contract selected by the agent.
	%
	Hence, by removing $\gvec^1$ from the menu, principal's expected utility can only increase.
	Thus, it is possible to focus w.l.o.g.~on the case in which the principal commits to a single randomized contract.

	

	Now, let us consider two different sets of contracts, namely $\cP_1$ and $\cP_2$, where
	\[
	\cP_1:=\left\{\vp \in \mcP \mid p_{\omega_1}\ge \delta/\alpha \wedge p_{\omega_2} = 0 \right\}
	\] 
	and 
	\[
	\cP_2:= \left\{ \vp \in \mcP \mid p_{\omega_1}< \delta/\alpha \wedge p_{\omega_2} = 0\right\}.
	\]
	It is easy to see that, under every contract $\vp\in\cP_1$, action $a_1$ is a best response for the agent, no matter their cost vector. Thus, for all $\vc\in\cC$, it holds that:
	\[
	\Up(\vp,\vc)=1-p_{\omega_1}\le \Up(\vp^1,\vc),
	\]
	where we let $\vp^1 := (\delta/\alpha, 0)$.
	Thus, the principal is always better off using $\vp^1$ rather than any other contract in $\mcP_1$.
	Instead, under every contract $\vp\in\cP_2$, an agent with cost vector $(\delta,0)$ plays action $a_2$ as best response, resulting in principal's expected utility $\Up(\vp, (\delta,0)) = (1-\alpha)(1-p_{\omega_1})$.
	Thus, against an agent with cost vector $(\delta,0)$, the principal is always better off using $\vp^2 :=(0,0)$ rather than any other contract in $\cP_2$. 
	Moreover, an agent with cost vector $(0,0)$ plays action $a_1$ as best response under every $\vp\in\cP_2$, so that $\vp_2$ is the best possible choice among contracts in $\mcP_2$ also in this case.
	As a result, we can focus w.l.o.g.~on randomized contracts $\gvec \in \mcPr$ that only have the contracts $\vp^1$ and $\vp^2$ defined above in their supports $\supp(\gvec)$.
	%

	Given the analysis above, it clearly holds that $\OPT(\cvec) = 1$ for $\cvec = (0,0)$, since the principal can incentivize the agent to play action $a_1$ while paying zero to them, by committing to contract $\vp^2$.
	Moreover, $\OPT(\cvec) = \max \left\{ 1-\alpha,1-\delta/\alpha \right\}$ for $\cvec = (\delta,0)$, since the principal can either commit to $\vp^1$, with the agent playing $a_1$ resulting in principal's expected utility $1- \delta/\alpha$, or commit to $\vp^2$, with the agent playing $a_2$ and expected utility $1-\alpha$.
	Thus, the value of $ \OPT(\vc)-\Up(\vp,\vc) $ for $\vp \in \left\{ \vp^1, \vp^2 \right\}$ and $\cvec \in \mcC$ is defined as in the following table:
	\begin{center}
		\begin{tabular}{c|c|c|}
			& $\cvec = (0,0)$         & $\cvec=(\delta,0)$                                          \\ \hline
			$\vp^1=(\delta/\alpha,0)$                 & $0$& $\alpha-\delta/\alpha$     \\ \hline
			$\vp^2=(0,0)$ & $\delta / \alpha$                     & $0$            \\ \hline
		\end{tabular}
	\end{center}

	Let us fix $\gvec \in \mcPr$ with $\supp(\gvec) = \left\{ \vp^1, \vp^2 \right\}$.
	Then, we can write the following:
	\[
		\mcR(\gvec) = \sup_{\cvec \in \mcC} \left\{ \OPT(\vc)-\Up(\gvec,\vc)  \right\} = \sup_{\cvec \in \mcC} \left\{ \gamma_{\pvec^1}  \Big( \OPT(\vc)- \Up(\vp^1,\cvec) \Big) + \gamma_{\pvec^2} \Big( \OPT(\vc)- \Up(\vp^2,\cvec) \Big) \right\} .
	\]
	Since $\mcC = \{ (0,0), (\delta,0) \}$, by plugging in the values defined in the table above we get:
	\[
		\mcR(\gvec) = \max \left\{\gamma_{\pvec^2} \frac{\delta}{\alpha} , \gamma_{\pvec^1} \left( \alpha- \frac\delta\alpha \right)  \right\} = \max \left\{ \left( 1-\gamma_{\pvec^1} \right) \sqrt{\frac{\delta}{2}} , \gamma_{\pvec^1} \sqrt{\frac{\delta}{2}}  \right\},
	\]
	where the last equality follows from $\gamma_{\pvec^2} = 1-\gamma_{\pvec^1}$ and $\alpha = \sqrt{2\delta}$.
	
	In conclusion, since the minimum possible value of $\mcR(\gvec)$ is $ \sqrt{2\delta}/4$, we get that the regret attained by randomized contracts in the PAPU instances in Definition~\ref{def:hard} is at least $\sqrt{2\delta}/4$.
\end{proof}

\lemmaWorstTre*

\begin{proof}
	Point~(i) follows from the convexity of $\cC$. Indeed, given a cost vector $\vc \in \cC_\delta(\vc_0)$, it holds that $\vc= \vc_0 (1-\delta) + \vc' \delta$ for some $\vc'\in \cC$. 
	Consider the cost vector $\hat \vc := \frac{\delta'-\delta}{\delta'}\cvec_0 +\frac{\delta}{\delta'}\vc'$. By convexity of $\cC$, it must be the case that $\hat \vc \in \cC$.
	Then, the following holds: \[(1-\delta') \vc_0 + \delta' \hat \vc=(1-\delta') \vc_0 + \delta' \left( \frac{\delta'-\delta}{\delta'}\cvec_0 +\frac{\delta}{\delta'}\vc' \right)=(1-\delta) \vc_0+ \delta \vc'=\vc, \]
	which shows that $\vc$ belongs to $\cC_{\delta'}(\vc_0)$. This concludes the first part of the proof.
	
	To prove point~(ii), we notice that the $L_{\infty}$-diameter of the set $\mcC_{\delta'}(\vc_0)$ is
	\begin{align*}
		\sup_{c,c'\in \mcC_\delta(\vc_0)} \lVert \vc-\vc'\rVert_{\infty} & = \sup_{c,c'\in \mcC} \lVert (1-\delta)\vc_0 +\delta \vc- (1-\delta)\vc_0 - \delta \vc' \rVert_{\infty}\\
		& \le \delta \sup_{c,c'\in \mcC} \lVert \vc-\vc' \rVert_{\infty}\\
		&=\delta d(\mcC) .
	\end{align*}
	This concludes the proof.
\end{proof}

\theoremWorstDue*

\begin{proof}
	%
	Let $\pi:\cC_{\delta+h}(\vc_0) \to \cC_\delta(\vc_0)$ be a function such that $\|\vc-\pi(\vc)\|_\infty\le h$ for all $\vc\in\cC_{\delta+h}(\vc_0)$.
	Such a function is guaranteed to exist.
	Indeed, for every cost vector $\vc \in \mcC_{\delta+h}(\vc_0)$, by definition of the set $\cC_{\delta+h}(\vc_0)$ it holds that $\vc= (\delta+h) \vc' +\vc_0(1-\delta-h) $ for some $\vc' \in \mcC$.
	Thus, one possible option is to take $\pi(\vc) = \delta \vc' + (1-\delta) \vc_0$, so that $\pi(\vc) \in \cC_{\delta}(\vc_0)$ and $\|\vc-\pi(\vc)\|_\infty\le h$.

	%
	In order to prove the first `$\leq$' inequality in the statement of the theorem, we notice that:
	%
	\begin{align*}
		\inf_{\vp \in \mcP} \mcR_{\delta, \vc_0}(\pvec)&=\inf_{\vp \in \mcP} \sup\limits_{\vc \in\cC_\delta(\vc_0)} \left\{ \OPT(\vc)-\Up(\pvec,\vc) \right\}\\
		&\le \inf_{\vp \in \mcP} \sup\limits_{\vc \in\cC_{\delta+h}(\vc_0)} \left\{ \OPT(\vc)-\Up(\pvec,\vc) \right\} \\
		& = \inf_{\vp \in \mcP} \mcR_{\delta+h, \vc_0}(\pvec),
	\end{align*}
	where the inequality follows from the fact that $\cC_\delta(\vc_0)\subseteq \cC_{\delta+h}(\vc_0)$, by point~(i) of Lemma~\ref{lem:scaling}.

	In order to prove the second `$\leq$' inequality in the statement of the theorem, we let $\vp^\star \in \mcP$ be a regret-minimizing contract for the PAPU instance in which the uncertainty set is $\mcC_\delta(\vc_0)$, namely, one that satisfies $\mcR_{\delta,\vc_0}(\pvec^\star) = \inf_{\vp \in \mcP} \mcR_{\delta,\vc_0}(\vp)$.
	Notice that such a $\vp^\star$ is guaranteed to exist by Corollary~\ref{cor:det_exist}.
	Then, by letting $\tilde \vp :=\vp^\star+\sqrt{2h}(\rvec-\vp^\star)$, we can prove the following:
	\begin{align*}
	\inf_{\vp \in \mcP} \mcR_{\delta+h, \vc_0}(\pvec) &= \inf_{\vp \in \mcP} \sup\limits_{\vc\in\cC_{\delta+h}(\vc_0)} \left\{ \OPT(\vc)-\Up(\vp,\vc) \right\}\\
	&\le \sup\limits_{\vc\in\cC_{\delta+h}(\vc_0)} \left\{  \OPT(\vc)-\Up(\tilde \vp,\vc) \right\}\\
	&\le \sup\limits_{\vc\in\cC_{\delta+h}(\vc_0)} \left\{ \OPT(\vc)-\Up(\vp^{\star},\pi(\vc))+2\sqrt{2h} \right\}\\
	& = \sup\limits_{\vc\in\cC_{\delta+h}(\vc_0)} \left\{ \OPT(\pi(\vc))-\Up(\vp^{\star},\pi(\vc))+ \OPT(\vc)- \OPT(\pi(\vc))+2\sqrt{2h} \right\}\\
	& \le \sup\limits_{\vc\in\cC_{\delta+h}(\vc_0)} \left\{ \OPT(\pi(\vc))-\Up(\vp^{\star},\pi(\vc))+4\sqrt{2h} \right\}\\
	& \le \sup\limits_{\vc\in\cC_{\delta}(\vc_0)} \left\{ \OPT(\vc)-\Up(\vp^{\star},\vc)+4\sqrt{2h} \right\} \\
	&= 	\inf_{\vp \in \mcP} \mcR_{\delta, \vc_0}(\pvec)+4\sqrt{2h},
	\end{align*}
	where the second inequality follows from \cref{lem:delta_opt_raw} with $\alpha = \sqrt{2h}$ and $\vp = \vp^\star$, the third one from \cref{lem:holder_opt}, the last inequality holds since $\mcC_\delta(\vc_0) \subseteq \mcC_{\delta+h}(\vc_0)$ and $\pi(\vc) \in \mcC_\delta(\vc_0)$, while the last equality follows from the definition of $\vp^\star$.
	%
	%
\end{proof}
\section{Proofs Omitted from Section~\ref{sec:robustness}}\label{sec:app_robust}

\propositionCompareOne*

\begin{proof}
	Consider the family of PAPU instances $\mathcal{I}_\delta := (A,\Omega,\rvec,\mcC,\Fvec)$ parametrized by $\delta \in [0,1]$ described in the following. The agent has five actions, namely $A:= \left\{ a_1, a_2, a_3, a_4,  a_5 \right\}$. 
	There are three outcomes, namely $\Omega := \left\{ \omega_1, \omega_2, \omega_3\right\} $.
	The reward of the principal for outcomes $\omega_1$ and $\omega_2$ is $1$, namely $r_{\omega_1} = r_{\omega_2} = 1$, while the reward of outcome $\omega_3$ is $r_{\omega_3} = 0$.
	The probability matrix $\Fvec$ is such that agent's actions induce outcomes as follows:
	\begin{itemize}
		\item Action $a_1$ induces outcome $\omega_1$ with probability $1$.
		\item Action $a_2$ induces outcome $\omega_1$ with probability $1-\sqrt{\delta}$ and outcome $\omega_3$ with probability $\sqrt{\delta}$.
		\item Action  $a_3$ induces outcome $\omega_2$ with probability $1$.
		\item Action $a_4$ induces outcome $\omega_2$ with probability $1-\sqrt{\delta}$ and outcome $\omega_3$ with probability $\sqrt{\delta}$.
		\item Action $a_5$ induces outcomes $\omega_1,\omega_2$ with probability $\frac{1}{2}-\frac{\sqrt{\delta}}{4}$, and $\omega_3$ with probability $\frac{\sqrt{\delta}}{2}$. 
	\end{itemize}
	Moreover, by letting $c_2 := c_4 :=\frac{\sqrt{\delta}}{10}$ and $c_5 :=0$, the uncertainty set $\cC := \left\{ \vc^1,\vc^2,\vc^3 \right\}$ only contains three cost vectors, which are defined as follows:
	\begin{itemize}
		\item$\vc^1:=\left(\frac{\sqrt{\delta}}{10},c_2,\frac{\sqrt{\delta}}{10}+\delta,c_4,c_5 \right)$.
		\item$ \vc^2:=\left(\frac{\sqrt{\delta}}{10}+\delta,c_2,\frac{\sqrt{\delta}}{10},c_4,c_5 \right)$.
		\item $ \vc^3:= \left(\frac{\sqrt{\delta}}{10}+\delta,c_2,\frac{\sqrt{\delta}}{10}+\delta,c_4,c_5 \right)$.
	\end{itemize}
	Notice that the uncertainty level of the instances is $d(\mcC) = \delta$.
	
	%
	Before analyzing the regret of the two classes of contracts, we compute $\OPT(\vc)$ for every $\vc \in \mcC$.
	
	\paragraph{Value of $\OPT( \vc^1 )$}
	We show that, in an optimal deterministic contract, an agent with cost vector $\vc^1$ plays action $a_1$ as best response.
	First, we observe that, in any contract $\vp \in \mcP$ in which $a_1$ is IC, action $a_1$ must provide the agent with an expected utility greater than or equal to the one achieved by action $a_5$.
	Thus, from $\Ua(\vp,\vc^1,a_1) \geq \Ua(\vp,\vc^1,a_5)$, we can conclude that $$p_{\omega_1} - \frac{\sqrt{\delta}}{10}\ge p_{\omega_1} \left( \frac{1}{2}-\frac{\sqrt{\delta}}{4} \right) + p_{\omega_2} \left( \frac{1}{2}-\frac{\sqrt{\delta}}{4} \right) + p_{\omega_3} \frac{\sqrt{\delta}}{2} \geq p_{\omega_1} \left( \frac{1}{2}-\frac{\sqrt{\delta}}{4} \right) ,$$ which implies that $p_{\omega_1}\ge \frac{\frac{\sqrt{\delta}}{10}}{\frac{1}{2}+\frac{\sqrt{\delta}}{4}}.$
	%
	Now, let us consider a contract $\vp \in \mcP$ with payments defined as:
	$$p_{\omega_1} = \frac{\frac{\sqrt{\delta}}{10}}{\frac{1}{2}+\frac{\sqrt{\delta}}{4}} \quad \text{and} \quad p_{\omega_2} = p_{\omega_3} = 0.$$
	%
	%
	Clearly, under $\vp$, action $a_1$ is better than action $a_5$ for the agent.
	Moreover, it is easy to check that action $a_1$ achieves agent's expected utility greater than the one achieved by actions $a_2$, $a_3$, and $a_4$.
	Thus, the agent plays action $a_1$ as best response, providing the principal with expected utility: $$\Up(\vp,\vc^1)=  1-\frac{\frac{\sqrt{\delta}}{10}}{\frac{1}{2}+\frac{\sqrt{\delta}}{4}}.$$
	%
	%
	Finally, it is easy to prove that all the other agent's actions provide a smaller expected utility to the principal than action $a_1$, independently of the contract.
	This implies that $\vp$ is an optimal contract against an agent with cost vector $\vc^3$, and, thus, the following holds: $$\OPT ( \vc^1 )=1-\frac{\frac{\sqrt{\delta}}{10}}{\frac{1}{2}+\frac{\sqrt{\delta}}{4}}.$$
	
	\paragraph{Value of $\OPT( \vc^2 )$}
	By symmetry of the instance, we have that $\OPT(\vc^2)=\OPT ( \vc^1)=1-\frac{\frac{\sqrt{\delta}}{10}}{\frac{1}{2}+\frac{\sqrt{\delta}}{4}}$.

	\paragraph{Value of $\OPT( \vc^3 )$}
	%
	%
	We show that, under an optimal deterministic contract against an agent with cost vector $\vc^3$, the agent plays action $a_5$ as best response.
	Indeed, since action $a_5$ has cost $c_5 = 0$, the contract $\vp \in \mcP$ assigning zero payment to every outcome induces the agent to play $a_5$ as best response (as it is the action with the lowest cost) and results in $ \Up(\vp,\vc^3)  = 1-\frac{\sqrt{\delta}}{2}$.
	If action $a_1$ is the best response played by the agent under a contract $\vp \in \mcP$, then, since $a_1$ must be IC, it must provide agent's expected utility greater than or equal to the one achieved by $a_2$.
	Hence, $$p_{\omega_1} - \frac{\sqrt{\delta}}{10}-\delta \ge p_{\omega_1} \left( 1-\sqrt{\delta} \right) + p_{\omega_3} \sqrt{\delta}-\frac{\sqrt{\delta}}{10} \geq p_{\omega_1} \left( 1-\sqrt{\delta} \right) -\frac{\sqrt{\delta}}{10},$$
	which implies that $p_{\omega_1}\ge \sqrt{\delta}$, and, thus, $\Up(\vp,\vc^3) \leq 1 - \sqrt{\delta} \leq 1-\frac{\sqrt{\delta}}{2}$. 
	Similar arguments also hold for contracts in which the best response played by the agent is $a_2$, $a_3$, or $a_4$.
	Thus, in an optimal contract against an agent with cost vector $\vc^3$, the agent must plays action $a_5$.
	As a result, the optimal contract is the one assigning zero payment to every outcome, which implies that $$\OPT(\vc^3) = 1-\frac{\sqrt{\delta}}{2}.$$
	%
	
	Now, we are ready to prove the result.
	In particular, we show that, in the family of PAPU instances defined above, there always exists a menu of deterministic contract achieving regret at most $R_\delta := \frac{\sqrt{\delta}}{10}$, while all the randomized contracts attain regret at least $R_\delta + \Omega(\sqrt{\delta})$ as $\delta \to 0$.
	%
	
	\paragraph{Menu of Deterministic Contracts}
	Let us consider a menu $\Pi = \left\{ \vp^1, \vp^2 \right\}$, where the deterministic contracts $\vp^1, \vp^2 \in \mcP$ are defined as follows:
	\[
		p^1_{\omega_1}=\frac{\frac{\sqrt{\delta}}{10}}{\frac{1}{2}-\frac{\sqrt{\delta}}{4}}, p^1_{\omega_2} = p^1_{\omega_3} =0 \quad \text{and} \quad p^2_{\omega_2}=\frac{\frac{\sqrt{\delta}}{10}}{\frac{1}{2}-\frac{\sqrt{\delta}}{4}}, p^2_{\omega_1} = p^2_{\omega_3} = 0.
	\]
	%
	%
	It is easy to check that an agent with cost vector $\vc^1$ is better off selecting $\vp^1$ rather than $\vp^2$, playing action $a_1$ as best response after that.
	Thus, by employing the value of $\OPT(\vc^1)$ computed above and the definition of $p_{\omega_1}$, we have that $\OPT(\vc^1) - \Up(\Pi, \vc^1) = 0$.
	%
	%
	By a symmetric argument, an agent with cost vector $\vc^2$ selects $\vp^2$ and best responds by playing $a_3$, resulting in $\OPT(\vc^2) - \Up(\Pi, \vc^2) = 0$.
	%
	%
	Finally, it is easy to see that an agent with cost vector $\vc^3$ is indifferent between the contracts $\vp^1$ and $\vp^2$, as the agent is always better off playing $a_5$ rather than any other action.
	Thus, principal's expected utility against an agent with cost vector $\vc^3$ satisfies the following relation:
	\[
		\Up(\Pi, \vc^3) = 1- \frac{\sqrt{\delta}}{2}-\frac{\frac{\sqrt{\delta}}{10}}{\frac{1}{2}+\frac{\sqrt{\delta}}{4}} \left( \frac{1}{2}-\frac{\sqrt{\delta}}{4} \right) \ge 1- \frac{\sqrt{\delta}}{2}-\frac{\sqrt{\delta}}{10}.
	\]
	This shows that $\OPT(\vc^3) - \Up(\Pi, \vc^3) \leq \frac{\sqrt{\delta}}{10}$.
	%
	%
	As a result, we have that $\inf_{\Pi \in \mcM} \mcR(\Pi) \leq \frac{\sqrt{\delta}}{10} = R_\delta$.

	\paragraph{Randomized Contracts}
	We show that any randomized contract $\gvec \in \Delta_{\supp(\gvec)}$ attains regret $\mcR(\gvec)$ at least a suitably-defined $\eta$, for sufficiently small values of $\delta$. In particular, $\eta:= \nicefrac{\sqrt{\delta}}{9}$, which is clearly equal to $R_\delta + \Omega(\sqrt{\delta})$ as $\delta \to 0$.
	%
	%
	By contradiction, suppose that $\mcR(\gvec) < \eta$.
	First, we consider the case of an agent with cost vector $\vc^1$, where we let $\mcP_1 := \left\{ \vp \in \mcP  \mid a^\star(\vp, \vc^1) = a_1 \right\}$ be the set of deterministic contracts under which the agent plays action $a_1$ as best response. 
	By the arguments used to compute the value of $\OPT(\vc^1)$, in any contract $\vp \in \mcP_1$, the payment $p_{\omega_{1}}$ must be such that: $$ p_{\omega_{1}} \geq \frac{\frac{\sqrt{\delta}}{10}}{\frac{1}{2}+\frac{\sqrt{\delta}}{4}}.$$
	Then, principal's expected utility $\Up(\gvec,\vc^1) $ is equal to:
	\begin{align*}
		\sum_{\substack{\vp \in \supp(\gvec) :\\ \vp \in \mcP_1}} \gamma_{\vp} \Up(\vp, \vc^1) + \sum_{\substack{\vp \in \supp(\gvec) :\\ \vp \notin \mcP_1}} \gamma_{\vp} \Up(\vp, \vc^1) \leq 
		\sum_{\substack{\vp \in \supp(\gvec) :\\ \vp \in \mcP_1}} \gamma_{\vp} \OPT(\vc^1) + 
		\sum_{\substack{\vp \in \supp(\gvec) :\\ \vp \notin \mcP_1}} \gamma_{\vp} \left( 1- \frac{\sqrt{\delta}}{2} \right),
	\end{align*}
	where the inequality follows from the fact that $\Up(\vp,\vc^1) \leq \OPT(\vc^1)$ for every $\vp \in \mcP_1$ (since the maximum possible principal's expected utility when the agent best responds with $a_1$ is $\OPT(\vc^1)$) and $\Up(\vp,\vc^1) \leq  1- \nicefrac{\sqrt{\delta}}{2} $ for every $\vp \notin \mcP_1$ (since the maximum possible principal's expected utility when the agent plays a best response different from $a_1$ is $1- \nicefrac{\sqrt{\delta}}{2}$).
	Since $\mcR(\gvec) < \eta$ by assumption and $\OPT(\cvec^1) - \Up(\vp,\vc^1) \leq \mcR(\gvec)$ by definition, we have that the following holds:
	\[
		\OPT(\cvec^1)  - \sum_{\substack{\vp \in \supp(\gvec) :\\ \vp \in \mcP_1}} \gamma_{\vp} \OPT(\vc^1)  - \left( 1 - \sum_{\substack{\vp \in \supp(\gvec) :\\ \vp \in \mcP_1}} \gamma_{\vp} \right) \left( 1- \frac{\sqrt{\delta}}{2} \right) < \eta = \frac{\sqrt{\delta}}{9}.
	\]
	%
	By manipulating the inequality above and using the definition of $\OPT(\vc^1)$, we obtain that the probability $\sum_{\vp \in \supp(\gvec) : \vp \in \mcP_1} \gamma_{\vp_1}$ with which an agent with cost vector $\vc^1$ plays action $a_1$ under $\gvec$ is greater than $\frac{17}{27}$.
	This implies that, under $\gvec$: $$p_{\omega_1}\ge \frac{\frac{\sqrt{\delta}}{10}}{\frac{1}{2}+\frac{\sqrt{\delta}}{4}} \quad \text{holds with probability greater than } \frac{17}{27}.$$
	A symmetric argument shows that, in the case of an agent with cost vector $\vc^2$, under $\gvec$:
	$$p_{\omega_2}\ge \frac{\frac{\sqrt{\delta}}{10}}{\frac{1}{2}+\frac{\sqrt{\delta}}{4}} \quad \text{holds with probability greater than }\frac{17}{27}.$$
	Finally, let us consider the case of an agent with cost vector $\vc^3$.
	%
	For ease of presentation, let us define the following three sets of contracts:
	\begin{align*}
		\widetilde{\mcP} := \left\{ \vp \in \supp(\gvec) \mid p_{\omega_1}, p_{\omega_2} \geq \frac{\frac{\sqrt{\delta}}{10}}{\frac{1}{2}+\frac{\sqrt{\delta}}{4}} \right\}, \\
		\widetilde{\mcP}_1 := \left\{ \vp \in \supp(\gvec) \mid p_{\omega_1}\ge \frac{\frac{\sqrt{\delta}}{10}}{\frac{1}{2}+\frac{\sqrt{\delta}}{4}}\right\} \setminus \widetilde{\mcP} ,\\
		\widetilde{\mcP}_2 := \left\{ \vp \in \supp(\gvec) \mid p_{\omega_2} \ge \frac{\frac{\sqrt{\delta}}{10}}{\frac{1}{2}+\frac{\sqrt{\delta}}{4}} \right\}  \setminus \widetilde{\mcP}.
	\end{align*}
	It is easy to see that, for any contract $\vp \in \widetilde{\mcP}$, principals' expected utility is the maximum possible when the agent best responds with action $a_5$ and the payment is the minimum possible on every outcome.
	%
	%
	Then, for every $\vp  \in \widetilde{\mcP}$, it holds that:
	\[  \Up(\vp,\vc^3)\le  1-\frac{\sqrt{\delta}}{2}-2 \left( \frac{1}{2}-\frac{\sqrt{\delta}}{4} \right) \frac{\frac{\sqrt{\delta}}{10}}{\frac{1}{2}+\frac{\sqrt{\delta}}{4}} .\]
	Similarly, for any contract $p \in \widetilde{\mcP}_1$, principals' expected utility is the maximum possible when the agent best responds with $a_5$ and the payment is the minimum possible on every outcome.
	%
	%
	Then, for every $\vp  \in \widetilde{\mcP}_1$, it holds that:
	\[  \Up(\vp,\vc^3)\le  1-\frac{\sqrt{\delta}}{2}-\left( \frac{1}{2}-\frac{\sqrt{\delta}}{4} \right) \frac{\frac{\sqrt{\delta}}{10}}{\frac{1}{2}+\frac{\sqrt{\delta}}{4}} .\]
	By a symmetric argument, for every contract $\vp \in \widetilde{\mcP}_2$, it holds that:
	\[  \Up(\vp,\vc^3)\le 1-\frac{\sqrt{\delta}}{2}-\left( \frac{1}{2}-\frac{\sqrt{\delta}}{4} \right) \frac{\frac{\sqrt{\delta}}{10}}{\frac{1}{2}+\frac{\sqrt{\delta}}{4}} .\]
	Moreover, notice that 
	\[ 2\sum_{\vp \in \widetilde{\mcP}} \gamma_{\vp} + \sum_{\vp \in \widetilde{\mcP}_1\cup \widetilde{\mcP}_2} \gamma_{\vp} \ge \frac{34}{27} . \]
	Thus, the regret $\mcR(\gvec) $, which is at least $ \OPT(\vc^3) - \Up(\gvec,\vc^3)$, satisfies the following relation:
	\begin{align*}
		\mcR(\gvec) & \ge\OPT(\vc^3) -  \sum_{\vp \in \supp(\gvec)\setminus (\widetilde{\mcP}\cup \widetilde{\mcP}_1\cup \widetilde{\mcP}_2)} \gamma_{\vp} \Up(\vp,\vc^3)- \sum_{\vp \in \widetilde{\mcP}} \gamma_{\vp} \Up(\vp,\vc^3) - \sum_{\vp \in \widetilde{\mcP}_1 \cup \widetilde{\mcP}_2}\Up(\vp,\vc^3)\\
		&\geq 1-\frac{\sqrt{\delta}}{2} -  \sum_{\vp \in \supp(\gvec)\setminus (\widetilde{\mcP}\cup\widetilde{ \mcP}_1\cup \widetilde{\mcP}_2)} \gamma_{\vp}\left(1-\frac{\sqrt{\delta}}{2}\right)- \sum_{\vp \in \widetilde{\mcP}} \gamma_{\vp} \left(1-\frac{\sqrt{\delta}}{2}-2 \left(\frac{1}{2}-\frac{\sqrt{\delta}}{4}\right) \frac{\frac{\sqrt{\delta}}{10}}{\frac{1}{2}+\frac{\sqrt{\delta}}{4}}\right) \\
		&\hspace{1.6cm}  -\sum_{\vp \in \widetilde{\mcP}_1 \cup \widetilde{\mcP}_2} \gamma_{\vp} \left(1-\frac{\sqrt{\delta}}{2}- \left(\frac{1}{2}-\frac{\sqrt{\delta}}{4}\right) \frac{\frac{\sqrt{\delta}}{10}}{\frac{1}{2}+\frac{\sqrt{\delta}}{4}}\right) \\
		&= \sum_{\vp \in \widetilde{\mcP}} \gamma_{\vp} 2 \left( \frac{1}{2}-\frac{\sqrt{\delta}}{4} \right) \frac{\frac{\sqrt{\delta}}{10}}{\frac{1}{2}+\frac{\sqrt{\delta}}{4}}   + \sum_{\vp \in \widetilde{\mcP}_1 \cup \widetilde{\mcP}_2} \gamma_{\vp}  \left( \frac{1}{2}-\frac{\sqrt{\delta}}{4} \right) \frac{\frac{\sqrt{\delta}}{10}}{\frac{1}{2}+\frac{\sqrt{\delta}}{4}}\\
		& =  \frac{34}{27} \left(\frac 1 2- \frac {\sqrt{\delta}} 4 \right)  \frac{\frac{\sqrt{\delta}}{10}}{\frac{1}{2}+\frac{\sqrt{\delta}}{4}} \\
		& >  \frac{34}{27} \left( \frac 1 2- \frac 1 {32} \right)  \frac{\frac{\sqrt{\delta}}{10}}{\frac{1}{2}+ \frac 1 {32}} \\
		& \geq  \frac{\sqrt{\delta}}{9} =  \eta
	\end{align*}
	%
	%
	%
	which leads to the desired contradiction for $\delta< \nicefrac 1 {64}$.
\end{proof}

\propositionCompareTwo*

\begin{proof}
	Let us consider the family of PAPU instances $\mathcal{I}_\delta := (A,\Omega,\rvec,\mcC,\Fvec)$ parametrized by $\delta \in [0,\nicefrac{1}{2}]$ introduced in Definition~\ref{def:hard}, where $\alpha=\sqrt{2\delta}$.
	Notice that such instances satisfy $d(\mcC) = \delta$, and, for ease of presentation, let $\vc^1\coloneqq (\delta,0)$ and $\vc^2\coloneqq (0,0)$.
	
	As a first step, let us computes the values of $\OPT(\vc^1)$ and $\OPT(\vc^2)$.
	It is easy to see that, against an agent with cost vector $\vc^1$, an optimal deterministic contract incentivizes the agent to play action $a_1$ as best response, by using payments $p_{\omega_1}=\sqrt{\nicefrac{\delta}{2}}$ and $p_{\omega_{2}} = 0$.
	Thus, it holds $\OPT(\vc^1)=1-\sqrt{\nicefrac{\delta}{2}}$. Moreover, it is immediate to check that $\OPT(\vc_2)=1$.
	
	Next, we show that, for the family of PAPU instances in Definition~\ref{def:hard}, there always exists a randomized contract which attains regret $R_\delta := \sqrt{\nicefrac{\delta}{8}}$, while every menu of deterministic contracts attains regret at least $R_\delta + \Omega(\sqrt{\delta})$ as $\delta \to 0$.
	
	\paragraph{Randomized Contracts}
	Consider a randomized contract $\vgamma \in \Delta_{\supp(\gvec)}$ that uniformly randomizes between two deterministic contracts $\vp^1$ and $\vp^2$ defined as follows.
	The contract $\vp^1$ sets zero payment on every outcome, while the contract $\vp^2$ sets payment $\sqrt{\nicefrac{\delta}{2}}$ on outcome $\omega_1$ and zero payment on outcome $\omega_2$.
	Then, it is easy to see that an agent with cost vector $\vc^1$ plays action $a_1$ as best response if the realized contract is $\vp^2$, while they best respond with action $a_2$ otherwise. Hence, for an agent with cost vector $\vc^1$, the following holds: $$\OPT(\vc^1) - \Up(\gvec,\vc^1) = 1-\sqrt{\frac{\delta}{2}}- \frac{1}{2}\left(1-\sqrt{2\delta} \right)-\frac{1}{2}\left(1-\sqrt{\frac{\delta}{2}} \right)=\sqrt{\frac{\delta}{8}}.$$
	Moreover, it is easy to see that an agent with cost vector $\vc_2$ always plays action $a_1$, resulting in an expected payment of $\sqrt{\nicefrac{\delta}{8}}$. Hence, it holds $\OPT(\vc^2) - \Up(\gvec,\vc^2) =1- (1-\sqrt{\nicefrac{\delta}{8}})=\sqrt{\nicefrac{\delta}{8}}$. As a result, the regret $\mcR(\gvec)$ attained by $\gvec$ is equal to $ \sqrt{\nicefrac{\delta}{8}} = R_\delta$.
	 
	\paragraph{Menus of Deterministic Contracts}
	We show that any menu of deterministic contracts attains regret at least $\sqrt{\nicefrac \delta 4} = R_\delta + \Omega(\sqrt{\delta})$ as $\delta \to 0$.
	By a revelation-principle-style argument, we can focus w.l.o.g. on menus of deterministic contracts $\Pi = \left\{ \vp^1, \vp^2 \right\}$ with only two elements, where $\vp^1$ is the contract selected by an agent with cost vector $\vc^1$ and $\vp^2$ is the one chosen by an agent with cost vector $\vc^2$.
	%
	Suppose by contradiction that $\mcR(\Pi) < \sqrt{\nicefrac \delta 4}$.
	%
	%
	%
	Then, the contract $\vp^1$ must incentivize the agent to play action $a_1$ as best response, otherwise the regret incurred by $\Pi$ would be at least $\sqrt{\nicefrac \delta 4}$. Indeed, if the agent plays action $a_2$, then the resulting regret would be at least $$ 1-\sqrt{\frac \delta 2}-\left(1-\sqrt{2\delta} \right)\ge \sqrt{\frac \delta 4}.$$
	Hence, it must be $p^1_{\omega_1}\ge \sqrt{\nicefrac \delta 2}$.
	However, since an agent with cost vector $\vc^2$ chooses contract $\vp^2$, this implies that the expected payment collected by an agent with cost vector $\vc^2$ in the selected contract is at least $\sqrt{\nicefrac \delta 2}$. Indeed, it must be the case that an agent with cost vector $\vc^2$ collects a larger payment under $\vp^2$ than under $\vp^1$. This implies that the regret term for cost vector $\vc^2$ is at least $ 1-( 1-\sqrt{\nicefrac \delta 2})  = \sqrt{\nicefrac \delta 2}$. Hence, we reach a contradiction, showing that $\mcR(\Pi)\ge \sqrt{\nicefrac \delta 4}$.
	%
	%
\end{proof}

\section{Proofs Omitted from Section~\ref{sec:beyond}} \label{app:beyond}

\theoremBeyondOne*

\begin{proof}
	We provide a reduction from a restricted version of the \emph{exact cover by 3-sets} (\textsf{X3C}) problem. In particular, given a finite set $E\coloneqq\{e_1,\ldots, e_k\}$ of $k$ elements and a finite set $S\coloneqq \left\{s_1,\ldots,s_m\right\}$ of $m$ subsets of $E$ of three elements, namely $s \subseteq E$ and $|s|=3$ for all $s \in S$,  the \textsf{X3C} problem consists in determining if there exists an \emph{exact cover} of $E$, \emph{i.e.}, a collection of subsets from $S$ such that each element of $E$ appears in exactly one of the selected subsets. In the following, we focus on a restricted version of the \textsf{X3C} problem in which each element of $E$ appears in exactly \emph{three} subsets in $S$. Notice that this restricted version of the \textsf{X3C} problem is already known to be \textsf{NP}-hard~\citep{GONZALEZ1985293}.

	\paragraph{Construction}
	Given an instance $(E,S)$ of the \textsf{X3C} problem, we build an instance $(A,\Omega,\rvec,\mcC,\Fvec)$ of the PAPU as described in the following.
	The instance has an outcome $\omega_s$ for every $s \in S$, and two additional outcomes $\omega_0$ and $\omega_1$.
	Formally, $\Omega \coloneqq \left\{ \omega_s \mid s \in S \right\} \cup \left\{  \omega_0,\omega_1\right\}$.
	Principal's reward for outcome $\omega_1$ is $r_{\omega_1} = 1$, while all the other rewards are zero, namely $r_\omega = 0$ for every $\omega \in \Omega \setminus \left\{ \omega_1 \right\}$.
	The set of agent's actions is $A \coloneqq \left\{ a_{s,1}, a_{s,2} \mid s \in S \right\} \cup \left\{ a_{e,1}, a_{e,2} \mid e \in E \right\} \cup \{ \bar a \}$, where:
	\begin{itemize}
		\item The agent has two actions $a_{s,1}$ and $a_{s,2}$ for every $s \in S$. Each action $a_{s,1}$ induces outcomes $\omega_s$ and $\omega_1$, each with probability $\epsilon := \frac{1}{10}$, while it results in outcome $\omega_0$ with the remaining probability of $\frac{8}{10}$. Similarly, each action $a_{s,2}$ induces outcomes $\omega_s$ and $\omega_1$ with probability $\frac \epsilon 2$, while it results in outcome $\omega_0$ with probability $\frac{9}{10}$.
		%
		\item The agent has two actions $a_{e,1}$ and $a_{e,2}$ for every $e \in E$. Each action $a_{e,1}$ induces each outcome $\omega_s$ such that $e \in s$ with probability $\frac 1 4$, while it results in outcome $\omega_1$ with probability~$\frac{1}{4}$. Similarly, each action $a_{e,2}$ induces each outcome $\omega_s$ such that $e \in s$ with probability $\frac 1 8$, while it results in $\omega_1$ with probability $\frac{1}{8}$ and in $\omega_0$ with probability $\frac{1}{2}$.
		%
		\item Action $\bar a$ induces outcome $\omega_1$ with probability $1$.
	\end{itemize}
	Finally, the uncertainty set $\cC \coloneqq \left\{ \vc^e \mid e \in E \right\} \cup \left\{ \vc^s \mid s \in S \right\} \cup \left\{ \bar \vc \right\}$ includes the following cost vectors:
	\begin{itemize}
		\item For every $e \in E$, there is a cost vector $\vc^e \in [0,1]^n$ such that $c^e_{a_{e,1}}=\frac{1}{8}-\frac \epsilon 2$ and $c^e_{e,2}=\frac{1}{16}-\frac \epsilon 2$, while, for all the other actions $a \in A \setminus \{a_{e,1},a_{e,2}\}$, it holds $c^e_a=1$.  
		\item For every $s \in S$, there is a cost vector $\vc^s \in [0,1]^n$ such that $c^s_{a_{s,1}}=\frac \epsilon 4$, and $c^s_{a_{s,2}}=0$, while, for all the other actions $a \in A \setminus \{a_{s,1},a_{s,2}\}$, it holds $c^s_a=1$.
		\item There is a cost vector $\bar \vc \in [0,1]^n$ such that $\bar c_{\bar a}=0$ and $\bar c_a=1$ for all $a \in A \setminus \{ \bar a \}$.  
	\end{itemize}
	
	\subsection*{Part I: \textnormal{\emph{If there exists an exact cover of $E$ in the \textsf{X3C} instance, then its corresponding PAPU instance admits a deterministic contract that achieves zero regret.}}} 
	
	%
	Suppose that there exists an exact cover $S^\star \subseteq S$. Let us consider a contract $\vp^\star \in \mcP$ whose payments are such that $p^\star_{\omega_{s}} =\frac{1}{2}\mathbbm{1}\left\{ s \in S^\star \right\}$ for every $s \in S$, while $p^\star_{\omega_{1}}=p^\star_{\omega_{0}}=0$.
	Next, we show that such a contract achieves the minimum possible regret, namely $\reg (\vp^\star) =0$.
	
	Consider a cost vector $\vc^s$ with $s \in S$. As a first step, we show that $\OPT(\vc^s)\le\frac{\epsilon}{2}$. It is easy to see that, against an agent with cost vector $\vc^s$, action $a_{s,1}$ is the only one that can possibly achieve principal's expected utility greater than $\frac{\epsilon}{2}$. If action $a_{s,1}$ is IC for an agent with cost vector $\vc^s$ under a contract $\vp \in \mcP$, then it must be the case that $\Ua (\vp, \vc^s, a_{s,1}) \geq \Ua (\vp, \vc^s, a_{s,2})$, namely:
		\[    \left( p_{\omega_1}+ p_{\omega_s} \right) \epsilon - \frac \epsilon 4 \ge \left( p_{\omega_1}+ p_{\omega_s} \right) \frac \epsilon 2,   \]
	which implies that 
	$p_{\omega_1}+  p_{\omega_s} \ge \frac{1}{2}$.
	Thus, it must be the case that $\OPT(\vc^s) \le  \epsilon-  \frac{\epsilon}{2} =  \frac \epsilon 2$.

	Next, we show that contract $\vp^\star$ achieves principal's expected utility $\Up(\vp^\star, \vc^s) = \frac \epsilon 2$ against an agent with cost vector $\vc^s$. We consider two cases. In the first one, $s \in S^\star$ and, thus, it holds $p^\star_{\omega_s}=\frac{1}{2}$. Then, it is easy to see that an action $a_{e,1}$ with $e \in s$ is IC, resulting in principal's expected utility $\frac{\epsilon}{2}$. In the second case, $s \in S^\star$, and, thus, it holds $p^\star_{\omega_s}=0$. Then, an action $a_{e,2}$ with $e \in s$ is IC, resulting in principal's expected utility $\frac{\epsilon}{2}$.
	This shows that $\Up(\vp^\star, \vc^s) = \frac \epsilon 2$ for every $s \in S$.
	
	As a result, we can conclude that 
	\[\sup_{\vc \in \{\vc^s \mid s \in S \}} \left\{ \OPT(\vc)- \Up(\vp^\star,\vc) \right\}=0.\]

	Now, let us consider a cost vector $\vc^e$ with $e \in E$. As a first step, we show that $\OPT(\vc^e)\le \frac{1}{8}$. 
	It is easy to see that, against an agent with cost vector $\vc^e$, action $a_{e,1}$ is the only one that can possibly achieve principal's expected utility greater than $\frac{1}{8}$. If action $a_{e,1}$ is IC for an agent with cost vector $\vc^e$ under a contract $\vp \in \cP$, then it holds that
	\begin{equation}\label{eq:inc}
		\left( p_{\omega_1}+ \sum_{s\in S:e \in s} p_{\omega_s} \right) \frac{1}{4} - \left( \frac{1}{8}-\frac \epsilon 2 \right) \ge \left( p_{\omega_1}+ \sum_{s \in S:e \in s} p_{\omega_s} \right) \frac{1}{8}- \left( \frac{1}{16}- \frac \epsilon 2 \right),
	\end{equation}
	which implies that
	\[   \left( p_{\omega_1}+ \sum_{s \in S:e \in s} p_{\omega_s}  \right)\ge \frac{1}{2}.  \]
	Thus, it must be the case that $\OPT(\vc^e)\le \frac{1}{4}-\frac{1}{2} \frac{1}{4}=\frac{1}{8}$.
	Moreover, contract $\vp^\star$ achieves principal's expected utility $\Up(\vp^\star, \vc^e) = \frac 1 8$. It is easy to see that \cref{eq:inc} holds for $\vp^\star$ (recall that $e$ is covered by \emph{exactly one} subset $s \in S^\star$), and, thus, an agent with cost vector $\vc^e$ plays action $a_{e,1}$, resulting in principal's expected utility $\frac 1 8$.
	As a result, we have that 
		\[\sup_{\vc \in \{\vc^e \mid e \in E\}} \left\{ \OPT(\vc)- \Up(\vp^\star,\vc) \right\}=0.\]
	
	Finally, against an agent with cost vector $\bar \vc$, it is easy to see that $\OPT(\bar \vc)=\Up(\vp^\star,\bar \vc))=1$.
	This allows us to conclude the first part of the proof, since 
		\[\reg(\vp^\star)=\sup_{\vc \in\cC} \left\{ \OPT(\vc)- \Up(\vp^\star,\vc) \right\}=0.\]
	
	\subsection*{Part II: \textnormal{\emph{If the PAPU instance admits a deterministic contract that attains regret at most some constant $\eta > 0$, then there exists an exact cover of $E$ in the \textsf{X3C} instance.}}} 
	
	%
	We show that, given any deterministic contract $\vp \in \mcP$ such that $\reg (\vp) \le \eta := \frac{1}{200}$, then it is possible to recover an exact cover $S^\star \subseteq S$ of $E$. 
	%
	
	First, we prove that $p_{\omega_1}\le \eta$. Indeed, since $\reg (\vp) \le \eta$, it must hold that $\OPT(\bar \vc)-\Up(\vp,\bar \vc) \le \eta$, which implies that $\Up(\vp,\bar \vc)\ge 1-\eta$, and, in turn, that $p_{\omega_1}\le \eta$.

	Next, we show that either $p_{\omega_s}\le \frac {2\eta} \epsilon$ or $p_{\omega_s}\ge \frac{1}{2}$ for every $s \in S$.
	Given any $s \in S$, by $\reg (\vp) \le \eta$, we have that $\OPT(\vc^s)-\Up(\vp,\vc^s)\le \eta$.
	Moreover, as shown in the previous part of the proof, we have that $\OPT(\vc^s)=\frac \epsilon 2$, and, thus, $\Up(\vp,\vc^s)\ge \frac \epsilon 2- \eta$.
	Suppose by contradiction that $p_{\omega_s} \in \big( \frac {2\eta} \epsilon, \frac{1}{2} -\eta  \big)$. Then, it must be the case that an agent with cost vector $\vc^s$ plays action $a_{s,2}$ as best response, since
			\[     \left( p_{\omega_1}+ p_{\omega_s} \right) \frac \epsilon 2 > \left( p_{\omega_1}+ p_{\omega_s} \right) \epsilon - \frac \epsilon 4,  \]
	which implies $\Ua(\vp,\vc^s, a_{s,2}) > \Ua(\vp,\vc^s, a_{s,1})$.
	As a result, we get
	\[ \Up(\vp,\vc^s)   \leq  \frac \epsilon 2- \frac \epsilon 2 p_{\omega_s} < \frac \epsilon 2- \eta, \] which leads to a contradiction. 
	This shows that $p_{\omega_s}\in \big[ 0, \frac {2\eta} \epsilon \big] \cup \big[ \frac{1}{2}-\eta, +\infty \big)$ for every $s \in S$.

	Now, we use the fact that $\OPT(\vc^e)-U(\vp, \vc^e)\le \eta$ for every $e \in E$. Given any $e \in E$, since $\OPT(\vc^e)=\frac{1}{8}$ as shown in the previous part of the proof, we have that $\Up(\vp, \vc^e) \ge \frac{1}{8} - \eta$.
	Then, it is easy to see that an agent with cost vector $\vc^e$ plays ${a_{e,1}}$ as best response. 
	Hence,
	\begin{equation*}
		\left( p_{\omega_1}+ \sum_{s \in S:e \in s} p_{\omega_s} \right) \frac{1}{4} - \left( \frac{1}{8}-\frac \epsilon 2 \right) \ge \left( p_{\omega_1}+ \sum_{s \in S:e \in s} p_{\omega_s} \right) \frac{1}{8}- \left(\frac{1}{16}- \frac \epsilon 2\right).   
	\end{equation*}
	This implies that 
	\begin{equation*}
		\sum_{s \in S:e \in s} p_{\omega_s} \ge 8 \left( \frac{1}{16} - \frac{1}{8}p_{\omega_1} \right)\ge \frac{1}{2} -\eta.
	\end{equation*} 
	Moreover, from the bound on principal's expected utility, we have that $$\Up(\vp, \vc^e)=\frac{1}{4}-\frac{1}{4} \left( \sum_{s \in S:e \in s} p_{\omega_s}+p_{\omega_1} \right) \ge \frac{1}{8}-\eta,$$ which implies that $$\sum_{s \in S:e \in s} p_{\omega_s}\le \frac{1}{2}+4\eta.$$
	The two bounds on the sum $\sum_{s \in S:e \in s} p_{\omega_s}$ allow to prove that exactly one subset $s \in S$ with $e \in s$ is such that its corresponding outcome $\omega_s$ has payment greater than $\frac{1}{2}-\eta$.
	Indeed, if two subsets $s, s' \in S$ are such that $p_{\omega_{s}} > \frac{1}{2}-\eta$ and  $p_{\omega_{s'}}> \frac{1}{2}-\eta$, then it must be the case that $$\sum_{s \in S:e \in s} p_{\omega_s}> 1- 2 \eta \geq  \frac{1}{2}+4\eta,$$ since~$\eta = \frac 1 {200}$.
	Moreover, if for all $s \in S $ with $e \in s$ it holds that $p_{\omega_{s}}<  \frac {2\eta} \epsilon$, then $$\sum_{s \in S:e \in s} p_{\omega_s} <  \frac {6\eta} \epsilon\le \frac{1}{2}-\eta .$$
	As a result, the set $S^\star \subseteq S$ of all the subsets $s \in S$ such that $p_{\omega_s}\ge \frac{1}{2}-\eta$ is an exact cover of $E$. 
\end{proof}

\lemmaTemplateOne*

\begin{proof}
	In the following, we assume that $\Gamma \in \cX$ is expressed as menu of randomized contracts, namely $\Gamma = \{ \gvec^1, \ldots, \gvec^K \} $.
	This is w.l.o.g. since $\mcMr$ is the most general class of contracts.

	First, let $\zeta : \mcC \to \mcC_\epsilon$ be a function that maps every cost vector $\vc \in \mcC$ to an element $\zeta(\vc)$ of the $\epsilon$-cover of $\mcC$ such that $\|\vc-\zeta (\vc)\|_\infty\le \epsilon$.
	Notice that such an element $\zeta (\vc)$ is guaranteed to exist by definition of $\epsilon$-cover.
	Then, let the function $t : \mcC \to [K]$ be such that, for every $\vc \in \mcC$, the value $t(\vc)$ is equal to the index of the randomized contract selected by an agent with cost vector $\zeta(\vc)$.
	%
	%
	Formally, for every $\vc \in \mcC$, $t(\vc) = i$ so that $\gvec^i $ coincides with $\gvec^\star(\Gamma, \zeta(\vc))$.
	%
	%
	Moreover, let the function $g: \mcC \times \mcP \to A$ be such that  $g(\vc,\vp)=a^\star(\vp, \zeta(\vc))$ for every cost vector $\vc \in \mcC$ and deterministic contract $\vp\in \supp(\gvec^{t(\vc)})$.
	Notice that, for every $\vc \in \mcC$, it is sufficient to define the function $g$ only for contracts $\vp \in \mcP$ that belong to the support of $\gvec^{t(\vc)}$.
	%
	
	First, it is easy to see that, by definition of the functions $t$ and $g$, the following hods:
	\[ \sup_{\vc \in \cC} \left\{  \OPT(\vc)- \sum_{\vp \in \supp(\gvec^{t(\vc)})} \gamma^{t(\vc)}_{\vp} \left[ \sum_{\omega \in \Omega} F_{g(\vc,\vp),\omega} \, r_\omega  - \sum_{\omega \in \Omega} F_{g(\vc,\vp),\omega} \, p_\omega \right] \right\} =  \sup_{\vc \in \cC_\epsilon} \Big\{ \OPT(\vc)-\Up(\Gamma,\vc)\Big\} .   \]

	Thus, we are left to prove that, given how the two function $t$ and $g$ are defined, the triplet $(\Gamma, t, g)$ is $(2\epsilon)$-incentive compatible.
	%
	%
	For every $\vc \in \cC$ and $i \in [K]$, by letting $\tilde \vc := \zeta(\vc)$, it holds
	\begin{align*}
		\sum_{\vp \in \supp(\gvec^{t(\vc)})} \gamma^{t(\vc)}_{\vp} \left[\sum_{\omega \in \Omega} F_{g(\vc,\vp),\omega} \, p_\omega  - c_{g(\vc,\vp)} \right] & \ge \sum_{\vp \in \supp(\gvec^{t (\vc)})} \gamma^{t (\vc)}_{\vp} \left[\sum_{\omega \in \Omega} F_{g(\vc,\vp),\omega} \, p_\omega - \tilde c_{g(\vc,\vp)} -\epsilon\right] \\
		& = \sum_{\vp \in \supp(\gvec^{t (\vc)})} \gamma^{t (\vc)}_{\vp} \left[\sum_{\omega \in \Omega} F_{a^\star(\vp, \tilde \vc),\omega} \, p_\omega - \tilde c_{a^\star(\vp, \tilde \vc)} \right]-\epsilon  \\
		& \ge\sum_{\vp \in \supp(\gvec^{i})} \gamma^{i}_{\vp}   \left[\sum_{\omega \in \Omega} F_{a^\star(\vp, \vc),\omega} \, p_\omega - \tilde c_{a^\star(\vp, \vc)} \right]-\epsilon\\
		& \ge\sum_{\vp \in \supp(\gvec^{i})} \gamma^{i}_{\vp}   \left[\sum_{\omega \in \Omega} F_{a^\star(\vp, \vc),\omega} \, p_\omega -  c_{a^\star(\vp, \vc)}-\epsilon \right]-\epsilon\\
		&= \sum_{\vp \in \supp(\gvec^{i})} \gamma^{i}_{\vp}   \left[\sum_{\omega \in \Omega} F_{a^\star(\vp, \vc),\omega} \, p_\omega -  c_{a^\star(\vp, \vc)} \right]-2\epsilon \\
		& = \Ua (\Gamma, \vc) - 2 \epsilon,
	\end{align*}
	where the first inequality holds by definition of the function $\zeta$, the subsequent (first) equality holds by definition of the function $g$, the second inequality holds by definition of the function $t$ and the fact that the actions $a^\star(\vp,\tilde \vc)$ are best responses for an agent with cost vector $\tilde \vc$, while the third inequality holds again by definition of the function $\zeta$.
	
	This clearly shows that $(\Gamma, t, g)$ is $(2\epsilon)$-incentive compatible, concluding the proof.
	%
	%
\end{proof}

\lemmaTemplateTwo*

\begin{proof}
	In the following, we assume that $\Gamma \in \cX$ is expressed as menu of randomized contracts, namely $\Gamma = \{ \gvec^1, \ldots, \gvec^K \} $.
	Moreover, we assume that the same holds for the contract $\widetilde{\Gamma} \in \cX$ produced by the procedure $\texttt{Linearize}(\Gamma,\beta)$, namely $\widetilde{\Gamma} = \{ \tilde \gvec^1, \ldots, \tilde \gvec^K \} $.
	These assumptions are w.l.o.g. since $\mcMr$ is the most general class of contracts.

	Let us fix a cost vector $\vc \in \cC$.
	%
	%
	Furthermore, for every $\vp \in \mcP$, let $\ell(\vp) := \vp + \beta (\rvec - \vp) = (1- \beta) \vp+\beta \rvec$ be the deterministic contract obtained by applying the ``linearization'' operation performed by the procedure $\texttt{Linearize}(\Gamma,\beta)$ to the deterministic contract $\vp$.
	We also introduce $\ell^{-1}(\vp)$ to denote the result of ``inverse'' operation applied to a deterministic contract $\vp \in \mcP$.
	Formally, we define $\ell^{-1}(\vp) := \frac{1}{1-\beta} \vp - \frac{\beta}{1-\beta} \rvec$.
	It is easy to see that, for every $\vp \in \mcP$, it holds $\ell^{-1}(\ell(\vp)) = \vp$.
	
	For ease of notation, let $i \in [K] $ be the index of the randomized contract selected by an agent with cost vector $\vc$ from $\Gamma$, namely $ \gvec^i$ coincides with $\gvec^\star( \Gamma, \cvec)$.
	Then, the following holds:
	%
	%
	\begin{align*}
		\sum_{\vp \in \supp(\gvec^i)}  \gamma^{i}_{\vp}  \left[ \sum_{\omega \in \Omega} F_{a^\star( \ell(\vp),\cvec),\omega} \right. & \left. \Big( (1-\beta) \, p_\omega+\beta \, r_\omega \Big)-c_{a^\star(\ell(\vp),\vc)} \right]\\
		&=\sum_{\vp \in \supp(\tilde \gvec^i)} \tilde \gamma^{i}_{\vp} \left[ \sum_{\omega \in \Omega} F_{a^\star(\vp,\cvec),\omega} \, p_\omega-c_{a^\star(\vp,\vc)} \right]\\
		&\ge \sum_{\vp \in \supp(\tilde \gvec^{t(\vc)})} \tilde \gamma^{t(\vc)}_{\vp} \left[ \sum_{\omega \in \Omega} F_{a^\star( \vp,\cvec),\omega} \, p_\omega-c_{a^\star(\vp,\vc)} \right] \\
		&\ge \sum_{\vp \in \supp(\tilde \gvec^{t(\vc)})} \tilde \gamma^{t(\vc)}_{\vp} \left[ \sum_{\omega \in \Omega} F_{g( \cvec,\ell^{-1}(\vp)),\omega}  \, p_\omega-c_{g(\vc,\ell^{-1}(\vp))} \right] \\
		&=\sum_{\vp \in \supp(\gvec^{t(\vc)})}  \gamma^{t(\vc)}_{\vp} \left[ \sum_{\omega \in \Omega} F_{g( \cvec,\vp),\omega}  \Big( (1- \beta) \, p_\omega+\beta \, r_\omega \Big)-c_{g(\vc,\vp)} \right],
	\end{align*}
	where the first equality holds by definition of $\widetilde{\Gamma}$ (as it is produced by the $\texttt{Linearize}(\Gamma,\beta)$ procedure), the subsequent (first) inequality holds given how the index $i$ is defined, the second inequality holds by definition of best response $a^\star(\vp,\vc)$, while the last equality holds by definition of $\widetilde{\Gamma}$ and $\ell^{-1}$, specifically, the fact that $\ell^{-1}(\ell(\vp)) = \vp$.
	%
	%

	Moreover, since $(\Gamma,t,g)$ is $\eta$-incentive compatible, the following holds:
	\begin{align*}
		\sum_{\vp \in \supp(\gvec^{t(\vc)})} \gamma^{t(\vc)}_{\vp} \left[ \sum_{\omega \in \Omega} F_{g(\vc,\vp),\omega} \, p_\omega - c_{g(\vc,\vp)} \right]& \geq \Ua(\gvec^i,\vc) - \eta \\
		&=\sum_{\vp \in \supp(\gvec^{i})} \gamma^{i}_{\vp} \left[ \sum_{\omega \in \Omega} F_{a^\star(\vp,\vc),\omega} \, p_\omega - c_{a^\star(\vp,\vc)}\right]-\eta\\
		&\ge \sum_{\vp \in \supp(\gvec^{i})} \gamma^{i}_{\vp} \left[ \sum_{\omega \in \Omega} F_{a^\star(\ell(\vp),\vc),\omega} \, p_\omega - c_{a^\star(\ell(\vp),\vc)}\right]-\eta,
	\end{align*}
	where the equality holds by definition of $\Ua(\gvec^i,\vc)$ and the last inequality holds since $a^\star(\vp,\vc)$ is a best response under $\vp$ for an agent with cost vector $\vc$.
	%

	By summing the two relations obtained above, we get
	\begin{align*}
		\sum_{\vp \in \supp(\gvec^i)} \gamma^{i}_{\vp} \left[ \sum_{\omega \in \Omega} F_{a^\star( \ell(\vp),\cvec),\omega} \left( -\beta p_\omega+\beta r_\omega \right) \right]\ge \sum_{\vp \in \supp( \gvec^{t(\vc)})}  \gamma^{t(\vc)}_{\vp} \left[\sum_{\omega \in \Omega} F_{g( \cvec, \vp),\omega}  \left(-\beta p_\omega+\beta r_\omega \right) \right]-\eta,
	\end{align*}
	%
	%
	which implies that 
	\[
	\sum_{\vp \in \supp(\gvec^i)} \gamma^{i}_{\vp} \left[ \sum_{\omega \in \Omega} F_{a^\star( \ell(\vp),\cvec),\omega}  \left( r_\omega - p_\omega  \right) \right] \ge \sum_{\vp \in \supp(\gvec^{t(\vc)})} \gamma^{t(\vc)}_{\vp} \left[ \sum_{\omega \in \Omega} F_{g( \cvec,\vp),\omega}  \left( r_\omega - p_\omega  \right) \right]-\frac \eta \beta.
	\]

	Then, the principal's expected utility when the agent is of type is $\vc$ is equal to
	%
	\begin{align*}
		\sum_{\vp \in \supp( \tilde \gvec^i)} \tilde \gamma^{i}_{\vp} \left[ \sum_{\omega \in \Omega} F_{a^\star( \vp,\cvec),\omega}  \left( r_\omega - p_\omega  \right) \right] &=	\sum_{\vp \in \supp(\gvec^i)} \gamma^{i}_{\vp} \left[ \sum_{\omega \in \Omega} F_{a^\star(\ell( \vp),\cvec),\omega}  \left( r_\omega - (1-\beta)p_\omega -\beta r_\omega  \right) \right] \\
		&=(1- \beta)\sum_{\vp \in \supp(\gvec^i)} \gamma^{i}_{\vp} \left[ \sum_{\omega \in \Omega} F_{a^\star(\ell(\vp),\cvec),\omega}  \left( r_\omega -p_\omega   \right) \right] \\
		&\ge (1-\beta) \sum_{\vp \in \supp(\gvec^{t(\vc)})}  \gamma^{t(\vc)}_{\vp} \left[ \sum_{\omega \in \Omega} F_{g( \cvec,\vp),\omega}  \left( r_\omega - p_\omega  \right) \right]-\frac \eta \beta\\
		&\ge \sum_{\vp \in \supp( \gvec^{t(\vc)})}  \gamma^{t(\vc)}_{\vp} \left[ \sum_{\omega \in \Omega} F_{g( \cvec, \vp),\omega}  \left( r_\omega - p_\omega  \right) \right]-\beta -\frac \eta \beta,
	\end{align*}
	where the first equality holds by definition of $\widetilde{\Gamma}$, the second inequality is obtained by rearranging terms, the subsequent (first) inequality holds thanks to the relation proved above, while the last inequality holds since rewards are upper bounded by one and payments are lower bounded by zero.
	
	By recalling the definition of $\Up(\widetilde{\Gamma},\vc)$ and how we set the index $i$ at the beginning of the proof, the inequality above clearly gives the result, concluding the proof.
\end{proof}

\theoremTemplateOne*

\begin{proof}
	By \cref{def:apxAlgo}, we have that the contract $\Gamma \in \cX$ computed by means of the the procedure \texttt{Compute-RM-Contract}$(\cX,\cC_\epsilon)$ satisfies the following: 
	\begin{align}\label{eq:meta1}
		\sup_{\vc \in \cC_\epsilon} \Big\{ \OPT(\vc)-\Up( \Gamma,\vc) \Big\} &  \le   \inf_{\Gamma' \in \cX}  \sup_{\vc \in \cC_\epsilon} \Big\{ \OPT(\vc)-\Up(\Gamma',\vc) \Big\}  .
	\end{align}
	
	By \cref{lm:existsEpsIC}, there exist two functions $t$ and $g$ as in Definition~\ref{def:apxIC} such that $(\Gamma,t,g)$ is $(2\epsilon)$-incentive compatible and the following holds:
	\begin{align}\label{eq:meta2}
		\sup_{\vc \in \cC} \left\{ \OPT(\vc)- \sum_{\vp \in \supp(\gvec^{t(\vc)})} \gamma^{t(c)}_{\vp} \left[\sum_{\omega \in \Omega} F_{g(\vc,\vp),\omega} \, p_\omega - c_{g(\vc,\vp)} \right]\right\}=  \sup_{\vc \in \cC_\epsilon} \Big\{ \OPT(\vc)-\Up( \Gamma,\vc) \Big\} .   
	\end{align}
	Moreover, by \cref{lm:linearized}, the contract $\widetilde{\Gamma} \in \cX$ produced by $\texttt{Linearize}(\Gamma,\sqrt{2 \epsilon})$ satisfies:
	\begin{align}\label{eq:meta3}
		\sup_{\vc \in \cC} \Big\{  \OPT(\vc)-\Up(\widetilde \Gamma,\vc) \Big\}&\le \sup_{\vc \in \cC} \left\{ \OPT(\vc)-\sum_{\vp \in \supp(\gvec^{t(\vc)})} \hspace{-2mm}\gamma^{t(\vc)}_{\vp}   \left[\sum_{\omega \in \Omega} F_{g(\vc,\vp),\omega} (r_\omega - p_\omega) \right]+2\sqrt{2\epsilon} \right\}.
	\end{align}
	Finally, we have that:
	\begin{align*}
		\sup_{\vc \in \cC} \Big\{ \OPT(\vc)-\Up(\widetilde \Gamma,\vc) \Big\}&\le \sup_{\vc \in \cC} \left\{  \OPT(\vc)-\sum_{\vp \in \supp(\gvec^{t(\vc)})} \gamma^{t(\vc)}_{\vp}   \left[\sum_{\omega \in \Omega} F_{g(\vc,\vp),\omega} (r_\omega - p_\omega) \right]+2\sqrt{2\epsilon}   \right\}\\
		&=\sup_{\vc \in \cC} \left\{ \OPT(\vc)-\sum_{\vp \in \supp(\gvec^{t(\vc)})} \gamma^{t(\vc)}_{\vp}   \left[\sum_{\omega \in \Omega} F_{g(\vc,\vp),\omega} (r_\omega - p_\omega) \right]    \right\}+2\sqrt{2\epsilon}\\
		& = \sup_{\vc \in \cC_\epsilon} \Big\{ \OPT(\vc)-\Up(\Gamma,\vc) \Big\} + 2\sqrt{2\epsilon}\\
		&\le \inf_{\Gamma' \in \cX}  \sup_{\vc \in \cC_\epsilon} \Big\{ \OPT(\vc)-\Up(\Gamma',\vc) \Big\}  +2\sqrt{2\epsilon} \\
		&\le   \inf_{\Gamma' \in X}  \sup_{\vc \in \cC} \Big\{  \OPT(\vc)-\Up(\Gamma',\vc) \Big\} +2\sqrt{2\epsilon},
	\end{align*}
	where the first inequality comes from \cref{eq:meta3}, the second equality from \cref{eq:meta2}, the second inequality from \cref{eq:meta1}, while the last inequality follows from $\cC_\epsilon\subseteq\cC$.
\end{proof}

\section{Proofs Omitted from Section~\ref{sec:beyond}}\label{sec:app_template}

\lemmaCoverOne*

\begin{proof}
	Let $(\Omega,A,\rvec,\mcC,\Fvec)$ be a single-dimensional PAPU instance characterized by a ``central'' cost vector $\vc^\circ \in \mcC$ and a range parameter $\rho \in \mathbb{R}_{>0}$.
	Let $\mcC_\epsilon \subseteq \mcC$ be defined as follows:
	\[
	\mcC_\epsilon := \left\{ \big(1- \rho + 2 \rho \, q \big)\vc^\circ \mid q \in \mathcal{Q} \right\} \quad \text{where} \quad \mathcal{Q} := \left\{  0, \frac 1 {\lceil \nicefrac {2 \rho} \epsilon \rceil}, \frac 2 {\lceil \nicefrac {2 \rho} \epsilon \rceil}, \ldots, \frac {\lceil \nicefrac {2 \rho} \epsilon \rceil - 1} {\lceil \nicefrac {2 \rho} \epsilon \rceil}, 1 \right\} .
	\]
	First, we prove that, given any cost vector $\vc \in \mcC$, there exists a cost vector $\tilde \vc \in \mcC_\epsilon$ with $\| \vc - \tilde \vc \|_{\infty} \leq \epsilon$.
	Indeed, since $\vc = \lambda \vc^\circ$ for some scaling parameter $\lambda \in [1-\rho, 1+\rho]$, by letting $\tilde \vc := (1- \rho + 2 \rho \, q^\star ) \vc^\circ$ where $q^\star \in \arg\min_{q \in \mathcal{Q}} | (1- \rho + 2 \rho \, q ) - \lambda |$, it holds that:
	\begin{align*}
	\| \vc - \tilde \vc \|_{\infty} & = \| \lambda \vc^\circ - (1- \rho + 2 \rho \, q^\star ) \vc^\circ \|_{\infty} \\
	& \leq | \lambda - (1- \rho + 2 \rho \, q^\star ) | \\
	& = 2\rho  \left| \frac {\lambda - (1-\rho)} {2\rho} - q^\star \right| \\
	& \leq 2 \rho \frac 1 {\lceil \nicefrac {2 \rho} \epsilon \rceil} \leq \epsilon,
	\end{align*}
	where the second to last inequality holds by definition of $q^\star$.
	This shows that $\mcC_\epsilon$ is an $\epsilon$-cover.
	
	Moreover, the size of the set $\mcC_\epsilon$ is clearly $\lceil \nicefrac {2 \rho} \epsilon \rceil + 1$, which concludes the proof.
\end{proof}

\lemmaCoverTwo*

\begin{proof}
		The proof trivially follows by noticing that any $L_p$-norm ball centered in $\vc_0\in [0,1]^n$ and with radius $D \in [0,1]$ is included in the $L_\infty$-norm ball centered in $\vc_0$ with radius $D$, and that the latter admits an $\epsilon$-cover $\mcC_\epsilon$ of size $O((\nicefrac1\epsilon)^{n})$.
		This is obtained by considering a grid over the hypercube $[0,1]^n$ using $\lceil \nicefrac 1 \epsilon \rceil + 1$ points along each dimension.
\end{proof}

\theoremEfficiently*

\begin{proof}
	The algorithm proving the theorem adopts an approach similar to one originally introduced by~\citet{gan2022optimal} in the context of \emph{generalized} principal-agent problems, and later extended to \emph{multi-agent Bayesian} principal-agent problems by~\citet{castiglioni23multi}.
	Even if the approaches are similar, our algorithm employs techniques different form those in~\citep{gan2022optimal,castiglioni23multi}.
	Our techniques can only be applied in our setting since, as we show next, Problem~\eqref{pr:general} admits an optimal solution (a minimum), while the problems tacked by~\citet{gan2022optimal}~and~\citet{castiglioni23multi} may \emph{not} admit an optimal solution (a maximum for them).
	
	The algorithm works by solving a \emph{linear programming} relaxation of Problem~\eqref{pr:general}, defined as:
	%
	%
	\begin{subequations}\label{pr:relaxed}
		\begin{align}
		\min_{\gamma, z, u, v} \,\,& \quad u \quad\quad \textnormal{s.t.} \hspace{10cm}\,\\
		& \specialcell{u \geq \OPT(\vc)- \sum_{j \in [W]}  \sum_{\omega \in \omega} F_{g(\vc,j),\omega} \left( \gamma^{t(\vc)}_j r_{\omega}-z^{t(\vc)}_{j,\omega} \right) \hfill \forall \vc \in \widetilde{ \mcC}} \label{cons_obj}\\
		& \specialcell{\sum_{j \in [W]}\sum_{\omega \in \Omega}  F_{g(\vc,j), \omega} \left( z^{t(\vc)}_{j,\omega} - \gamma^{t(\vc)}_j c_{g(\vc,j)} \right)\ge \sum_{j \in [W]} v^{\vc, i}_{j} \hfill \forall \vc \in \widetilde{ \mcC},\forall i \in [K]} \label{cons_lin_first}\\
		& \specialcell{v^{\vc, i}_{j} \geq \sum_{\omega \in \Omega}  F_{a,\omega} \left(  z^{i}_{j,\omega} - \gamma^{i}_j   c_{a} \right) \hfill\forall \vc \in \widetilde{ \mcC},  \forall i \in [K], \forall j \in [W], \forall a \in A} \label{cons_lin}\\
		& \specialcell{\sum_{j \in [W]}\gamma^i_j=1 \hfill \forall i \in [K] } \\
		& \specialcell{\gamma^i_j \geq 0 \hfill \forall i \in [K], \forall j \in [W] } \\
		& \specialcell{z^i_{j,\omega} \geq 0 \hfill \forall i \in [K], \forall j \in [W], \forall \omega \in \Omega. }
		\end{align}
	\end{subequations}
	Problem~\eqref{pr:relaxed} is derived from Problem~\eqref{pr:general} by substituting each product $\gamma^i_{j} \, p_{j,\omega}^i$ with a variable~$z_{j,\omega}^i$.
	%
	Additionally, the $\max$ over $\widetilde{ \mcC}$ in the objective of Problem~\eqref{pr:general} is encoded by the set of linear Constraints~\eqref{cons_obj}, by introducing an auxiliary variable $u$, while the $\max$ in the right-hand side of Constraints~\eqref{cons} is encoded by the set of linear  Constraints~\eqref{cons_lin}, by introducing the auxiliary variables $v^{\vc, i}_{j}$.
	Notice that Problem~\eqref{pr:relaxed} is a linear program since $\widetilde{\cC}$ is finite, and it can be solved in time polynomial in the size of the PAPU instance, the size of $\widetilde{\cC}$, $K$, and $W$.

	The crucial result that is exploited by the algorithm is that, given an optimal solution to the linear programming relaxation in Problem~\eqref{pr:relaxed}, it is possible to efficiently recover an optimal solution to Problem~\eqref{pr:general}.
	%
	We split the proof of this result into three parts, which are addressed in the following.
	Notice that, as a byproduct, we also get that Problem~\eqref{pr:general} always admits an optimal solution (and, thus, the $\inf$ can be safely replaced with a $\min$ in its definition).
	
	%
	%
	%
	
	\subsection*{Part I: \textnormal{\emph{The value of Problem~\eqref{pr:relaxed} is at least the value of Problem~\eqref{pr:general}}.}}
	
	We show that, given any feasible solution to Problem~\eqref{pr:general}, it is possible to recover a feasible solution to Problem~\eqref{pr:relaxed} such that the objective value of the latter is equal to the objective value of the former.
	
	Let $(\gamma, p)$ be a feasible solution to Problem~\eqref{pr:general}.
	Then, we define a solution $(\gamma,z,u,v)$ to Problem~\eqref{pr:relaxed} by letting $z_{j,\omega}^i = \gamma^i_j \, p_{j,\omega}^i$ for every $i \in [K]$, $j \in [W]$, and $\omega \in \Omega$.
	Moreover, we let $u$ be equal to the objective function value of the feasible solution to Problem~\eqref{pr:general} (\emph{i.e.}, the value of the $\max_{\vc \in \widetilde{\cC}}$ operator in Objective~\eqref{obj}), while we let all the variables $v^{\vc , i}_j$, for every $\vc \in \widetilde{\cC}$, $i \in [K]$, and $j \in [W]$, be equal to the values of their corresponding $\max_{a \in A}$ operators, appearing inside the summations in the right-hand sides of Constraints~\eqref{cons} in Problem~\eqref{pr:general}.

	It is immediate to see that all the constraints in Problem~\eqref{pr:relaxed} except for Constraints~\eqref{cons_obj} are satisfied, given that their counterparts in Problem~\eqref{pr:general} are satisfied as well.
	Moreover, by definition of $u$, Constraints~\eqref{cons_obj} are clearly satisfied, and, since the objective minimized by Problem~\eqref{pr:relaxed} is $u$, the solution $(\gamma,z,u,v)$ built above achieves the same objective value as the given solution $(\gamma, p)$.

	\subsection*{Part II: \textnormal{\emph{Given a solution to Problem~\eqref{pr:relaxed}, it is possible to efficiently recover a solution to Problem~\eqref{pr:relaxed} having the same objective function value and such that, if $z^{i}_{j,\omega} > 0$, then $\gamma_{j}^i > 0$}.}}
	
	Let $(\gamma,z,u,v)$ be a solution to Problem~\eqref{pr:relaxed}.
	For ease of notation, let us define
	\[
		M := \Big\{ (i,j) \mid i \in [K] \wedge j \in [W] \wedge \gamma^i_j = 0 \wedge \exists \omega \in \Omega : z^i_{j,\omega} > 0 \Big\}
	\]
	as the set of all the pairs of indexes $(i,j) \in [K] \times [W]$ such that the desired condition is \emph{not} satisfied for at least one $\omega \in \Omega$.
	Clearly, such a set can be identified in time polynomial in $|\Omega|$, $K$, and $W$.
	
	We build a new solution $(\gamma, \bar z, u, \bar v)$ to Problem~\eqref{pr:relaxed} as described in the following.
	Notice that such a solution can be built in time polynomial in the size of the PAPU instance, the size of $\widetilde{\cC}$, $K$, and $W$.
	
	For every pair of indexes $(i,j) \in M$, we set $\bar z^i_{j,\omega} := 0$ for all $\omega \in \Omega$.
	Moreover, for every $i \in [K]$, we let $j^\diamond(i) \in [W]$ be any index such that $(i,j^\diamond(i)) \notin M$.
	Notice that such an index always exists, otherwise the variables $\gamma^i_j$ would \emph{not} encode a probability distribution.
	Then, for every $i \in [K]$ and $\omega \in \Omega$, we set 
	$$
		\bar z^i_{j^\diamond(i), \omega} := z^i_{j^\diamond(i),\omega} + \sum_{j \in [W] : (i,j) \in M} \max_{a \in A} \sum_{\omega' \in \Omega} F_{a,\omega'} z^i_{j,\omega'},
	$$ 
	while, for every $(i,j) \notin M$ with $(i,j) \neq (i,j^\diamond(i))$, we let $\bar z^i_{j,\omega} := z^i_{j,\omega}$ for all $\omega \in \Omega$.
	Intuitively, the new variables are defined in such a way that positive values are ``moved'' from pairs $(i,j) \in M$ to a single pair $(i,j^\diamond(i)) \notin M$, in order to satisfy the desired condition.
	Indeed, it immediate to check that, if $\bar z^i_{j,\omega}>0$ for some $i \in [K]$, $j \in [W]$, and $\omega \in \Omega$, then it is the case that $\gamma^i_{j}> 0$.
	
	Additionally, for every $\vc \in \widetilde{\cC}$, $i \in [K]$, and $j \in [W]$, we set
	\[
		\bar v^{\vc,i}_j := v^{\vc,i}_{j} \cdot \mathbf{1} \{ (i,j) \notin M \} + \left( \sum_{j' \in [W]: (i,j') \in M} v^{\vc,i}_{j'} \right) \cdot \mathbf{1} \{ (i,j) = (i, j^\diamond(i)) \}.
	\]
	Notice that, by definition, it holds that $\bar v_{j}^{\vc,i} = v^{\vc,i}_j$ whenever $(i,j) \notin M$ and $(i,j) \neq (i, j^\diamond(i)) $, while $\bar v_{j}^{\vc,i} = 0$ whenever $(i,j) \in M$.
	Moreover, $\sum_{j \in [W]} \bar v_{j}^{\vc,i} = \sum_{j \in [W]}  v_{j}^{\vc,i} $ for all $\vc \in \widetilde{\cC}$ and $i \in [K]$.

	First, we prove the following crucial result:
	\[
		\sum_{j \in [W]}  \sum_{\omega \in \Omega} F_{g(\vc,j),\omega} z^{t(\vc)}_{j,\omega}  = \sum_{\substack{j \in [W] : \\(t(\vc),j) \notin M}}  \sum_{\omega \in \Omega} F_{g(\vc,j),\omega} \bar z^{t(\vc)}_{j,\omega} \quad \forall \vc \in \widetilde{\cC}.
	\]
	Indeed, we can prove that, for every $\vc \in \widetilde{\cC}$, it holds
	\begin{align*}
		\sum_{j \in [W]}  \sum_{\omega \in \Omega} F_{g(\vc,j),\omega} z^{t(\vc)}_{j,\omega}  & =
		\sum_{\substack{j \in [W]: \\(t(\vc),j) \in M}}  \sum_{\omega \in \Omega} F_{g(\vc,j),\omega} z^{t(\vc)}_{j,\omega} + \sum_{\substack{j \in [W] : \\(t(\vc),j) \notin M}}  \sum_{\omega \in \Omega} F_{g(\vc,j),\omega} z^{t(\vc)}_{j,\omega}  \\
		& = \sum_{\substack{j \in [W]:\\ (t(\vc),j) \in M}}  \sum_{\omega \in \Omega} F_{g(\vc,j),\omega} z^{t(\vc)}_{j,\omega} +\sum_{\substack{j \in [W] :\\ (t(\vc),j) \notin M}}  \sum_{\omega \in \Omega} F_{g(\vc,j),\omega} \bar z^{t(\vc)}_{j,\omega} \\
		& \quad\quad  - \sum_{\omega \in \Omega} F_{g(\vc,j^\diamond(t(\vc))), \omega} \sum_{\substack{j \in [W] : \\(t(\vc),j) \in M}} \max_{a \in A} \sum_{\omega' \in \Omega} F_{a,\omega'} z^{t(\vc)}_{j,\omega'}\\
		& = \sum_{\substack{j \in [W]: \\(t(\vc),j) \in M}}  \max_{a \in A}\sum_{\omega \in \Omega} F_{a,\omega} z^{t(\vc)}_{j,\omega} +\sum_{\substack{j \in [W] :\\ (t(\vc),j) \notin M}}  \sum_{\omega \in \Omega} F_{g(\vc,j),\omega} \bar z^{t(\vc)}_{j,\omega} \\
		& \quad\quad  - \sum_{\omega \in \Omega} F_{g(\vc,j^\diamond(t(\vc))), \omega} \sum_{\substack{j \in [W] :\\ (t(\vc),j) \in M}} \max_{a \in A} \sum_{\omega' \in \Omega} F_{a,\omega'} z^{t(\vc)}_{j,\omega'}\\
		& = \sum_{\substack{j \in [W]: \\(t(\vc),j) \in M}}  \max_{a \in A}\sum_{\omega \in \Omega} F_{a,\omega} z^{t(\vc)}_{j,\omega} +\sum_{\substack{j \in [W] : \\(t(\vc),j) \notin M}}  \sum_{\omega \in \Omega} F_{g(\vc,j),\omega} \bar z^{t(\vc)}_{j,\omega} \\
		& \quad\quad  - \sum_{\substack{j \in [W] :\\ (t(\vc),j) \in M}} \max_{a \in A} \sum_{\omega' \in \Omega} F_{a,\omega} z^{t(\vc)}_{j,\omega}\\
		& = \sum_{\substack{j \in [W] : \\(t(\vc),j) \notin M}}  \sum_{\omega \in \Omega} F_{g(\vc,j),\omega} \bar z^{t(\vc)}_{j,\omega} ,
	\end{align*}
	where the second equality hods by definition of $z^{t(\vc)}_{j,\omega}$ for all $j \in [W]$ with $ (t(\vc),j) \notin M$, while the third equality holds since, if $(t(\vc),j) \in M$, then $g(\vc,j) \in \arg\max_{a \in A} \sum_{\omega \in \Omega} F_{a,\omega} z_{j,\omega}^{t(\vc)}$.
	This latter fact can be easily proved by using the fact that the original solution satisfies the constraints of Problem~\eqref{pr:relaxed}, and in particular those in which $i = t(\vc)$. Intuitively, $g(\vc,j)$ is a ‘‘best response'' for the agent to the ‘‘contract'' with payments $z^{t(\vc)}_{j, \omega}$, when action costs are neglected.
	
	Now, we prove that the new solution $(\gamma, \bar z, u, \bar v)$ satisfies all the constraints of Problem~\eqref{pr:relaxed}.
	First, we prove that, for every $\vc \in \widetilde{\cC}$, it holds:
	\begin{align*}
		u & \geq \OPT(\vc)- \sum_{j \in [W]}  \sum_{\omega \in \Omega} F_{g(\vc,j),\omega} \left( \gamma^{t(\vc)}_j r_{\omega}-z^{t(\vc)}_{j,\omega} \right) \\
		& = \sum_{j \in [W]}  \sum_{\omega \in \Omega} F_{g(\vc,j),\omega} \left( \gamma^{t(\vc)}_j r_{\omega}-\bar z^{t(\vc)}_{j,\omega} \right),
	\end{align*}
	where the equality holds by our first crucial result.
	This shows that Constraints~\eqref{cons_obj} are satisfied.
	Moreover, for every $\vc \in \widetilde{\cC}$ and $i \in [K]$, it holds:
	\begin{align*}
		\sum_{j \in [W]}\sum_{\omega \in \Omega}  F_{g(\vc,j), \omega} \left( \bar z^{t(\vc)}_{j,\omega} - \gamma^{t(\vc)}_j c_{g(\vc,j)} \right) & =  \sum_{j \in [W]}\sum_{\omega \in \Omega}  F_{g(\vc,j), \omega} \left(  z^{t(\vc)}_{j,\omega} - \gamma^{t(\vc)}_j c_{g(\vc,j)} \right)\\
		& \geq  \sum_{j \in [W]}  v^{\vc, i}_{j} \\
		& = \sum_{j \in [W]} \bar v^{\vc, i}_{j},
	\end{align*}
	where the last equality holds by definition of $\bar v_{j}^{\vc,i}$.
	This shows that Constraints~\eqref{cons_lin_first} are satisfied.

	Finally, we look at Constraints~\eqref{cons_lin}.
	We consider three cases.
	The first case is when $(i,j) \notin M$ and $(i,j) \neq (i,j^\diamond(i))$, in which the following holds for every $\vc \in \widetilde{\cC}$ and $a \in A$:
	\begin{align*}
		\bar v^{\vc, i}_{j} =  v^{\vc, i}_{j} \geq \sum_{\omega \in \Omega}  F_{a,\omega} \left(  z^{i}_{j,\omega} - \gamma^{i}_j   c_{a} \right) = \sum_{\omega \in \Omega}  F_{a,\omega} \left(  \bar z^{i}_{j,\omega} - \gamma^{i}_j   c_{a} \right) ,
	\end{align*}
	where the first and the last equalities hold by definition of $\bar v^{\vc, i}_{j}$ and $\bar z^{i}_{j,\omega} $, respectively, since $\bar v^{\vc, i}_{j} =  v^{\vc, i}_{j}$ and $\bar z^{i}_{j,\omega} =  z^{i}_{j,\omega} $ whenever $(i,j) \notin M$ and $(i,j) \neq (i,j^\diamond(i))$.
	
	The second case is when $(i,j) \in M$, where, for every $\vc \in \widetilde{\cC}$ and $a \in A$, the relation
	\begin{align*}
		\bar v^{\vc, i}_{j} \geq \sum_{\omega \in \Omega}  F_{a,\omega} \left(  \bar z^{i}_{j,\omega} - \gamma^{i}_j   c_{a} \right) 
	\end{align*}
	trivially holds since $\bar v^{\vc, i}_{j} = 0$, $\bar z^{i}_{j,\omega} = 0$ for all $\omega \in \Omega$, and $\gamma^i_j = 0$.
	
	The third case is when $(i,j) = (i,j^\diamond(i))$, where, for every $\vc \in \widetilde{\cC}$ and $a \in A$, it holds:
	\begin{align*}
		\bar v^{\vc, i}_{j} & = v^{\vc,i}_j + \sum_{j' \in [W]: (i,j') \in M} v^{\vc,i}_{j'}  \\
		& \geq \sum_{\omega \in \Omega}  F_{a,\omega} \left(  z^{i}_{j,\omega} - \gamma^{i}_j   c_{a} \right) + \sum_{j' \in [W]: (i,j') \in M}\max_{a' \in A} \sum_{\omega \in \Omega}   F_{a',\omega}   z^{i}_{j',\omega}  \\
		& = \sum_{\omega \in \Omega}  F_{a,\omega} \left(  \bar z^{i}_{j,\omega} - \gamma^{i}_j   c_{a} \right)  - \sum_{\omega \in \Omega}  F_{a,\omega} \left(  \sum_{j' \in [W] : (i,j') \in M} \max_{a' \in A} \sum_{\omega' \in \Omega} F_{a',\omega'} z^i_{j',\omega'}\right) \\
		& \quad + \sum_{j' \in [W]: (i,j') \in M}\max_{a' \in A} \sum_{\omega \in \Omega}   F_{a',\omega}   z^{i}_{j',\omega} \\
		& = \sum_{\omega \in \Omega}  F_{a,\omega} \left(  \bar z^{i}_{j,\omega} - \gamma^{i}_j   c_{a} \right) ,
	\end{align*}
	where the first inequality holds by the fact that the original solution satisfies the constraints, while the second equality follows from the definition of $\bar z^i_{j,\omega}$.
	This shows that also the last missing set of constraints is satisfied, proving that $(\gamma, \bar z, u, \bar v)$ is feasible for Problem~\eqref{pr:relaxed}.
	
	In conclusion, it is immediate to see that the two solutions achieve the same objective value $u$.

	\subsection*{Part III: \textnormal{\emph{Given an optimal solution to the linear programming relaxation in Problem~\eqref{pr:relaxed}, it is possible to efficiently recover an optimal solution to Problem~\eqref{pr:general}.}}}

	By the second part of the proof, given any optimal solution $(\gamma,z,u,v)$ to Problem~\eqref{pr:relaxed}, we can recover another solution $(\gamma, \bar z, u, \bar v)$ having the same objective function value and such that, if $z_{j,\omega}^i > 0$ for some $i \in [K]$, $j \in [W]$, and $\omega \in \Omega$, then $\gamma_i^j > 0$.
	Moreover, this recovering step be done in time polynomial in the size of the PAPU instance, the size of $\widetilde{\cC}$, $K$, and $W$.

	From $(\gamma, \bar z, u, \bar v)$, we can recover a solution $(\gamma,p)$ to Problem~\eqref{pr:general} by setting $$p_{j,\omega}^i = \frac{z_{j,\omega}^i}{\gamma^i_j} \quad \text{whenever} \quad \gamma_j^i > 0,$$ while $p_{j,\omega}^i$ can be set to any non-negative value otherwise.
	It is immediate to check that the obtained solution is feasible for Problem~\eqref{pr:general} and achieves objective function value equal to the one achieved by $(\gamma, \bar z, u, \bar v)$ in Problem~\eqref{pr:relaxed}.
	The solution $(\gamma,p)$ can be clearly obtained in polynomial time.
	
	Since in the first part of the proof we have shown that the value of Problem~\eqref{pr:relaxed} is at least the value of Problem~\eqref{pr:general} and $(\gamma, \bar z, u, \bar v)$ is optimal for Problem~\eqref{pr:relaxed}, we have that the solution $(\gamma,p)$ built above is optimal for Problem~\eqref{pr:general}.
	This concludes the proof.
	%
	%
	%
	%
	%
	%
\end{proof}

\thmRandomizedMenus*

\begin{proof}
By classical revelation-principle-style arguments~\citep{gan2022optimal}, there exists a regret-minimizing menu of randomized contracts $\Gamma = \{ \gvec^1, \ldots \gvec^K \} \in \mcMr$ of size $K = |\widetilde{\cC}|$ such that:
\begin{enumerate}
	\item $\Gamma$ is \emph{direct}, \emph{i.e.}, its elements correspond one-to-one to cost vectors in $\widetilde{\mcC} := \left\{ \vc^1, \ldots, \vc^K \right\}$.
	\item $\Gamma$ is \emph{truthful}, \emph{i.e.}, an agent with cost vector $\vc^i$ is incentivized to select the $i$-th element of the menu, namely, the randomized contract $\gvec^i$.
	\item $\Gamma$ is \emph{succinct}, \emph{i.e.}, every $\gvec^i$ has at most $W = |A|$ deterministic contracts in its support, and each of them induces a different best response for an agent with cost vector $\vc^i$.
\end{enumerate}
%

As a result, we can apply the algorithm introduced in Theorem~\ref{thm:computeGeneral} with $K = |\widetilde{C}|$, $W = |A|$, and functions $t$ and $g$ defined as follow.
In particular, $t$ is such that each cost vector $\vc^i \in \widetilde{\cC}$ is assigned to the index $i$.
Moreover, by letting $A := \{ a_1, \ldots, a_n \}$, the function $g$ is such that $g(\vc^i, j) = a_j$. 
%
\end{proof}

\thmOtherCntracts*

\begin{proof}
	%
	For each class among $\mcP$, $\mcPr$, and $\mcM$, we exploit the following arguments:
	\begin{enumerate}
		\item[(i)] It is possible to efficiently enumerate over all the possible functions $t$ and $g$ when the size of the set $\widetilde{\cC}$ is constant.
		\item[(ii)] Suitable functions $t$ and $g$ can be identified \emph{a priori}. 
	\end{enumerate}
	%
	In both cases, a direct application of Theorem~\ref{thm:computeGeneral} gives the result.
	
	\paragraph{Case $\cX = \mcP$}
	For the class of deterministic contracts, there is only one possible function $t$, which is one that maps each cost vector to the index $1$ (the only available contract).
	Moreover, the function $g$ can be specified by using the best responses of all agents with cost vectors in $\widetilde{\cC}$ to the only available contract.
	Thus, there are $|A|^{|\widetilde{\cC}|}$ possible functions $g$, which can be enumerated efficiently when the size of the set $\widetilde{\cC}$ is constant.

	\paragraph{Case $\cX = \mcPr$}
	For the class of randomized contracts, by a revelation-principle-style argument, we know that there exists a regret-minimizing randomized contract that does \emph{not} include in the support two contracts that induces the same tuple of a best response per agent's cost vector. Hence, as for deterministic contracts, we only have one possible function $t$, which is one that maps each cost vector to the index $1$ (the only available contract). Moreover, we let the support of the randomized contract to be of size $|A|^{|\widetilde{\cC}|}$, and assign to each contract in the support a different action profile through $g$.

	\paragraph{Case $\cX = \mcM$}
	For the class of menus of deterministic contracts, by a revelation-principle-style argument, we know that there exists a regret-minimizing direct menu of deterministic contracts that includes a contract for each cost vector  in $\widetilde{\cC}$. Moreover, the function $t$ is such that each contract is assign to a different cost vector.
	Hence, we are left to define the function $g$. There are $|A|^{|\widetilde{\cC}|}$ possible functions that specify a best response for each cost vector.

	This concludes the proof.
\end{proof}

\end{document}